\documentclass[reqno,11pt,letterpaper]{amsart}
\usepackage{amsmath,amssymb,amsthm,graphicx,mathrsfs,url}
\usepackage[usenames,dvipsnames]{color}
\usepackage[colorlinks=true,linkcolor=Red,citecolor=Green,urlcolor=Blue]{hyperref}
\usepackage[super]{nth}
\usepackage[open, openlevel=2, depth=3, atend]{bookmark}
\hypersetup{pdfstartview=XYZ}

\def\?[#1]{\textbf{[#1]}\marginpar{\Large{\textbf{??}}}}

\numberwithin{equation}{section}

\newtheorem{theorem}{Theorem}[section]
\newtheorem{lemma}[theorem]{Lemma}
\newtheorem{proposition}[theorem]{Proposition}
\newtheorem{corollary}[theorem]{Corollary}

\newtheorem{remark}[theorem]{Remark}

\newcounter{conj}
\newtheorem{conjecture}[conj]{Conjecture}

\newcommand{\mc}{\mathcal}
\newcommand{\rr}{\mathbb{R}}
\newcommand{\nn}{\mathbb{N}}
\newcommand{\cc}{\mathbb{C}}
\newcommand{\hh}{\mathbb{H}}
\newcommand{\zz}{\mathbb{Z}}

\newcommand{\la}{\lambda}
\newcommand{\eps}{\epsilon}

\newcommand{\pl}{\partial}
\newcommand{\x}{\times}

\newcommand{\til}{\widetilde}
\newcommand{\bbar}{\overline}

\newcommand{\cjd}{\rangle}
\newcommand{\cjg}{\langle}

\newcommand{\demi}{\tfrac{1}{2}}

\DeclareMathOperator{\supp}{supp}

\def\indic{\operatorname{1\hskip-2.75pt\relax l}}

\renewcommand{\tilde}{\widetilde}          
\DeclareMathSymbol{\leqslant}{\mathalpha}{AMSa}{"36} 
\DeclareMathSymbol{\geqslant}{\mathalpha}{AMSa}{"3E} 
\DeclareMathSymbol{\eset}{\mathalpha}{AMSb}{"3F}     
\renewcommand{\leq}{\;\leqslant\;}                   
\renewcommand{\geq}{\;\geqslant\;}                   

\newcommand{\C}{\mathbb{C}}

\newcommand{\R}{\mathbb{R}}
\newcommand{\Z}{\mathbb{Z}}
\renewcommand{\H}{\mathbb{H}}
\newcommand{\N}{\mathbb{N}}

\newcommand{\E}{\mathbb{E}}
\renewcommand{\P}{\mathbb{P}}

\def\eps{\varepsilon}

\def\bi{\begin{itemize}}
\def\ei{\end{itemize}}

\def\bnum{\begin{enumerate}}
\def\enum{\end{enumerate}}
\def\<#1{\langle #1 \rangle}

\title[LQG and bosonic $2d$ string theory]{Polyakov's formulation of $2d$ bosonic string theory}

\author{Colin Guillarmou}
\address{CNRS, Universit\'e Paris-Sud, D\'epartement de Math\'ematiques, 91400
Orsay, France}
\email{cguillar@math.cnrs.fr}
\author{R\'emi Rhodes}
\address{Aix Marseille Universit\'e, CNRS, Centrale Marseille, I2M, Marseille, France.}
\email{remi.rhodes@univ-amu.fr}
\author{Vincent Vargas}
\address{DMA, U.M.R. 8553 CNRS, \'Ecole Normale Superieure, 45 rue d'Ulm,
75230 Paris cedex 05, France.}
 \email{Vincent.Vargas@ens.fr}

\begin{document}

\maketitle

\begin{abstract}
Using probabilistic methods, we first define Liouville quantum field theory on Riemann surfaces of genus ${\bf g}\geq 2$ and show that it is a conformal field theory. We use the partition function of Liouville quantum field theory to give a mathematical sense to Polyakov's partition function 
of  \textit{noncritical bosonic string theory}  \cite{Pol} (also called $2d$  \textit{bosonic string theory}) and to Liouville quantum gravity. More specifically, we show the convergence of Polyakov's partition function over the moduli space of Riemann surfaces in genus  ${\bf g}\geq 2$ in the case of $D\leq 1$ boson. This is done by performing a careful analysis of the behavior of the partition function at the boundary of moduli space. An essential feature of our approach is that it is probabilistic and \textit{non perturbative}. The interest of our result is twofold. First, to the best of our knowledge, this is the first mathematical result about convergence of string theories. Second, our construction describes conjecturally  the scaling limit of higher genus random planar maps weighted by Conformal Field Theories: we make precise conjectures about this statement at the end of the paper.
\end{abstract}

\footnotesize



\normalsize




\section{Introduction}

In physics,  string theory   or more generally Euclidean $2d$ Quantum Gravity (LQG)  is an attempt to quantize the Einstein-Hilbert functional coupled to matter fields (matter is replaced by the free bosonic string in the case of string theory).  The problem can be briefly summarized as follows. 

First of all, a \emph{quantum field theory} on a surface $M$ can be viewed as a 
way to define a measure $e^{-S_g(\phi)}D\phi $ over an infinite dimensional space $E$ of fields $\phi$ living over $M$ (typically $\phi$ are sections of some bundles over $M$), where $D \phi$ is a ``uniform measure" and $S_g: E\to \mathbb{R}$ is a functional on $E$ called  the \emph{action}, depending on a background Riemannian metric $g$ on $M$. The total mass of the measure 
\begin{equation}\label{Z(g)CFT}
Z(g):=\int_{E} e^{-S_g(\phi)}D\phi 
\end{equation}
is called the \emph{partition function}. Defining the \emph{n-point correlation functions} amounts to taking $n$ points $x_1,\dots,x_n\in M$ and weights $\alpha_1,\dots,\alpha_n\in\rr$ and to defining 
\[Z(g; (x_1,\alpha_1),\dots,(x_n,\alpha_n)):=\int_{E} e^{\sum_{i=1}^n\alpha_i \phi(x_i)}e^{-S_g(\phi)}D\phi, \]
at least if the fields $\phi$ are functions on $M$.

A \emph{conformal field theory} (CFT in short) on a surface is a quantum field theory which possesses certain conformal symmetries. More specifically, the partition function 
$Z(g)$ of a CFT  satisfies the diffeomorphism invariance $Z(\psi^*g)=Z(g)$ for all smooth diffeomorphisms $\psi:M\to M$ and  a so-called \emph{conformal anomaly} of the following form: for all $\omega\in C^\infty(M)$
\begin{equation}\label{CFTanomaly}
Z(e^{\omega}g)=Z(g)\exp\Big(\frac{{\bf c}}{96\pi}\int_{M}(|d\omega|_g^2+2K_g\omega) {\rm dv}_g\Big)
\end{equation}
where   ${\bf c}\in\mathbb{R}$ is called the \emph{central charge} of the theory, $K_g$ is the scalar curvature of $g$ and $ {\rm  dv}_g$ the volume form. The $n$-point correlation functions should also satisfy similar types of conformal anomalies and diffeomorphism invariance (see \eqref{diffeo2} and \eqref{scale2}). Usually, it is difficult to give a mathematical sense to \eqref{Z(g)CFT} because the measure $D \phi$, which is formally the Lebesgue measure on an infinite dimensional space, does not exist mathematically.  
Hence, CFT's are mostly studied using axiomatic and algebraic methods, or perturbative methods (formal stationary phase type expansions): see for example \cite{difrancesco,gaw}.\\

\textbf{Liouville Quantum Field Theory.} The first part of our work is to construct Liouville quantum field theory (LQFT in short) 
on a Riemann surface of genus ${\bf g}\geq 2$ and to show that this is a CFT.
We use probabilistic methods to give a mathematical sense to the path integral \eqref{Z(g)CFT}, when $S_g(\phi)=S_L(g,\phi)$ is the classical Liouville action, a natural convex functional coming from the theory of uniformisation of Riemann surfaces that we describe now. 
Given a two dimensional connected compact Riemannian manifold $(M,g)$ without boundary, we define the Liouville functional on $C^1$ maps $\varphi:M\to\R$ by
\begin{equation}\label{introaction}
S_L(g,\varphi):= \frac{1}{4\pi}\int_{M}\big(|d\varphi|_g^2+QK_g \varphi  + 4\pi \mu e^{\gamma \varphi  }\big)\,{\rm dv}_g
\end{equation}
where   $Q,\mu,\gamma>0$ are parameters to be discussed later. If  $Q=\frac{2}{\gamma}$, finding the minimizer $u$ of this functional  allows one to uniformize $(M,g)$. Indeed, the metric $g'=e^{\gamma u}g$ has constant scalar curvature $K_{g'}=-2\pi\mu\gamma^2$ and it is the unique such metric in the conformal class of $g$. 
The quantization of the Liouville action is precisely LQFT: one wants to make sense of  the following measure on some appropriate functional space $\Sigma$ (to be defined later) made up of (generalized) functions $\varphi:M\to \R$
\begin{equation}\label{pathintegral}
F\mapsto \Pi_{\gamma,\mu} (g,F):=\int_\Sigma F(\varphi)e^{-S_L(g,\varphi)}\,D\varphi
\end{equation}
where $D\varphi$ stands for the ``formal uniform measure'' on $\Sigma$. Up to renormalizing this measure by its total mass, this formalism describes the law of some random (generalized) function $\varphi$ on $\Sigma$, which stands for the (log-)conformal factor   of a random metric of the form $e^{\gamma \varphi}g$ on $M$.  In physics, LQFT is known to be a CFT with central charge ${\bf c_{\rm L}}:=1+6Q^2$  continuously ranging in $[25,\infty)$
for the particular values 
\begin{equation}
\gamma\in ]0,2],\quad Q=\frac{2}{\gamma}+\frac{\gamma}{2}.
\end{equation}
 Of course, this description is purely formal and giving a mathematical description of this picture is a longstanding problem, which goes back to the work of Polyakov \cite{Pol}.  The rigorous construction of such an object has been carried out recently in \cite{DKRV} in genus $0$, \cite{DRV} in genus $1$ (see also \cite{HRV} for the case of the disk). 
 Let us also mention that Duplantier-Miller-Sheffield \cite{DMS} have developed in the case of the sphere, disk or plane a  theory based on an equivalence class of random measures (equivalence classes are pushforwards of a given measure by  elements of a non trivial subgroup of biholomorphic transformations of the domain). From the point of view of LQFT, their approach lies at the ``boundary" of LQFT in the sense that they introduce the appropriate formalism in the case of the sphere, disk or plane to understand the $2$-point correlation function of LQFT; however, there is no cosmological constant (i.e. the constant $\mu$) in their approach and they do not work at the level of correlation functions. Yet, another approach by Takhtajan-Teo \cite{TT} 
was to develop a perturbative analysis (a semiclassical Liouville theory in the so-called background field formalism): 
in this non-probabilistic approach, LQFT is expanded as a formal power series in $\gamma$ around the minimum of the action \eqref{introaction} and the parameter $Q$ in the action is given by its value in classical Liouville theory $Q=\frac{2}{\gamma}$. 

We consider the genus ${\bf g}\geq 2$ case and give a mathematical, non perturbative, definition to 
 \eqref{pathintegral}. To explain our result, we need to summarize the construction. On a compact surface $M$ with genus ${\bf g}\geq 2$, we fix a smooth metric $g$ and define for $s\in\mathbb{R}$ the Sobolev space $H^s(M):=(1+\Delta_g)^{-s/2}(L^2(M))$ of order $s$ with scalar product defined using the metric $g$ and where $\Delta_g$ is the non-negative Laplacian associated to $g$. 
Using  the theory of the Gaussian free field (GFF in short), we show that for each $s>0$ there is a measure $\mc{P}'$ on $H^{-s}(M)$ which is independent of the choice of metric $g$ in the conformal class 
$[g]$, and which represents the following formal Gaussian measure defined for $F\in L^1(H^{-s}(M),\mc{P}')$ by 
\begin{equation}\label{mesureGFFint}
 \int F(\varphi)e^{-\frac{1}{4\pi}\int_M|\nabla\varphi|_g^2{\rm dv}_g} D\varphi:=
\frac{\sqrt{{\rm Vol}_g(M)}}{\sqrt{{\rm det}'(\Delta_g})}\int F(\varphi)d\mc{P}'(\varphi) 
\end{equation}
where ${\rm det}'(\Delta_g)$ is the regularized determinant of the Laplacian, defined as in Ray-Singer \cite{RaSi}. The method to do this is to consider a probability space $(\Omega,\mc{F},\mathbb{P})$ and a sequence $(a_j)_j$ of independent identically distributed real Gaussians in $\mc{N}(0,1)$ and to consider 
the following random variable (called GFF)
\begin{equation}
X_g=\sqrt{2\pi}\sum_{j\geq 1}a_j\frac{\varphi_j}{\sqrt{\la_j}} 
\end{equation}
with values in $H^{-s}(M)$ for all $s>0$, where $(\varphi_j)_{j\geq 0}$ is an orthonormal basis of eigenfunctions of $\Delta_g$ with eigenvalues $(\la_j)_{j\geq 0}$ (and with $\la_0=0$). The covariance of $X_g$ is the Green function of 
 $\frac{1}{2 \pi}\Delta_g$ and there is a probability measure $\mc{P}$ on $H^{-s}_0(M):=\{u\in H^{-s}(M); \cjg u,1\cjd=0\}$ 
 so that the law of $X_g$ is given by $\mc{P}$ and for each $\phi\in H_0^{s}(M)$, $\cjg X_g,\phi\cjd$ is a random variable on $\Omega$ with zero mean and variance $2\pi \cjg \Delta_g^{-1}\phi,\phi\cjd$. Then $H^{-s}(M)=H_0^{-s}(M)\oplus \mathbb{R}$ and we define  $\mc{P}'$ as the pushforward of the measure $\mc{P}\otimes dc$ under the mapping $(X,c)\in H_0^{-s}(M)\times\mathbb{R}\mapsto c+X$, where $dc$ is the uniform Lebesgue measure in $\mathbb{R}$. The formal equality \eqref{mesureGFFint} is an analogy with the finite dimensional setting. The next tool needed to the construction is   Gaussian multiplicative chaos theory introduced by Kahane \cite{cf:Kah}, which allows us to define the random measure $\mc{G}_g^\gamma:= e^{\gamma X_g}{\rm dv}_g$ on $M$ for $0<\gamma\leq 2$ when $X_g$ is the GFF. This is done by using a renormalization procedure, more precisely a regularization of $X_g$. We can then define the quantity which plays the role of the formal integral \eqref{pathintegral} as follows: 
for $F:  H^{-s}(M)\to\R$ (with $s>0$) a bounded continuous functional, we set
\begin{align}\label{defLQFT}
 \Pi_{\gamma, \mu}(g,F):=& ({\det}'(\Delta_{g})/{\rm Vol}_{g}(M))^{-1/2}  \\
 &\times \int_\R  \E\Big[ F( c+  X_{g}) \exp\Big( -\frac{Q}{4\pi}\int_{M}K_{g}(c+ X_{g} )\,{\rm dv}_{g} - \mu  e^{\gamma c}\mc{G}_{g}^\gamma(M)  \Big) \Big]\,dc \nonumber
\end{align}
and call it the functional integral  of LQFT (when $F=1$ this is the partition function). Our first result is that this quantity is finite and satisfies diffeomorphism invariance and a certain conformal anomaly when $Q=\frac{\gamma}{2}+\frac{2}{\gamma}$.  
 \begin{theorem}[LQFT is a CFT]\label{introweyl}
Let $Q=\frac{\gamma}{2}+\frac{2}{\gamma}$ with $\gamma\leq 2$ 
and $g$ be a smooth metric on $M$.
For each  bounded continuous functional $F:  H^{-s}(M)\to\R$ (with $s>0$) and each $\omega\in C^\infty(M)$,
 $\Pi_{\gamma, \mu}(e^{\omega}g,F)$ is finite  and satisfies the following conformal anomaly:  
\[\Pi_{\gamma, \mu}(e^{\omega}g,F)= \Pi_{\gamma, \mu}(g,F(\cdot\,-\tfrac{Q}{2}\omega))\exp\Big(\frac{1+6Q^2}{96\pi}\int_{M}(|d\omega|_g^2+2K_g\omega) {\rm dv}_g\Big).\]
Let $g$ be any metric on $M$ and $\psi:M\to M$ be an orientation preserving diffeomorphism, then we have for each bounded measurable $F:H^{-s}(M)\to \R$ with $s>0$
\[ \Pi_{\gamma, \mu} (\psi^*g ,F)= \Pi_{\gamma, \mu}(g,F(\cdot \circ \psi)) .\]
\end{theorem}
This Theorem says that LQFT is a conformal field theory with central charge ${\bf c}_L=1+6Q^2$.
As a quantum field theory, the other objects of importance for LQFT are the correlation functions.
In Section \ref{correlfct}, we   define the $n$-point correlation functions with vertex operators $e^{\alpha_i X_g(x_i)}$ where $\alpha_i$ are weights and $x_i\in M$ some points, and we show their conformal anomaly required to be a CFT. 
This amounts somehow to taking $F(\varphi)=\prod_{i=1}^ne^{\alpha_i\varphi(x_i)}$ in \eqref{defLQFT}, but it again requires renormalization since $\varphi$ lives in $H^{-s}(M)$ with $s>0$. At this level, the construction follows the method initiated by \cite{DKRV} on the sphere. We stress that for the sphere, only the $n$-point correlation functions for $n\geq 3$  are well defined, while here the partition function is already well-defined.\\ 

\textbf{Liouville Quantum Gravity and Polyakov partition function.} Our next result is the main part of the paper and consists in giving   sense to the Liouville quantum gravity (LQG in short) partition function following the work of Polyakov \cite{Pol}.
We stress that, even though the object  comes from theoretical physics, our result and proof is purely mathematical and can be viewed, from the perspective of a mathematician, as  a way to understand the behavior of   some natural interesting function near the boundary of moduli space, namely the LQFT partition function.

Given a connected closed surface $M$ with genus $\mathbf{g}\geq 2$, quantizing the coupling of the gravitational field with matter fields amounts to making sense of the formal integral (partition function) 
\begin{equation}\label{Zlqg} 
Z= \int_{\mc{R}}e^{-S_{{\rm EH}}(g)}\Big(\int e^{-S_{\rm M}(g,\phi_m)}D_g\phi_m\Big) Dg
\end{equation}
where the measure $Dg$ lives over the space of Riemannian structures  $\mc{R}$ on $M$, i.e. the space of metrics $g$ modulo diffeomorphisms. The functional integral for matter fields $\int e^{-S_{\rm M}(g,\phi_m)}D_g\phi_m$ stands for  the quantization of an action $\phi_m\mapsto S_{\rm M}(g,\phi_m)$ over an infinite dimensional space of fields describing matter, and $S_{\rm EH}$ is the Einstein-Hilbert action 
\begin{equation}\label{introEH}
S_{\rm EH}(g)=\frac{1}{2\kappa}\int_M K_g\,d{\rm v}_g+\mu_0{\rm Vol}_g(M),
\end{equation}
where $\kappa$ is the Einstein constant, $\mu_0\in\mathbb{R}$ is the cosmological constant. 
The measure $Dg$ represents the formal Riemannian measure associated to the $L^2$ metric on the space of Riemannian metrics, or in fact its reduction to $\mc{R}$.
There are several possible choices for the matter fields and we shall focus on the choice described in Polyakov \cite{Pol}.  Giving a mathematical definition to the functional integral \eqref{Zlqg} has been a real challenge, 
and Polyakov \cite{Pol} suggested a decomposition of this integral  in the case of bosonic string theory with $D$ free bosons. In that case, $Z_{\rm M}(g):=\int e^{-S_{\rm M}(g,\phi_m)}D_g\phi_m$ is the partition function of a CFT with central charge ${\bf c_{\rm M}}=D$, and is mathematically given by a certain power of the determinant of the Laplacian.
The argument of Polyakov \cite{Pol}, pursued by D'Hoker-Phong \cite{DhPh}, for defining \eqref{Zlqg} was based on the observation that each metric $g$ on $M$ can be decomposed as
\begin{equation}\label{decompg}
g=\psi^*(e^\omega g_\tau)
\end{equation}
where $\omega\in C^\infty(M)$, $\psi$ is a diffeomorphism and $(g_\tau)_{\tau\in \mc{M}_{\rm g}}$ is a family of hyperbolic metrics on $M$ parameterizing the moduli space $\mc{M}_g$ of genus-${\bf g}$ surfaces. We recall that $\mc{M}_g$ is the space of equivalence classes of conformal structures: it is a $6{\bf g}-6$ dimensional orbifold equipped with a natural metric, called the Weil-Petersson metric, whose volume form denoted $d\tau$ has finite volume. In this way, the space  of Riemannian structures $\mc{R}$ is identified to the product of moduli space $\mc{M}_{\rm g}$ with the Weyl group $C^\infty(M)$ acting on metrics by $(\varphi,g)\mapsto e^{\varphi}g$. Applying the change of variables \eqref{decompg} in the formal integral \eqref{Zlqg} produces a Jacobian, called the ghost determinant, taking into account the quotient of the space of metrics by the space of diffeomorphisms of $M$. The ghost determinant turns out to be the partition function of a CFT with central charge ${\bf c_{{\rm ghost }}}=-26$ and Polyakov  noticed that at the specific value $D=26$  the conformal anomaly of $Z_{\rm M}$  cancels out that of the ghost term,  giving rise to a Weyl invariant  partition function 
\begin{equation}\label{reductionDhPh} 
Z=\int_{\mc{M}_g} Z_{\rm M}(g_\tau)Z_{\rm Ghost}(g_\tau)\sqrt{\det J_{g_\tau}}d{\tau}
\end{equation} 
called  {\it critical string theory}\footnote{The term "critical" refers in fact to the {\it critical dimension} $D=26$ needed to get a Weyl invariant theory without quantizing the Weyl factor $e^{\omega}$ in \eqref{decompg}.}. Concretely, this was further discussed by D'Hoker-Phong \cite{DhPh} who wrote
\begin{equation}\label{ghostdet} 
 Z_{\rm Ghost}(g)=\Big(\frac{\det(P^*_gP_g)}{\det J_g}\Big)^{1/2}, \quad Z_{\rm M}(g)=C\Big(\frac{\det'(\Delta_g)}{{\rm Vol}_g(M)}\Big)^{-\tfrac{D}{2}}\end{equation}
for some constant $C$, where the determinants are defined using spectral zeta functions, $P_g$ is a first-order elliptic operator mapping $1$-forms to trace-free symmetric $2$-tensors, and $J_g$ is the Gram matrix of a fixed basis of $\ker P_g^*$ (see Section \ref{fullpartfct}  further details). Then Belavin-Knizhnik \cite{BeKn} and Wolpert \cite{Wo2} proved that the 
integral \eqref{reductionDhPh} with $D=26$ diverges  at the boundary of (the compactification of) moduli space, a problematic fact in order to establish well-posedness of the partition function for critical  $D=26$ (bosonic) strings.

Noncritical string theories are not formulated within the critical dimension $D=26$, yet they are Weyl invariant. The idea, emerging once again from the paper \cite{Pol}, is that for $D\not=26$ the integral \eqref{Zlqg} possesses one further degree of freedom to be integrated over corresponding to the Weyl factor $e^\omega$ in \eqref{decompg}. For $D\leq 1$, hence ${\bf c_{\rm M}} \leq 1 $, Polyakov argued that integrating this factor requires using Liouville quantum field theory. In other words, applying once again the change of variables \eqref{decompg} to \eqref{Zlqg}  yields
\begin{equation}\label{intro-part}
Z=\int_{\mc{M}_{\bf g}} Z_{\rm M}(g_\tau)Z_{\rm Ghost}(g_\tau)\Pi_{\gamma,\mu} (g_\tau,1)\sqrt{\det J_{g_\tau}}\, d\tau
\end{equation}
where $\Pi_{\gamma,\mu} (g_\tau,1)$ is the partition function of LQFT in the background metric $g_\tau$. As explained above, the partition function $\Pi_{\gamma,\mu} (g_\tau,1)$ depends on two parameters $\gamma$ and $\mu$ (Weyl invariance forces $\gamma$ to be an explicit function of  ${\bf c_{\rm M}}$).  Later, Polyakov's argument   was generalized by David and Distler-Kawai \cite{cf:Da,DistKa} to  CFT type matter field theories with central charge ${\bf c_{\rm M}} \leq 1$ (thus including the case $D=1$).  In that CFT context, the integral \eqref{intro-part} is often called partition function of LQG, so that we will write $Z_{{\rm LQG}}$ for the partition function $Z$. This approach has an important consequence related to string theory as it paves the way to a rigorous construction of {\it noncritical bosonic string theory }\footnote{Noncritical bosonic string theory is sometimes referred to as {\it critical $D=2$ string theory}, by opposition to the critical $D=26$ string theory. The two dimensions correspond to one dimension for the embedding into $\R$ and one dimension for the Weyl factor: in other words the Weyl factor $\omega$ in \eqref{decompg}  plays the role of a hidden dimension, see the explanations in \cite{polchinski} page 121.}  provided one can make sense of \eqref{intro-part}. This is the main purpose of this paper. The importance of this theory is discussed  in great details in \cite[section 5.1]{polchinski}  or \cite{Kleb} for instance.

In what follows, we will therefore consider the partition function $Z_{{\rm LQG}}$ defined by \eqref{intro-part}
where  we choose for the matter partition function 
\begin{equation}\label{introstringaction}
Z_{\rm M}(g)=\Big(\frac{{\det}' \Delta_g}{{\rm Vol}_g(M)}\Big)^{-\frac{{\bf c_{\rm M}}}{2}},
\end{equation}
while the ghost determinant $Z_{\rm Ghost}(g_\tau)$ is defined by  \eqref{ghostdet} and $d\tau$ is the Weil-Petersson measure. Notice that \eqref{introstringaction} is nothing but the partition function \eqref{ghostdet} for $D=c_{\rm M}$ free bosons  extended to all possible values $D\leq 1$. The parameters in  \eqref{defLQFT}   are tuned in such a way that the global conformal anomaly of the product 
$$Z_{\rm M}(g_\tau)Z_{\rm Ghost}(g_\tau)\Pi_{\gamma,\mu} (g_\tau,1)$$ vanishes, hence ensuring Weyl invariance of the whole theory \eqref{intro-part}. In view of Theorem \ref{introweyl}, this gives  the relation
$${\bf c_{\rm M}}-26+1+6Q^2=0,$$
hence determining the value of $\gamma$ (encoded by $Q$) in terms of the central charge ${\bf c_{\rm M}}$ of the matter fields.  We refer to Section \ref{fullpartfct} for more explanations.

 The main result of this paper is the following:
\begin{theorem}[Convergence of the partition function]\label{mainth2}
For surfaces of genus ${\bf g}\geq 2$, the integral defining the partition function $Z_{\rm LQG}$ of \eqref{intro-part} converges for $\gamma\in]0,2]$, that is for $\mathbf{c}_{\rm M}\leq 1$.
\end{theorem}

The  integral defining $Z_{\rm LQG}$ in the case   $\mathbf{c}_{\rm M}=0$ corresponds to the case of {\it pure gravity} (i.e. no matter),  $\mathbf{c}_{\rm M}=-2$ to  uniform spanning trees and $\mathbf{c}_{\rm M}=1$ (equivalently $D=1$)   to  noncritical strings (or $D=2$ string theory).  As far as we know, this is the first proof of convergence of string theory on hyperbolic surfaces with fixed genus.

Using this Theorem, we can now see the metric $g$ on $M$ as a random variable with law ruled by the partition function \eqref{intro-part}. The Riemannian volume and modulus of this random metric are called    the quantum gravity volume form  (LQG volume form) and quantum gravity modulus, see Theorem \ref{defLiouvillemeas}.  Furthermore we formulate  conjectures relating the LQG volume form   to the scaling limit of random planar maps in the case of pure gravity $\mathbf{c}_{\rm M}=0$, hence providing the scaling limit of the model  studied in \cite{miermont}, or to the scaling limit of random planar maps weighted by the discrete Gaussian free field in the case $\mathbf{c}_{\rm M}=1$, see Section \ref{randplanar}. Though we do not explicitly write a conjecture, we further mention here that the limit of random planar maps with fixed topology weighted by uniform spanning trees corresponds to $\mathbf{c}_{\rm M}=-2$.

The main input of our paper is the proof of Theorem \ref{mainth2}. We need to analyse the integrands in \eqref{intro-part} near the boundary of moduli space and show that we can control them. The moduli space $\mc{M}_{\bf g}$ can be viewed as a 
$6{\bf g}-6$ dimensional non-compact orbifold of hyperbolic metrics on $M$, that can be compactified in such a way that its boundary corresponds 
to pinching closed geodesics. Hyperbolic metrics on $M$ corresponding to points in $\pl\mc{M}_{\bf g}$ are complete hyperbolic surfaces with cusps and finite volume. The estimates of Wolpert \cite{Wo2} describe the parts involving the ghost and matter terms. The heart of our work
is to analyse the behavior of Gaussian multiplicative chaos under degeneracies of the hyperbolic surfaces: this is rather involved since there are in general 
small eigenvalues of Laplacian tending to $0$ and the covariance of the GFF (i.e. the Green function) is thus diverging. There is yet a huge conceptual gap between the cases  $\mathbf{c}_{\rm M}<1$ and  $\mathbf{c}_{\rm M}=1$. Roughly, the reason is the following: the product $Z_{\rm M}(g_\tau)Z_{\rm Ghost}(g_\tau)$ is at leading order comparable to $\prod_{j}e^{-\frac{\pi^2}{3\ell_j}(1-\tfrac{\mathbf{c}_{{\rm M}}}{2 })}$, where the product runs over pinched geodesics with lengths $\ell_j\to 0$ when approaching the boundary of the (compactification of) moduli space  -- see subsection \ref{detoflap} for more precise statements-- whereas $\Pi_{\gamma,\mu}(g_\tau,1)$ is comparable to $\prod_{j}e^{-\frac{\pi^2}{3\ell_j}(-\tfrac{1}{2 })}\times F(\mc{G}_g^\gamma(M))$ where $F(\mc{G}_g^\gamma(M))$ is an explicit functional expectation of the Gaussian multiplicative chaos $\mc{G}_g^\gamma(M)$. Hence, for  $\mathbf{c}_{\rm M}<1$, we prove soft estimates on the functional $F(\mc{G}_g^\gamma(M))$ that are enough to get an exponential decay of the product $Z_{\rm M}(g_\tau)Z_{\rm Ghost}(g_\tau)\Pi_{\gamma,\mu}(g_\tau,1)$ at the boundary of $\mc{M}_{\bf g}$ and thus integrability with respect to the Weil-Petersson measure. In the case $\mathbf{c}_{\rm M}=1$, the leading exponential behaviors  cancel out exactly  so that the analysis must determine the polynomial corrections behind the leading exponential behavior, rendering the computations more much intricate. 
In order to analyse the mass of the Gaussian multiplicative chaos measure in these degenerating regions, we need to prove uniform estimates (that, as far as we know, are new) on the Green function and on the eigenfunctions associated to the small eigenvalues in the pinched necks of the surface, as functions of the moduli space parameters $\tau$ when $\tau$ approaches the boundary $\mc{\pl}\mc{M}_{\bf g}$. 
Roughly speaking, the crucial observation is that the GFF behaves like two independent Brownian motions in the variable transverse to the closed geodesic being pinched, and this allows us to translate the problem in terms of explicit functionals of Brownian motion. For (and only for) $\mathbf{c}_{\rm M}=1$, we show that the pinching produces an extra rate of decay of $\Pi_{\gamma,\mu} (g_\tau,1)$ as we approach $\mc{\pl}\mc{M}_{\bf g}$, implying the convergence.  

To conclude this introduction, we point out that an interesting different approach to define path integrals for random K\"ahler metrics on surfaces was introduced recently by Ferrari-Klevtsov-Zelditch \cite{FKZ,KlZe}, but the link with our work is not established rigorously.

\subsubsection*{Acknowledgements:} This project has received funding from the European Research Council (ERC) under the European Union’s Horizon 2020 research and innovation programme (grant agreement No. 725967). C.G. was also partially funded by grants ANR-13-BS01-0007-01, R.R. and V.V.  are partially funded by the ANR project Liouville (ANR-15-CE40-0013). We would like to thank A. Bilal, E. D'Hoker, G. Freixas, I. Klebanov, F. Rochon, L. Takhtajan for useful conversations, and the anonymous referees for their comments and suggestions.

\section{Geometric background and Green functions}  \label{sectionNotation}
 
\subsection{Uniformisation of compact surfaces of genus ${\bf g}\geq 2$}
 
Let $M$ be a compact surface of genus ${\bf g}\geq 2$ and let $g$ be a smooth Riemannian metric. Recall that Gauss-Bonnet tells us that 
\begin{equation}\label{GB} 
\int_{M}K_g{\rm dv}_g=4\pi\chi(M)
\end{equation}
where $\chi(M)=(2-2{\bf g})$ is the Euler characteristic, $K_g$ the scalar curvature of $g$ and ${\rm dv}_g$ the Riemannian measure.  
The uniformisation theorem says that in the conformal class 
\[ [g]:=\{ e^{\varphi}g; \varphi\in C^\infty(M)\}\]
of $g$, there exists a unique metric $g_0=e^{\varphi_0}g$ of scalar curvature $K_{g_0}=-2$. 
For a metric $\hat{g}=e^{\varphi}g$, one has the relation 
\[ K_{\hat{g}}=e^{-\varphi}(\Delta_g \varphi+K_{g})\]
where $\Delta_g=d^*d$ is the non-negative Laplacian (here $d$ is exterior derivative and $d^*$ its adjoint). 
Finding $\varphi_0$ is achieved by minimizing the following functional 
\[ F: C^\infty(M)\to \R^+, \quad F(\varphi):= \int_{M} (\demi |d\varphi|^2_g+K_g\varphi+2e^{\varphi}){\rm dv}_g
\] 
and taking $\varphi_0$ to be the function such that $F(\varphi)$ is minimum at $\varphi=\varphi_0$.
We will embed this functional into a more general one, depending on three parameters, called \emph{Liouville functional}:  
let $\gamma ,Q,\mu>0$ and  define
\begin{equation}\label{QLiouville}
S_L(g,\varphi):=\frac{1}{4\pi}\int_M \big(|d\varphi|^2_g+QK_g\varphi +4\pi \mu e^{\gamma \varphi}\big)\, {\rm dv}_g.
\end{equation}
When $Q=2/\gamma$ and $\pi\mu\gamma^2=1$, we can write $S_L(g,\varphi)=\frac{1}{2\gamma^2\pi}F(\gamma\varphi)$. In fact,   if $\hat{g}=e^{\omega}g$ for some $\varphi$, the functional $S_L$ satisfies the relation  
\[S_L\Big(\hat{g},\varphi-\frac{\omega}{\gamma}\Big)=S_L(g,\varphi)+\frac{1}{4\pi}\int_{M}\Big(\big(\frac{1}{\gamma^2}-\frac{Q}{\gamma}\big)|d\omega|^2_g-\frac{Q}{\gamma}K_g\omega+\big(Q-\frac{2}{\gamma}\big)\varphi\Delta_g\omega 
\Big){\rm dv}_g\]
and in particular if $Q=2/\gamma$ it satisfies 
\begin{equation}\label{conformS}
S_L\Big(\hat{g},\varphi-\frac{\omega}{\gamma}\Big)=S_L(g,\varphi)-\frac{1}{4\pi\gamma^2}\int_{M}(|d\omega|^2_g+2K_g\omega){\rm dv}_g,\end{equation}
which is called conformal anomaly: changing the conformal factor of the metric entails a variation  of the functional proportional to the Liouville functional. Similar properties will be shared by the quantum version  of the Liouville theory, which fall under the scope of Conformal Field Theory. At this stage let us just mention that we will show that the value of $Q$ for the quantum Liouville theory to possess a conformal anomaly has to be adjusted to take into account quantum effects. More precisely we will have in the quantum theory
\begin{equation}\label{valueQ}
Q=\frac{2}{\gamma}+\frac{\gamma}{2}.
\end{equation}

 \subsection{Hyperbolic surfaces, Teichm\"uller space and Moduli space}\label{Sec:Teichmuller}

Let $M$ be a surface of genus ${\bf g}\geq 2$. The set of smooth metrics on $M$ is a Fr\'echet manifold denoted by 
${\rm Met}(M)$ and contained in the Fr\'echet space of smooth symmetric tensors $C^\infty(M;S^2T^*M)$ of order $2$. This space has a natural $L^2$ metric given by 
\begin{equation}\label{L2metric} 
\cjg h_1,h_2\cjd := \int_M \cjg h_1,h_2\cjd_g{\rm dv}_g
\end{equation}
where $h_1,h_2\in T_g{\rm Met}(M)=C^\infty(M;S^2T^*M)$ and $\cjg\cdot,\cdot\cjd_g$ is the usual scalar product on endomorphisms of $TM$ when we identify symmetric $2$-tensors with endomorphisms of $TM$ through the metric $g$.
 A metric with Gaussian curvature $-1$ will be called hyperbolic, 
 we denote by ${\rm Met}_{-1}(M)$ the set of such metrics on $M$. 
 The group $\mc{D}(M)$ of smooth diffeomorphisms acts smoothly and properly on ${\rm Met}(M)$ and on 
 ${\rm Met}_{-1}(M)$ by pull-back $\phi.g:=\phi^*g$, moreover it acts by isometries with respect to
  the metric \eqref{L2metric}. The subgroup $\mc{D}_0(M)\subset \mc{D}(M)$ of elements contained in the connected component of the Identity also acts properly and smoothly and ${\rm Mod}(M):=\mc{D}(M)/\mc{D}_0(M)$ is a discrete subgroup called \emph{mapping class group} or \emph{moduli group}. The Fr\'echet space $C^\infty(M)$ acts on ${\rm Met}(M)$ by conformal multiplication 
$(\varphi, g)\mapsto e^{\varphi}g$. The orbits of this action are called \emph{conformal classes} and the conformal class of a metric 
$g$ is denoted by $[g]$.

The \emph{Teichm\"uller space} of $M$ is defined by 
\[\mc{T}(M) := {\rm Met}_{-1}(M)/\mc{D}_0(M).\]
By taking slices transverse to the action of $\mc{D}_0(M)$, we can put a structure of smooth manifold with real dimension $6{\bf g}-6$, it is topologically a ball, and
its tangent space at a metric $g$ (representing a class in $\mc{T}(M)$) can be identified naturally with the space of divergence-free trace-free tensors with respect to $g$ by choosing appropriately the slice.
Teichm\"uller space has a complex structure and is equipped with a natural K\"ahler metric called the \emph{Weil-Petersson} metric, which is defined by 
\[ \cjg h_1,h_2\cjd_{{\rm WP}} := \int_M \cjg h^{\rm tf}_1,h_2^{\rm tf}\cjd_g{\rm dv}_g\]
if $h_1,h_2\in T_g\mc{T}(M)$ and $h_i^{\rm tf}=h_i-\demi {\rm Tr}_g(h_i)g$ denotes the trace-free part. 
The Weil-Petersson metric is not the metric induced by \eqref{L2metric} after quotienting by $\mc{D}_0(M)$ but it is rather induced by the $L^2$ metric on almost complex structures, when we identify almost complex structures with metrics of constant curvature. We refer for to the book of Tromba \cite{Tr} for more details about this approach of Teichm\"uller theory, the Weil-Petersson metric is discussed in Section 2.6 there.\\

The group ${\rm Mod}(M)$ acts 
properly discontinuously on $\mc{T}(M)$ by isometries of the Weil-Petersson metric, but the action is not free and there are elements of finite order.
The quotient $\mc{M}(M):=\mc{T}(M)/{\rm Mod}(M)$ 
is a Riemannian orbifold called the \emph{moduli space} of $M$, its orbifold singularities corresponding to 
hyperbolic metrics admitting isometries. 
 Since $\mc{T}(M)$, ${\rm Mod}(M)$ and $\mc{M}(M)$ actually depend only on the genus ${\bf g}$ of 
$M$, we shall denote them $\mc{T}_{\bf g}$, ${\rm Mod}_{\bf g}$ and $\mc{M}_{\bf g}$.
The manifold $\mc{M}_{\bf g}$ is open but can be compactified into $\bbar{\mc{M}_{\bf g}}$, the locus of the compactification is a divisor $D\subset\bbar{\mc{M}_{\bf g}}$ and the Weil-Petersson distance is complete on that space. 
Since we will need to understand the behavior of certain quantities on the moduli space, we 
now recall its geometry near the divisor $D$ and we shall follow the description 
given by Wolpert (\cite{Wo1,Wo2,Wo3}) for this compactification. On a surface $M$ of genus ${\bf g}$, there is a unique geodesic in each free homotopy class, and we call a partition of $M$ a collection of $3{\bf g}-3$ simple closed curves 
$\{\gamma_1,\dots,\gamma_{3{\bf g}-3}\}$ which are 
not null-homotopic and not mutually homotopic. If $g\in {\rm Met}_{-1}(M)$, there is a unique simple geodesic homotopic to each $\gamma_j$ and we obtain a decomposition of $(M,g)$ into $2{\bf g}-2$ hyperbolic pants 
(a pant is a topological sphere with $3$ disks removed, equipped with a hyperbolic metric and with totally geodesic boundary). A \emph{subpartition} of $M$ is a collection of $n_p$ simple curves 
$\{\gamma_1,\dots,\gamma_{n_p}\}$ which are not null-homotopic and not mutually homotopic, with $n_p\leq 3{\bf g}-3$; they disconnect the surface into surfaces with boundary. A surface $(M_0,g_0)$ 
in $\pl \bbar{\mc{M}_{\bf g}}$ is a surface with nodes: $M_0$ is the interior of 
a compact surface $M$ with $n_p$ simple curves $\gamma_1,\dots,\gamma_{n_p}$ removed and $g_0$ is a complete hyperbolic metric with finite volume on $M_0$, the metric in a collar neighborhood $[-1,1]_\rho\x (\R/\zz)_\theta$ of each $\gamma_j$ (with $\gamma_j=\{\rho=0\}$) being 
\[ g_0= \frac{d\rho^2}{\rho^2}+\rho^2d\theta^2.\] 
Notice that these corresponds to a pair of hyperbolic cusps, each one 
isometric to $(\R^+_t\times (\R/\Z)_\theta,dt^2+e^{-2t}d\theta^2)$ by setting $\rho=\pm e^{-t}$. 
Now there is a neighborhood $\bbar{U}_{g_0}$ of $g_0$ in $\bbar{\mc{M}_{\bf g}}$ represented by 
hyperbolic metrics $g_{s,\tau}$ on $M_0$ for some parameter $(s,\tau)\in \cc^{3g-3-m}\x\cc^{m}$ near $0$, with $g_{s,\tau}$ which are smooth metrics on $M$ when $\tau\not=0$ and complete smooth metrics with hyperbolic cusps on $M_0$ when $\tau=0$, and $g_{0,0}=g_0$. Moreover, the metrics 
$g_{s,\tau}$ are continuous with respect to $(s,\tau)$ on compact sets of $M_0$ for $(s,\tau)$ near $0$ (in the $C^\infty$ topology), they are given in the fixed collar neighborhood $[-1,1]_\rho\x (\R/\zz)_\theta$ of $\gamma_j$  by 
\begin{equation}\label{expressiongst} 
g_{s,\tau}= e^{\varphi_{s,\tau}}\Big( \frac{d\rho^2}{\eps_j^2+\rho^2}+(\eps_j^2+\rho^2)d\theta^2\Big) \end{equation}
with $\eps_j:=2\pi^2/|\log|\tau_j||$ and $\varphi_{s,\tau}\in C^\infty(M)$ satisfies 
\[ e^{\varphi_{s,\tau}}-1\to 0 \textrm{ as }(s,\tau)\to 0\]
in $C^0$ norm (and in fact in $C^\infty$ on compact sets of $M_0$). Here we notice that the metric 
$\frac{d\rho^2}{\eps_j^2+\rho^2}+(\eps_j^2+\rho^2)d\theta^2$ has curvature $-1$ in the collar 
and is isometric to 
\begin{equation}\label{modelepincement} 
[-t_j,t_j]_t\x (\R/\zz)_{\theta}, \quad dt^2+\eps_j^2\cosh(t)^2d\theta^2
\end{equation}
by setting $\eps_j\sinh(t)=\rho$ and $\eps_j\sinh(t_j)=1$. 
The geodesic $\gamma_j(g_{s,t})$ for $g_{s,t}$ homotopic to $\gamma_j$ has length 
\[ \ell_j(g_{s,\tau})=\eps_j(1+o(1)) \textrm{ as }|(s,\tau)|\to 0\]
and is contained in a neighborhood $[-c\eps_j,c\eps_j]\x \R/\zz$ of the collar near $\gamma_j$ for some $c>0$ independent of $\eps_j$ (or equivalently in $t\in [-c,c]$).
Using $|e^{\varphi_{s,\tau}}-1|<\delta$ for some small $\delta>0$,
the set $B_j=\{m\in M; d_{g}(\gamma_j(g_{s,\tau}),m)\geq |\log \ell_j(g_{s,\tau})|\}$ is contained in a compact set of $M_0$ uniform in $(s,\tau)$ where the metrics depend continously on $(s,\tau)$. We can then use geodesic normal coordinates with respect to $g$ around $\gamma_j(g_{s,\tau})$ and the collar $\mc{C}_j(g_{s,\tau})=M\setminus B_j$ is isometric to  
\begin{equation}\label{modelepincementgeo} 
[-|\log \ell_j(g_{s,\tau})|, |\log \ell_j(g_{s,\tau})|]_t\x (\R/\zz)_{\theta}, \quad dt^2+\ell_j(g_{s,\tau})^2\cosh(t)^2d\theta^2.
\end{equation}
To summarize, the geometry is uniformly bounded outside $\cup_{j=1}^r\mc{C}_j(g)$ for metrics $g$
in a small neighborhood $\bbar{U}_{g_0}$ of $g_0$ in $\bbar{\mc{M}_{{\bf g}}}$.
 
The loops $(\gamma_j)_j$ define a subpartition. The open strata of $D$ correspond to subpartitions up to equivalence by elements in ${\rm Mod}_{\bf g}$.
For each $\beta>0$, the set of metrics in $\mc{M}_{\bf g}$ such that all geodesics have length larger than $\beta$ is a compact subset of $\mc{M}_{\bf g}$ called the $\beta$-thick part. The $\beta$-thin part of $\mc{M}_{\bf g}$ is the complement of the $\beta$-thick part.
By Lemma 6.1 of \cite{Wo2}, there exists a constant $\beta>0$ so that the $\beta$-thin part of $\mc{M}_{\bf g}$ is covered by a finite set of neighborhoods $U({\rm SP}_j)$, $j=1,\dots,J$ where ${\rm SP}_j$ denote some subpartitions of $M$ 
and $U({\rm SP}_j)$ denote the set of surfaces in Teichm\"uller space 
(up to ${\rm Mod}_{\bf g}$ equivalence) for which the geodesics in the homotopy class of curves of 
${\rm SP}_j$ have length less than $\beta$ and the other ones have length bounded below by $\beta/2$. Each $U({\rm SP}_j)$ is a neighborhood of a strata of $D$. 

For each pants decomposition of the surface (with genus ${\bf g}$), one has associated coordinates $\tau=(\ell_1,\dots,\ell_{3{\bf g}-3},\theta_1,\dots, \theta_{3{\bf g}-3})$ where $\ell_j$ are the lengths of the simple closed geodesics bounding the pair of pants and $\theta_j\in [0,2\pi)$ are the twist angles (see \cite{Wo1}). The Weil-Petersson volume form is given in these coordinates by 
\begin{equation}\label{WPvol}
d\tau:=C_{\bf g} \prod_{j=1}^{3{\bf g}-3}\ell_j d\theta_j d\ell_j
\end{equation}
for some constant $C_{\bf g}>0$ depending only on the genus.
   
\subsection{Determinant of Laplacians} \label{detoflap}

For a Riemannian metric $g$ on a connected oriented compact surface $M$, the non-negative Laplacian $\Delta_g=d^*d$ has discrete spectrum
${\rm Sp}(\Delta_g)=(\la_j)_{j\in \N_0}$ with $\la_0=0$ and $\la_j\to +\infty$. We can define the determinant of $\Delta_g$ by 
\[ {\det} '(\Delta_g)=\exp(-\pl_s\zeta(s)|_{s=0})\]
where $\zeta(s):=\sum_{j=1}^\infty \la_j^{-s}$ is the spectral zeta function of $\Delta_g$, which admits a meromorphic continuation from ${\rm Re}(s)\gg 1$ to $s\in \cc$ and is holomorphic at $s=0$. We recall that if $\hat{g}=e^{\varphi}g$ for some $\varphi\in C^\infty(M)$, one has the so-called Polyakov formula (see \cite[eq. (1.13)]{OPS}) 
\begin{equation}\label{detpolyakov} 
\log \frac{{\det}'(\Delta_{\hat{g}})}{{\rm Vol}_{\hat{g}}(M)}= \log \frac{{\det}'(\Delta_g)}{{\rm Vol}_{g}(M)} -\frac{1}{48\pi}\int_M(  |d\varphi|_g^2+2K_g\varphi){\rm dv}_g\end{equation}
where $K_g$ is the scalar curvature of $g$ as above. It is interesting to compare \eqref{detpolyakov} with the conformal anomaly
\eqref{conformS} of the Liouville action $S_L$.
To compute ${\det}'(\Delta_g)$, it thus suffices to know it for an element in the conformal class, and by the uniformisation theorem we can choose a metric $g$ of scalar curvature $-2$ (or equivalently Gaussian curvature $-1$) if $M$ has genus ${\bf g}\geq 2$. Such hyperbolic surface can be realized as a quotient $\Gamma\backslash \hh^2$ of the hyperbolic half-plane 
\[\hh^2 :=\{z\in \cc; {\rm Im}(z)>0\} \textrm{ with metric }g_{\hh^2}=\frac{|dz|^2}{({\rm Im}(z))^2}\]
 by a discrete co-compact subgroup $\Gamma\subset {\rm PSL}_2(\R)$ with no torsion. 
 In each free homotopy class on $M=\Gamma\backslash \hh^2$, there is a unique closed geodesic, and we can form the Selberg zeta function 
\[ Z_{g}(s)=\prod_{\gamma\in \mc{P}}\prod_{k=0}^\infty(1-e^{-(s+k)\ell(\gamma)}), \quad {\rm Re}(s)>1\]
where $\mc{P}$ denotes the set of primitive closed geodesics of $(M,g)\simeq \Gamma\backslash \hh^2$ and $\ell(\gamma)$ are their lengths (recall that primitive closed geodesics are oriented 
closed geodesics that are not iterates of another closed geodesic). By the work of Selberg, The function $Z_g(s)$ admits an analytic continuation to $s\in \C$ and it is proved by D'Hoker-Phong \cite{DhPh} that 
\begin{equation}\label{detzeta}
{\det}'\Delta_{g}=Z'_g(1)e^{(2{\bf g}-2)C}
\end{equation}
where $C$ is an explicit universal constant (see also D'Hoker-Phong or Sarnak \cite{DhPh2,Sa}). The behavior of $Z'_g(1)$ near the boundary of $\mc{M}_{\bf g}$ is studied by Wolpert \cite{Wo2}: there exists $C_{\bf g}>0$ a constant depending only on the genus such that for all $g\in \mc{M}_{\bf g}$
\begin{equation}\label{Wolpert1} 
C_{\bf g}^{-1}\prod_{j=1}^{n_p}\frac{e^{-\frac{\pi^2}{3\ell_j(g)}}}{\ell_j(g)} \prod_{\la_k(g)<1/4}\la_k(g)\leq Z'_g(1)\leq 
C_{\bf g}\prod_{j=1}^{n_p}\frac{e^{-\frac{\pi^2}{3\ell_j(g)}}}{\ell_j(g)} \prod_{\la_k(g)<1/4}\la_k(g)
\end{equation}
where $\la_k(g)$ are the eigenvalues of $\Delta_g$ and $\ell_j(g)$ are the lengths of closed geodesics with length less than $\eps>0$ for some small fixed $\eps>0$.\\

There is another operator which appears in the work of Polyakov \cite{Pol} and whose determinant is important in 2D quantum gravity. Let $P_g$ be the differential operator mapping differential $1$-forms on $M$ to symmetric trace-free $2$-tensors, defined by 
\[ P_g\, \omega:=2 \mc{S}\nabla^g \omega - {\rm Tr}_{g}(\mc{S}\nabla^g \omega)g.\]
Here $\nabla^g$ is the Levi-Civita connection for $g$, ${\rm Tr}_g$ denotes the trace with respect to $g$ and 
$\mc{S}$ denotes the orthogonal projection on symmetric 
$2$-tensors. The kernel of $P_g$ is the space of conformal Killing vector fields, which is thus trivial in genus ${\bf g}\geq 2$.  Its adjoint $P_g^*$ is given by $P_g^*u:=\delta_g(u)=-{\rm Tr}_g(\nabla^g u)$ and called the divergence operator on symmetric trace-free $2$-tensors. Its kernel has real dimension $6{\bf g}-6$ and is conformally invariant. We denote by $(\phi_1,\dots,\phi_{6{\bf g}-6})$ a fixed basis of $\ker P_g^*$ and 
by $J_g$ the matrix $(J_g)_{ij}=\cjg \phi_i,\phi_j\cjd_{g}$ 
The operator $P_g^*P_g$ is an elliptic positive self-adjoint second order differential operator acting on $1$-forms,  and we can define its determinant by 
\[ {\det}(P_g^*P_g)=\exp(-\pl_s\zeta_1(s)|_{s=0}), \quad \zeta_1(s)=\sum_{j=1}^\infty \mu_j^{-s}\]
where $\mu_j>0$ are the non-zero eigenvalues of $P_g^*P_g$. 
The conformal anomaly for this operator is proved by Alvarez \cite[Eq. 4.27]{Al} and reads 
\begin{equation}\label{anomaliePg} 
\log \frac{{\det}(P_{\hat{g}}^*P_{\hat{g}})}{\det J_{\hat{g}}}= \log \frac{{\det}(P_g^*P_g)}{\det J_{g}} -\frac{13}{24\pi}\int_M(  |d\varphi|_g^2+2K_g\varphi){\rm dv}_g
\end{equation}
if $\hat{g}=e^{\varphi}g$.
By \cite{DhPh}, one has for $(M,g)$ a hyperbolic surface 
realized as $\Gamma\backslash \H^2$
\begin{equation}
\label{detzeta_1}
{\rm det}(P_g^*P_g)^{\demi}=Z_g(2)e^{(2{\bf g}-2)C'}
\end{equation}
for some universal constant $C'$, and $Z_g(s)$ is again the Selberg zeta function. The behavior of $Z_g(2)$ near the boundary of $\mc{M}_{\bf g}$ is also studied by Wolpert \cite{Wo2}: there exists $C_{\bf g}>0$ a constant depending only on the genus such that for all $g\in \mc{M}_{\bf g}$
\begin{equation}\label{Wolpert2} 
C_{\bf g}^{-1}\prod_{j=1}^{n_p}\frac{e^{-\frac{\pi^2}{3\ell_j(g)}}}{\ell_j^{3}(g)} \leq Z_g(2)\leq 
C_{\bf g}\prod_{j=1}^{n_p}\frac{e^{-\frac{\pi^2}{3\ell_j(g)}}}{\ell_j^3(g)} 
\end{equation}
where the $\ell_j(g)$ are the lengths of closed geodesics with length less than $\eps>0$ 
for some small fixed $\eps>0$.

\subsection{Green function and resolvent of Laplacian}\label{sec:Green}
Each compact Riemannian surface $(M,g)$ has a (non-negative) Laplace operator $\Delta_g=d^*d$ and a Green function $G_g$ defined 
to be the integral kernel of the resolvent operator $R_g:L^2(M)\to L^2(M)$ satisfying $\Delta_gR_g={\rm Id}-\Pi_0$, $R_g^*=R_g$ and $R_g1=0$, where $\Pi_0$ is the orthogonal projection in $L^2(M,{\rm dv}_g)$ on $\ker \Delta_g$ (the constants). By integral kernel, we mean that for each $f\in L^2(M)$
\[ R_gf(x)=\int_{M} G_g(x,x')f(x'){\rm dv}_g(x').\]
It is well-known 
(see for example \cite{Pa}) that the hyperbolic space $\hh^2$ 
also has a family of Green functions (here $d_{\hh^2}(z,z')$ denotes the hyperbolic distance between $z,z'$)
\begin{equation}\label{Fla}
G_{\hh^2}(\la; z,z')=F_\la(d_{\hh^2}(z,z')), \quad \la\in D(0,1/4)\subset \cc
\end{equation}
so that $F_\la(r)$ is a holomorphic function of $\la$ for $r\in(0,\infty)$ satisfying 
\begin{equation}\label{Fs}
F_\la(r)\sim -\frac{1}{2\pi}\log(r) \textrm{ as }r\to 0, \quad F_0(r)=-\frac{1}{2\pi}\log(r)+ m(r^2)
\end{equation}
with $m$ being a smooth functions on $[0,\infty)$, and $G_{\hh^2}(s)$ satisfies   
\[ (\Delta_{\hh^2}-\la)G_{\hh^2}(\la; \cdot,z')=\delta_{z'}\]
where $\delta_{z'}$ denotes the Dirac mass at $z'$; in other words, 
$G_{\hh^2}(\la)$ is the Schwartz kernel of the operator $(\Delta_{\hh^2}-\la)^{-1}$ on $L^2(\hh^2)$. 
To obtain the Green function $G_g(x,x')$, it suffices to   know it for $g$ hyperbolic 
(ie. $g$ has constant Gaussian curvature $-1$) since for any other conformal metric $\hat{g}=e^\varphi g$, we have   that 
\begin{equation}\label{greenhatg}
\begin{gathered}
G_{\hat{g}}(x,x')=G_g(x,x')+ \alpha -u(x)-u(x'), \\
\textrm{ with }\alpha=\frac{1}{{\rm Vol}_{\hat{g}}(M)^2}\cjg G_g,1\otimes 1\cjd_{\hat{g}}, \quad u(x):=\frac{1}{{\rm Vol}_{\hat{g}}(M)}\int_{M}G_g(x,y){\rm dv}_{\hat{g}}(y).
\end{gathered}\end{equation}
This follows from an easy computation and the identity $\Delta_{\hat{g}}=e^{-\varphi}\Delta_g$.

Let us then assume that $g$ is hyperbolic. We have 
\begin{lemma}\label{greenneardiag}
If $g$ is a hyperbolic metric on the surface $M$, the Green function $G_g(x,x')$ for $\Delta_g$ has the following form
near the diagonal
\begin{equation}\label{greenfct2} 
G_g(x,x')= -\frac{1}{2\pi}\log(d_g(x,x'))+m_g(x,x')
\end{equation}
for some smooth function $m_g$ on $M\x M$. Near each point $x_0\in M$, there are coordinates $z\in B(0,1)\subset \cc$ so that $g=4|dz|^2/(1-|z|^2)^2$ and near $x_0$ 
\[G_g(z,z')= -\frac{1}{2\pi}\log |z-z'|+F(z,z')\]
with $F$ smooth. Finally, if $\hat{g}$ is any metric conformal to $g$, \eqref{greenfct2} holds with $\hat{g}$
replacing $g$ but with $m_{\hat{g}}$ continuous.
\end{lemma}
\begin{proof}
Near each point $x_0\in M$, there is an isometry from a geodesic ball $B_g(x_0,\eps)$ for $g$ to the hyperbolic ball $B_{\hh^2}(0,\eps)$ in $\hh^2$ ($0$ denotes the center of $\hh^2$ viewed as the unit disk), which provides in particular some local complex coordinates $z\in B(0,1)\subset \cc$ near $x_0$ so that 
$g=4|dz|^2/(1-|z|^2)^2$ in the ball $B_g(x_0,\eps)$. In these coordinates, 
\begin{equation}\label{logdg}
\log d_g(x,x')=\log d_{\hh^2}(z,z')=\log (2|z-z'|)+ L(z,z') \textrm{ with }L \textrm{ smooth}
\end{equation}
and $L(z,z)=0$ where $d_g$ denotes the distance for the metric $g$.
Near any given point $x'\in M$, one has
\begin{equation}\label{Fladg} 
(\Delta_g-\la)F_\la(d_{g}(\cdot,x'))-\delta_{x'}\in C^{\infty}(B_g(x',\eps))
\end{equation} 
where $B_g(x',\eps)$ is a geodesic ball of center $x'$ and radius $\eps>0$ small. Denote by $R_g(\la)=(\Delta_g-\la)^{-1}$ the resolvent of $\Delta_g$ 
for $\la\notin {\rm Sp}(\Delta_g)$. By the spectral theorem, at $\la=\la_0$ with $\la_0\in {\rm Sp}(\Delta_g)$ we have the Laurent expansion 
\[R_g(\la)=\frac{\Pi_{\la_0}}{\la-\la_0}+R_g(\la_0)+\mc{O}((\la -\la_0)), \quad \la\to \la_0\]
for some bounded operator $R_g(\la_0)$ and $\Pi_{\la_0}$ being the orthogonal projector on 
$\ker(\Delta_g-\la_0)$. 
Thus we obtain 
\[ (\Delta_g-\la_0)R_g(\la_0)={\rm Id}-{\Pi}_{\la_0}\] 
and by elliptic regularity and \eqref{Fladg}, the Schwartz kernel $G_g(\la;x,x')$ of 
$R_g(\la)$ for $\la\notin {\rm Sp}(\Delta_g)$ 
satisfies for $d_{g}(x,x')<\eps$ with $\eps>0$ small enough
\begin{equation}\label{greenfct1} 
G_g(\la; x,x')= F_\la(d_g(x,x'))+E_g(\la;x,x')
\end{equation}
with $E_g$ some smooth function on $M\x M$ depending meromorphically of $\la$. At $\la=0$ we 
deduce \eqref{greenfct2}.
The part about $\hat{g}$ just follows from \eqref{greenhatg} and the fact that 
$d_{\hat{g}}(x,x')=e^{\varphi(x)/2}d_g(x,x')+\mc{O}(d_g(x,x')^2)$ as $x'\to x$.
\end{proof}
The function $x\mapsto m_g(x,x)$ is often called the \emph{Robin constant} at $x$.
Notice that if we view the hyperbolic metric $g$ as an element representing a point of 
$\mc{T}_{\bf g}$ and if $\psi\in {\rm Mod}_{\bf g}$, then we have the modular invariance 
\begin{equation}\label{modulargreen} 
G_{\psi^*g}(\la;x,x')=G_{g}(\la;\psi(x),\psi(x')), \quad m_{\psi^*g}(x,x')=m_g(\psi(x),\psi(x')).
\end{equation}

We shall need to describe the Green function $G_g$ when the metric $g$ approaches the boundary of the compactification of $\mc{M}_{\bf g}$. It turns out that positive small eigenvalues of $\Delta_g$ appear sometime when $g$ approaches a point in $\pl \bbar{\mc{M}_{\bf g}}$: Schoen-Wolpert-Yau \cite{SWY} proved that there exist two positive constants $\alpha_1,\alpha_2$ depending only on the genus ${\bf g}$ so that the $n$-th positive eigenvalue $\la_n(g)$ of $\Delta_g$ satisfy
\[\begin{gathered} 
 \alpha_1L_n(M,g)\leq \la_n(g)\leq \alpha_2 L_n(M,g) \textrm{ if }n\leq 2{\bf g}-2, \quad \textrm{ and }
\alpha_1\leq \la_{2{\bf g}-1}\leq \alpha_2
\end{gathered}\]
where $L_n(M,g)$ is the minimum (over subpartitions) of the sums of lengths of simple geodesics in 
subpartitions of $M$ disconnecting $M$ into $n+1$ connected components. Each metric 
$g_0\in\pl \bbar{\mc{M}_{\bf g}}$ is in a stratum corresponding to a subpartition ${\rm SP}$ containing $n_p$ curves, with $n_s\leq n_p$ of these simple curves $\gamma_1,\dots, \gamma_{n_s}$ in the subpartition 
that disconnect the surface $M$ into $m+1$ connected components. 
There is $c_0>0$ depending on $g_0$ such that for all  $\delta>0$ small enough, there is a neighborhood $\bbar{U}_{g_0}\subset \bbar{\mc{M}_{\bf g}}$ of $g_0$ such that for all $g$ in the interior $U_{g_0}:=\bbar{U}_{g_0}\cap  \mc{M}_{\bf g}$, there is at most $m$ positive eigenvalues less than $\delta$ and all other eigenvalues are bigger than $c_0$. We call these eigenvalues the \emph{small eigenvalues} of $g$ near $g_0$.

\begin{proposition}\label{greenpinching}
Let $(M_0,g_0)\in \pl \bbar{\mc{M}_{\bf g}}$ where $M_0$ is a surface with nodes. 
For $\delta>0$ arbitrarily small, take $g$ in a sufficiently small open neighborhood $\bbar{U}_{g_0}$
of $g_0$ in  $\bbar{\mc{M}_{\bf g}}$ so that the small eigenvalues of $g$ in $U_{g_0}=\bbar{U}_{g_0}\cap \mc{M}_{\bf g}$ 
satisfy $\la_1(g)\leq  \dots\leq  \la_m(g)\leq \delta$.
The Green function $G_g$ restricted to 
$M_0$ can be written for $g\in U_{g_0}$ as 
\begin{equation}\label{greenfct3}
G_g(x,x')= \sum_{\la_j(g)\leq \delta} \frac{\Pi_{\la_j(g)}(x,x')}{\la_j(g)}+A_g(x,x')
\end{equation}
where $\Pi_{\la_j(g)}$ is the orthogonal projector on the corresponding eigenspace. In each compact set $\Omega$ of $M_0$, the map 
$(g,x,x')\mapsto A_g(x,x')$ is continuous on $\bbar{U}_{g_0}\x (\Omega^2_{\rm diag})$ if $\Omega^2_{\rm diag}:=(\Omega\x \Omega)\setminus {\rm diag}$ and, near the diagonal of $\Omega\x\Omega$, one has
\[ A_g(x,x')=-\frac{1}{2\pi}\log(d_g(x,x'))+B_g(x,x')\]
with $(g,x,x')\mapsto B_g(x,x')\in C^0(\bbar{U}_{g_0}\x \Omega\x \Omega)$.  
The Schwartz kernel $\sum_{j=1}^m\Pi_{\la_j(g)}(x,x')$ extends continuously to 
$(g,x,x')\in \bbar{U}_{g_0}\x \Omega\x \Omega$ with value at $g=g'\in \pl \bbar{U}_{g_0}$ 
the orthogonal projector $\Pi_0(g';x,x')$ 
onto $\ker_{L^2}\Delta_{g'}$.
\end{proposition}
\begin{proof} 
After possibly splitting $\Omega$ in smaller pieces, we can assume that the radius of 
injectivity of all $g\in U_{g_0}$ on $\Omega$ is bounded below by some uniform $\alpha>0$.
Using the residue formula applied to $R_g(\la)/\la$ in a disk $D(0,\delta)$ of radius $\delta$ centered at $\la=0$, one has 
\begin{equation}\label{contourrep} 
R_g(0)-\sum_{\la_j(g)\leq \delta} \frac{\Pi_{\la_j(g)}}{\la_j(g)}=\frac{1}{2\pi i}\int_{\pl D(0,\delta)} \frac{R_g(\la)}{\la} d\la
\end{equation}
and we denote by $A_g(x,x')$ the Schwartz kernel of $\frac{1}{2\pi i}\int_{\pl D(0,\delta)} \frac{R_g(\la)}{\la} d\la$. 
Let $\Omega'\subset M_0$ be a small neighborhood $\Omega'$ of $\Omega$ so that the radius of injectivity of each $g\in U_{g_0}$ is bounded below by $\alpha/2$. 
Let $L_g(\la)$ be the operator on 
$\Omega'$ with Schwartz kernel 
\[F_{\la}(d_g(x,x'))\] 
where $F_\la$ is the function of \eqref{Fla}. Take 
$\chi,\til{\chi}\in C_c^\infty(M_0)$ equal to $1$ on $\Omega$ 
but with support contained in $\Omega'$, and such that $\til{\chi}\chi=\chi$.
Then on $\Omega'$ and on $M_0$, we have 
\begin{equation}\label{eqresolvent} 
(\Delta_g-\la)\til{\chi}L_g(\la) \chi= \chi + [\Delta_g,\til{\chi}]L_g(\la)\chi .
\end{equation} 
Multiplying \eqref{eqresolvent} by $\chi R_g(\la)$ on the left, we get 
\[\chi R_g(\la)\chi=\chi L_g(\la)\chi-\chi R_g(\la)[\Delta_g,\til{\chi}]L_g(\la)\chi.\] 
The operators $[\Delta_g,\til{\chi}]L_g(\la)\chi$  have smooth kernel 
(we use that $[\Delta_g,\til{\chi}]=0$ on ${\rm supp}(\chi)$), and extends continuously 
to $g\in\bbar{U}_{g_0}$ since $g$ extends continuously as a smooth metric to $\Omega$ and 
$d_g$ on $\Omega\x \Omega$ as well. Now we use the fact that for $\la \in \pl D(0,\delta)$, $g\mapsto R_g(\la)$ extends continuously to $\bbar{U}_{g_0}$
as bounded operators $H^k_{\rm comp}(M_0)\to H^k_{\rm loc}(M_0)$ for all $k\geq 0$
by a result of Schulze \cite{Sc}: this implies that $\chi R_g(\la)[\Delta_g,\til{\chi}]L_g(\la)\chi$ extend continuously 
in $g\in \bbar{U}_{g_0}$ as a family of bounded operators $H^{-k}(M_0)\to H^k_{\rm comp}(M_0)$ for all $k\geq 0$, since $[\Delta_g,\til{\chi}]L_g(\la)\chi$ maps $H^{-k}(M_0)\to H^k(M_0)$ uniformly in $g\in U_{g_0}$.
Thus the Schwartz kernels of the operators $[\Delta_g,\til{\chi}]L_g(\la)\chi$ extend as 
a uniform family of continuous Schwartz kernels (when $g\in U_{g_0}$). 
We then deduce that 
\[\frac{1}{2\pi i}\int_{\pl D(0,\delta)} \frac{\chi R_g(\la) \chi}{\la} d\la= \frac{1}{2\pi i}\int_{\pl D(0,\delta)} \frac{\chi L_g(\la)\chi}{\la} d\la + B'_g\]
where $B'_g$ is a family of operators, with Schwartz kernel $B'_g(x,x')$ continuous as a function of $(g,x,x')\in \bbar{U}_{g_0}\x \Omega\x\Omega$. Next, since by Cauchy formula $\frac{1}{2\pi i}\int_{\pl D(0,\delta)}\frac{F_{\la}(z)}{\la}d\la=F_0(z)$, we deduce that 
\[\Big(\frac{1}{2\pi i}\int_{\pl D(0,\delta)} \frac{\chi L_g(\la)\chi}{\la} d\la\Big)(x,x') = 
\chi(x)\chi(x') F_0(d_g(x,x'))\]
and this Schwartz kernel has the desired property by using \eqref{Fs}. 
This ends the proof of \eqref{greenfct3}. The proof of the fact that $\sum_{j=1}^m\Pi_{\la_j(g)}(x,x')$ converge to the projector onto the kernel of $g'$ as $g\to g'\in \pl\bbar{U}_{g_0}$ is essentially the same as what we did (and even simpler) by applying the residue formula to $R_g(\la)$ in $D(0,\delta)$ instead of $R_g(\la)/\la$. The 
convergence in $C^0$ norm is clear since convergence in $L^2$ of $\sum_{j=1}^m\Pi_{\la_j(g)}$ implies convergence in $C^\infty$ on $\Omega$  by elliptic regularity.
\end{proof}

\subsection{Small eigenvalues and associated eigenvectors} \label{sectionsmall}
In this section, we recall the asymptotics of the small positive eigenvalues $\la_1(g)\leq \dots\leq \la_m(g)$ as $g\in \mc{M}_{\bf g}$ approaches an element 
$g_0\in \bbar{\mc{M}_{\bf g}}$ by following Burger \cite{Bu1, Bu2} and we will see that the proof of \cite{Bu2} also gives an approximation of the projectors $\Pi_{\la_{j}(g)}$. Let $(M_0,g_0)$ be a surface with nodes, with corresponding subpartition of the closed Riemann surface $M$ given by simple curves $\gamma_1,\dots, \gamma_{n_p}$ and $\gamma_1,\dots ,\gamma_{n_s}$ (with $n_s\leq n_p$) are disconnecting $M$ into $m+1$ connected components $S_1,\dots,S_{m+1}$. 
For all $\delta>0$ small enough, there is a small neighborhood $\bbar{U}_{g_0}\subset \bbar{\mc{M}_{\bf g}}$ of $g_0$ in $\bbar{\mc{M}_{\bf g}}$ so that  for each $g\in U_{g_0}=\bbar{U}_{g_0}\cap \mc{M}_{\bf g}$ there are $m$ small eigenvalues $0<\la_1(g)\leq \dots\leq \la_m(g)\leq \delta$  and all others are larger than a constant $c_0>0$ depending only on $g_0$. Each metric $g\in U_{g_0}$ has a unique simple closed geodesic $\gamma_j(g)$ homotopic to $\gamma_j$ for $j\leq n_p$, with length $\ell_j(g)\leq c_1\delta$, while all other primitive closed geodesics have length bigger than $c_2>0$, where  $c_1,c_2$ are constants depending only on $g_0$.
The geodesics $\gamma_j(g)$ decompose $M$ into $m+1$ connected components $S_1(g),\dots, S_{m+1}(g)$ respectively homeomorphic to $S_1,\dots,S_{m+1}$. 
Define the length 
$L_{ij}(g):=\sum_{k\in E_{ij}}\ell_k(g)$ where $E_{ij}=\{ 1\leq k\leq n_s;  \gamma_k\in \pl S_i\cap \pl S_j\}$. Let 
$||\cdot ||_g$ be the norm on $\R^{m+1}$ given by 
\begin{equation}\label{normgraph} 
||a||_{g}^2=\sum_{j=1}^{m+1}{\rm Vol}_{g}(S_j(g))a_j^2, \quad \textrm{ with }a=(a_1,\dots,a_{m+1})
\end{equation}
and let $Q_g$ be the quadratic form on $\R^{m+1}$ given by 
\begin{equation}\label{Qg}
Q_g(a)=\sum_{1\leq i,j\leq m+1}(a_i-a_j)^2L_{ij}(g).
\end{equation}
Notice that ${\rm Vol}_{g}(S_j(g))$ are positive constants depending only on the topology of $S_j$ (and not on $g$) by Gauss-Bonnet theorem.
Then Burger \cite{Bu2} showed the following estimate: 
\begin{theorem}[Burger]\label{thburger}
If $\nu_1(g)\leq \dots\leq \nu_m(g)$ are the positive eigenvalues of $Q_g$ with respect to the norm $||\cdot||_g$ on $\R^{m+1}$, then there is $C>0$ such that for all $g\in U_{g_0}$ and each $1\leq j\leq m$ 
\[  \frac{\nu_j(g)}{\pi}(1-C\delta^{\demi})\leq \la_j(g)\leq \frac{\nu_j(g)}{\pi}(1+C\delta|\log \delta|).\]
\end{theorem}
Each simple small geodesic $\gamma_j(g)$ of $g$ (homotopic to $\gamma_j$) has a collar neighborhood
\begin{equation}\label{collarCj}
\mc{C}_{j}(g)=\{ x\in M;\,  \sinh (d_{g}(x,\gamma_j(g)))\leq 1/\sinh(\ell_j(g))\}
\end{equation}
and these collars are disjoints one from the other. The set $M\setminus \cup_{j\leq n_s}\mc{C}_{j}(g)$ has $m+1$ connected components $S'_1,\dots, S'_{m+1}$ respectively homeomorphic to $S_1(g),\dots, S_{m+1}(g)$. 
One can define a map 
\begin{equation}\label{deffa} 
a\in \R^{m+1}\mapsto f_a\in H^1(M)
\end{equation}
by setting $f_a(x)=a_j$ if $x\in S'_j$ and $f_a$ being the unique harmonic function in $\mc{C}_{j}(g)$ 
so that $f_a$ is continuous on $M$. In \cite{Bu2}, Burger proved the following 
\begin{lemma}[Burger]\label{lemmaburger}
There is $C>0$ such that for all $a\in \R^{m+1}$ and all $g\in U_{g_0}$
\[ \begin{gathered}
\frac{1}{\pi}Q_g(a)\leq ||df_a||^2_{L^2(M,g)}\leq \frac{1}{\pi}Q_g(a)(1+C\delta), \quad ||a||^2_g(1-C\delta|\log \delta|)\leq ||f_a||^2_{L^2(M,g)}\leq ||a||_g^2.
\end{gathered}\]
\end{lemma}

An estimate for $\la_j(g)$ in terms of the pinched geodesics $\ell_k(g)$ is given by Schoen-Wolpert-Yau \cite{SWY}: let $D$ be an $n$-disconnect, i.e a collection of closed simple geodesics $\gamma_1(g),\dots, \gamma_{n_s}(g)$ with respective lengths $\ell_1(g),\dots,\ell_{n_s}(g)$ disconnecting $M$ into $n$ connected components, and define
$L_{n}(D,g):= \sum_{j=1}^{n_s}\ell_j(g)$.
We set 
\[ L_n(M,g):= \min_{D\in \mc{D}_n}L_n(D,g) \]
where $\mc{D}_n$ is the set of all $n$-disconnects of $M$. Then, ordering the eigenavlues by increasing order, it is proved in \cite{SWY} that there is $C>1$ depending only on the genus of $M$ such that for each $n\leq 2{\bf g}-2$
\[
C^{-1}L_n(M,g)\leq \la_n(g)\leq CL_n(M,g).
\]
As a consequence, in a  neighborhood $U_{g_0}\subset \mc{M}_{\bf g}$ of a metric $g_0\in \pl \mc{M}_{\bf g}$,  we have the rough estimate for all $j\leq m$ and $g\in U_{g_0}$
\begin{equation}\label{estimeevpSWY} 
\la_j(g)\geq C^{-1}\ell_{j}(g)
\end{equation}
where $m+1$ is the number of connected components of the surface with cusps 
$(M_0,g_0)$, $n_s$ is the number of pinched geodesics disconnecting the surface
 and $\ell_1(g)\leq \ell_2(g)\leq \dots \leq \ell_{n_s}(g)$ are the lengths of these separating geodesics ranked by increasing order.

Below, we take the convention that we repeat each eigenvalue according to its multiplicity, thus $\la_j(g)$ can be equal to $\la_{j+1}(g)$, and similarly for the $\nu_j(g)$.
\begin{lemma}\label{approximation}
For each $g\in U_{g_0}$, let $v_0=(4\pi({\bf g}-1))^{-1}$ and $v_1,\dots,v_m \in \R^{m+1}$ so that $(v_i)_{i=0,\dots,m}$ is an orthonormal basis of eigenvectors for $Q_g$ with $v_i$ associated to $\nu_i(g)$. There is $C>0$ and $L\in\nn$ such that for all $g\in U_{g_0}$, there exists an orthonormal basis $\varphi_1,\dots,\varphi_{m}$ of $\oplus_{j=1}^m \ker(\Delta_g-\la_j(g))$ satisfying
\[ ||f_{v_j}-\varphi_j||_{L^2(M,g)}\leq C\delta^{\frac{1}{L}}, \textrm{ and }
\quad  \sum_{\la_j(g)\leq \delta} \frac{\Pi_{\la_j(g)}(x,x')}{\la_j(g)}=\sum_{j=1}^m \frac{f_{v_j}(x)f_{v_j}(x')}{\nu_j(g)} +\mc{O}\Big(\frac{\delta^{\frac{1}{L}}}{\nu_{1}(g)}\Big)\]
where the error term is in $L^\infty$ norm on compact sets disjoints from  $\cup_j\mc{C}_j(g)$.
\end{lemma}
\begin{proof} To simplify notations, we denote by $f_j$ the function $f_{v_j}$. We construct the basis 
$\varphi_j$ by an inductive process.  Let $(\phi_j)_{j\in\nn_0}$ be an orthonormal basis of $L^2(M,g)$ of eigenvectors for $\Delta_g$, ie. $\Delta_g\phi_j=\la_j(g)\phi_j$. By Lemma \ref{lemmaburger}, we have for 
$k\leq m$
\[||df_k||^2_{L^2}=\sum_{j=1}^{\infty}\la_{j}(g)\cjg f_k,\phi_j\cjd^2
=\frac{Q_g(v_k)}{\pi}(1+\mc{O}(\delta))=\lambda_k(g)||f_k||^2_{L^2}(1+\mc{O}(\delta^{1/2}))\]
and by Theorem \ref{thburger}, this gives for each $k=1,\dots, m$
\begin{equation}\label{eqlak}
\sum_{j=1}^{\infty}\Big(\frac{\la_{j}(g)}{\la_k(g)}-1\Big)\cjg f_k,\phi_j\cjd^2 =\mc{O}(\delta^{\demi} ||f_k||^2_{L^2}).
\end{equation}
If $|\frac{\la_2}{\la_1}-1|> \delta^{\frac{1}{4}}$, 
we set $\varphi_1:=\phi_1$ and $i_1=1$. By \eqref{eqlak} with $k=1$, we get  
$\sum_{j=2}^{\infty}\cjg f_k,\phi_j\cjd^2=\mc{O}(\delta^{\frac{1}{4}})$ and thus 
$f_1=\pm \phi_1+\mc{O}_{L^2}(\delta^{\frac{1}{8}})$. Since $\cjg f_i,f_j\cjd_{L^2}=\mc{O}(\delta|\log \delta|)$ for 
$i\not=j$ by Lemma \ref{lemmaburger}, we get $\cjg f_i,\varphi_1\cjd=\mc{O}(\delta^{\frac{1}{8}})$ for all $i>1$. If $|\frac{\la_2}{\la_1}-1|\leq \delta^{\frac{1}{4}}$, we let $i_1\geq 2$ be the 
smallest integer such that for each $i\leq i_1$,
$|\frac{\la_{i}}{\la_{i-1}}-1|\leq \delta^{\frac{1}{2^i}}$ and 
$|\frac{\la_{i_1+1}}{\la_{i_1}}-1|>\delta^{\frac{1}{2^{i_1+1}}}$, clearly $i_1\leq m$ since there are $m$ small eigenvalues. 
We define 
\[\varphi_1 = \frac{\sum_{j=1}^{i_1} \cjg f_1,\phi_j\cjd\phi_j}{||\sum_{j=1}^{i_1} \cjg f_1,\phi_j\cjd\phi_j||}, 
 \textrm{ and }  \til{\varphi}_i = \frac{\sum_{j=1}^{i_1} \cjg f_i,\phi_j\cjd\phi_j}{||\sum_{j=1}^{i_1} \cjg f_i,\phi_j\cjd\phi_j||} 
\textrm{ for }1\leq i\leq i_1.\]
Then we construct $\varphi_2,\dots,\varphi_{i_1}$ by the Gram-Schmidt orthonormalization process 
from $\til{\varphi}_2,\dots,\til{\varphi}_{i_1}$. Since $|\frac{\la_{i_1+1}}{\la_{i_1}}-1|>\delta^{\frac{1}{2^{i_1+1}}}$ 
and $|\frac{\la_{i_1}}{\la_1}-1|=\mc{O}(\delta^{\frac{1}{2^{i_1}}})$, \eqref{eqlak} tells us that for each $i\leq i_1$,
\[ f_i= \sum_{j=1}^{i_1} \cjg f_i,\phi_j\cjd\phi_j+ \mc{O}_{L^2}(\delta^{\frac{1}{2^{i_1+2}}})\]
and thus $ \til{\varphi}_i= f_i+\mc{O}_{L^2}(\delta^{\frac{1}{2^{i_1+2}}})$ for $i=1,\dots,i_1$. Since $\cjg f_i,f_j\cjd_{L^2}=\delta_{ij}+\mc{O}(\delta|\log \delta|)$ by Lemma \ref{lemmaburger}, we deduce that for $i=1,\dots,i_1$
\[ \varphi_i=f_i+\mc{O}_{L^2}(\delta^{\frac{1}{2^{i_1+2}}}).\]
Now we prove the induction process in a way  similar to the first step. 
Suppose we have constructed an orthornormal basis $\varphi_1,\dots,\varphi_\ell$ of 
$\oplus_{j=1}^{\ell}\R\phi_j$ so that $\varphi_j=f_j+\mc{O}_{L^2}(\delta^{\frac{1}{L}})$ for some $L\in \nn$ and $\ell< m$. Notice that $\cjg f_{k},\phi_j\cjd=\mc{O}(\delta^{\frac{1}{L}})$ for all $k\geq \ell+1$ and $j\leq \ell$ by the induction assumption. Then if $|\frac{\la_{\ell+1}}{\la_{\ell}}-1|>\delta^{\frac{1}{L}}$, \eqref{eqlak} with $k=\ell+1$ gives $\sum_{j=\ell+2}^\infty \cjg f_{\ell+1},\phi_j\cjd^2=\mc{O}(\delta^{\frac{1}{L}})$, thus if we set $i_{\ell+1}=\ell+1$ and
\[\varphi_{\ell+1}=\frac{\phi_{\ell+1}- \sum_{j=1}^\ell \cjg \phi_{\ell+1},\varphi_j\cjd\varphi_j}
{||\phi_{\ell+1}- \sum_{j=1}^\ell \cjg \phi_{\ell+1},\varphi_j\cjd\varphi_j||}\]
we get $\varphi_{\ell+1}=f_{\ell+1}+\mc{O}_{L^2}(\delta^{\frac{1}{2L}})$ and we have increased the induction step by $1$.
If $|\frac{\la_{\ell+1}}{\la_{\ell}}-1|\leq \delta^{\frac{1}{L}}$, we let $i_{\ell+1}\leq m$ be the 
smallest integer such that for all $i=\ell+1,\dots,i_{\ell+1}$, 
$|\frac{\la_i}{\la_{i-1}}-1|\leq \delta^{\frac{1}{L2^{i-\ell-1}}}$ and
$|\frac{\la_{i_{\ell+1}+1}}{\la_{i_{\ell+1}}}-1|>\delta^{\frac{1}{L2^{i_{\ell+1}-\ell}}}$, and we will construct $\varphi_{\ell+1},\dots,\varphi_{i_{\ell+1}}$. Let $L'=L2^{i_{\ell+1}-\ell}$ and define
\[\varphi_{\ell+1} = \frac{\sum_{j=\ell+1}^{i_{\ell+1}} \cjg f_{\ell+1},
\phi_j\cjd\phi_j}{||\sum_{j=\ell+1}^{i_{\ell+1}} \cjg f_{\ell+1},\phi_j\cjd\phi_j||}, 
 \textrm{ and }  \til{\varphi}_i = \frac{\sum_{j={\ell+1}}^{i_{\ell+1}} \cjg f_i,\phi_j\cjd\phi_j}{||\sum_{j=\ell+1}^{i_{\ell+1}} \cjg f_i,\phi_j\cjd\phi_j||} 
\textrm{ for }\ell+1 \leq i\leq i_{\ell+1}.\]
Then we construct $\varphi_{\ell+2},\dots,\varphi_{i_{\ell+1}}$ by the Gram-Schmidt orthonormalization process 
from $\til{\varphi}_{\ell +2},\dots,\til{\varphi}_{i_{\ell+1}}$. By induction assumption and $|\frac{\la_{i_{\ell+1}+1}}{\la_{i_{\ell+1}}}-1|>\delta^{\frac{1}{L'}}$,  
\eqref{eqlak} tells us that for each $i=\ell+1,\dots,i_{\ell+1}$,
\[ f_i= \sum_{j=\ell+1}^{i_{\ell+1}} \cjg f_i,\phi_j\cjd\phi_j+ \mc{O}_{L^2}(\delta^{\frac{1}{2L'}})\]
and thus $ \til{\varphi}_i= f_i+\mc{O}_{L^2}(\delta^{\frac{1}{2L'}})$ for $i=\ell+1,\dots,i_{\ell+1}$.
Since $\cjg f_i,f_j\cjd_{L^2}=\delta_{ij}+\mc{O}(\delta|\log \delta|)$ by Lemma \ref{lemmaburger}, we deduce that for $i=\ell+1,\dots,i_{\ell+1}$
\[ \varphi_i=f_i+\mc{O}_{L^2}(\delta^{\frac{1}{2L'}})\]
and we have increased the induction step by $i_{\ell+1}-(\ell+1)\geq 1$. 
This inductive construction produces a sequence of integers
$j_0=1,j_1=i_1, j_2=i_{i_1+1},\dots, j_N=m$ and  
$N$ associated blocks $E_1,\dots E_N$, 
with $E_k=\{\varphi_{j_{k}},\dots ,\varphi_{j_{k+1}}\}$ where the span of elements in $E_k$ is the span 
of $\{\phi_{j_k},\dots, \phi_{j_{k+1}}\}$.
By construction we have 
\[\begin{split}
 \sum_{\la_j(g)\leq \delta}\frac{\Pi_{\la_j}(x,x')}{\la_j(g)}=& 
\sum_{k=0}^N\sum_{j=j_{k}}^{j_{k+1}}\frac{\varphi_j(x)\varphi_j(x')}{\la_j(g)}+\mc{O}(\delta^{1/L})\\
&=  \sum_{k=0}^N\sum_{j=j_{k}}^{j_{k+1}}\frac{\varphi_j(x)\varphi_j(x')}{\nu_j(g)}+\mc{O}\Big(\frac{\delta^{\frac{1}{L}}}{\nu_1(g)}\Big)\\
&= \sum_{j=1}^m \frac{f_{j}(x)f_{j}(x')}{\nu_j(g)}+\mc{O}\Big(\frac{\delta^{\frac{1}{L}}}{\nu_1(g)}\Big)\end{split}
\]
for some $L\in \N$ large. Here we notice that the $\mc{O}\Big(\frac{\delta^{\frac{1}{L}}}{\nu_1(g)}\Big)$ can be taken in $L^\infty$ norm since $L^2$ norms on eigenfunctions give directly uniform $L^\infty$ norms on compact sets outside the collars $\mc{C}_j(g)$.
\end{proof}

\subsection{Example: the case of genus $2$}
For pedagogical purpose, let us discuss more particularly the case of genus ${\bf g}=2$. In this case there can only be $3$ simple curves in a partition and the maximal number of connected components separated by these curves is $2$: either 
$1$ curve separates $M$ into two surfaces of genus $1$ with $1$ boundary component (Case 1) 
or two hyperbolic pants with $3$ boundary components (Case 2). 
Consequently, the number of eigenvalues approaching $0$ when we approach $\pl \bbar{\mc{M}}_2$ is  $m\in\{1,2\}$ (including  the eigenvalue $\la=0$), we call them $\la_0=0$ and $\la_1(g)>0$ when $m=2$. 

In Case 1, take any partition ${\rm SP}_1$ by $\gamma_1,\gamma_2,\gamma_3$ with $\gamma_1$ being the only separating curve, and denote by $\ell_j(g)$ the length of the geodesic for $g$ freely homotopic to $\gamma_j$.  
We have 
\begin{equation}\label{la1case1}
\la_1(g)\sim c_1\ell_1(g), \textrm{ as }
\ell_1(g)\to 0\textrm{ with }g\in U({\rm SP}_1)
\end{equation}
where $c_1>0$ depending only on ${\bf g}=2$.
\begin{figure}\label{figure}
\centering
\def\svgwidth{20em}
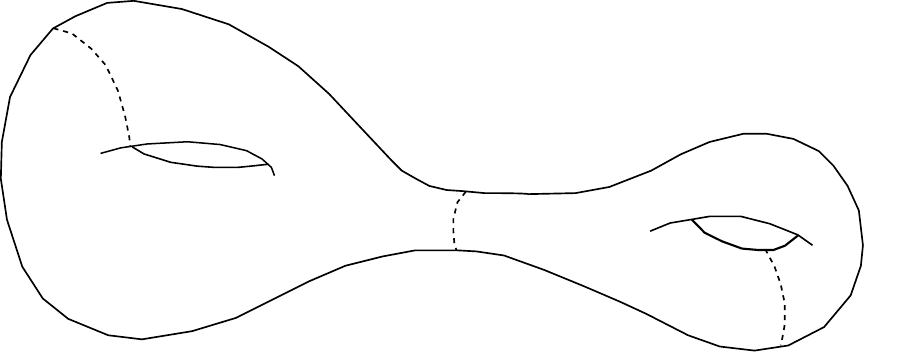
\def\svgwidth{15em}
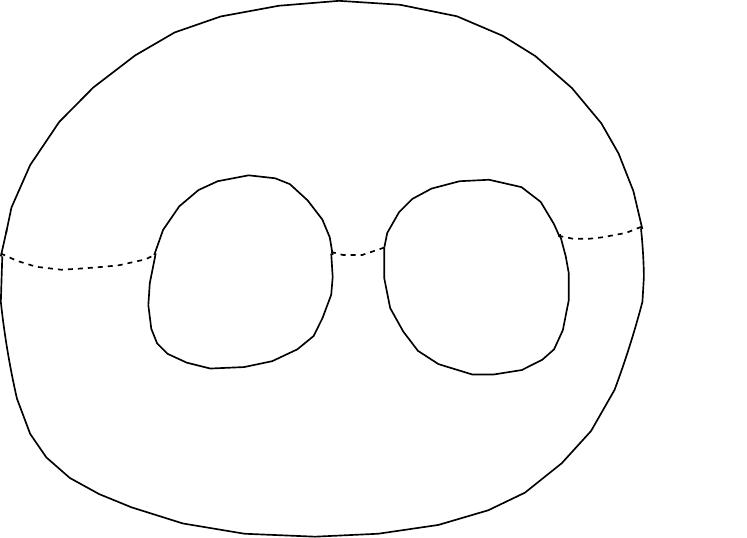
\caption{On the left: Case 1. On the right: Case 2} 
\label{fig:FD}
\end{figure}
 
In Case 2, take ${\rm SP}_2$ any partition where $\gamma_1,\gamma_2,\gamma_3$ are separating 
simple curves and, if $\ell_j(g)$ is the length of the geodesic for $g$ homotopic to $\gamma_j,$ then by \cite{Bu1}, 
\begin{equation}\label{la1case2}
\la_1(g)\sim c_2(\ell_1(g)+\ell_2(g)+\ell_3(g)), \textrm{ as } (\ell_1(g),\ell_2(g),\ell_3(g))\to 0 \textrm{ with  }g\in U({\rm SP}_2)
\end{equation}
for some $c_2>0$ depending only on ${\bf g}=2$.

In both cases, if $g\to g_0$ with $(M_0,g_0)$ a surface with nodes, $(M_0,g_0)$ decomposes into two finite volume hyperbolic surfaces $S_1$ and $S_2$, with volume $2\pi$ (by Gauss-Bonnet), 
and  by Proposition \ref{greenpinching}
\begin{equation}\label{convergencepila1}
\Pi_{\la_1(g)}(x,x')\to \frac{1}{4\pi}(\indic_{S_1}(x)-\indic_{S_2}(x))(\indic_{S_1}(x')-\indic_{S_2}(x')) \textrm{ as }g\to g_0 
\end{equation}
uniformly in $(x,x')$ on compact sets of $M_0\x M_0$.
 
 \subsection{Green's function near pinched geodesics}
For later purpose, we will need a more detailed description of the Green's function $G_g$ than Proposition \ref{greenpinching}, in particular we shall need to know the behaviour of $G_g$ near the pinched geodesics. Some analysis of $G_g$ in such cases were done by Ji \cite{Ji}, while
a recent work of Albin-Rochon-Sher \cite{ARS} gives a parametrix for $\Delta_g-\la$ when $\la$ is 
near $0$ and when there is one pinched geodesic (or the geodesics are pinched at the same speed).
The recent work of Melrose-Zhu \cite{MeZh} also gives a parametrix but for a slightly different Green function. Here, in contrast, we need to know the behaviour in all possible directions of approach of the boundary of $\mc{M}_{\bf g}$ and our estimates below are designed to be applied later for the study of the Gaussian multiplicative chaos measure in the cusp. 

Let $(M_0,g_0)$ be a surface with nodes viewed as 
a surface with pairs of hyperbolic cusps, and let 
$\bbar{U}_{g_0}$ be a local neighborhood of $g_0$ in $\bbar{\mc{M}_{{\bf g}}}$ made of hyperbolic metrics $g_{s,t}$ as explained in Section \ref{Sec:Teichmuller}, and denote $U_{g_0}=\mc{M}_{{\bf g}}\cap \bbar{U}_{g_0}$. 
For convenience, we remove the parameters $(s,t)$ and just write $g$ for $g_{s,t}$.
As in Proposition \ref{greenpinching}, we set
\begin{equation}\label{Ag} 
A_g(x,x'):=G_g(x,x')-\sum_{j=1}^m\frac{\Pi_{\la_j}(x,x')}{\la_j}=\frac{1}{2\pi i}\int_{\pl D(0,\eps)}\frac{R_g(\la;x,x')}{\la}d\la\end{equation}
where $\la_j=\la_j(g)$ are the small eigenvalues tending to $0$ as $g$ approaches the boundary of moduli space, $D(0,\eps)$ is a small disk containing these eigenvalues (and only these ones) and $R_g(\la;x,x')$ is the integral kernel of the resolvent $R_g(\la)=(\Delta_g-\la)^{-1}$ of $\Delta_g$, with $\la\in \cc$.

The hyperbolic surface $(M,g)$ decomposes  into $M= S(g)\cup_{i=1}^{n_p} \mc{C}_{j}(g)$ where $\mc{C}_j(g)$ are the collars isometric to $[-1,1]_\rho \x (\R/\zz)_\theta$ close to a given curve $\gamma_j$, where the metric $g$ is given (in geodesic normal coordinates to the geodesic $\gamma_j(g)$ homotopic to $\gamma_j$) by 
\begin{equation}\label{collarrho} 
g_j=\frac{d\rho^2}{\rho^2+\ell_j^2}+ (\rho^2+\ell_j^2)d\theta^2,
\end{equation}
where $\ell_j=\ell_j(g)$ is the length of the geodesic $\gamma_j(g)$ and $S(g)$ is a compact manifold with boundary contained in a fixed (independent of $g$) compact set of $M_0$. 
This is also isometric to (by the coordinates change $\rho=\ell_j\sinh(t)$)
\[  [-d_j,d_j]_t\x (\R/\zz)_\theta,  \quad g_j=dt^2+\ell_j^2\cosh^2(t)d\theta^2, \quad \sinh(d_j)=1/\ell_j.\] 
The complete hyperbolic cylinder $\cjg z\mapsto e^{\ell_j}z\cjd \backslash\hh^2$ is isometric to
\[ \mc{F}_j:=\Big( \R_\rho\x (\R/\zz)_\theta, \,\,\, g_j=\frac{d\rho^2}{\rho^2+\ell_j^2}+ (\rho^2+\ell_j^2)d\theta^2\Big).\] 
Notice that as $\ell_j\to 0$, the Riemannian manifold $\mc{F}_j\setminus \{\rho=0\}$ 
converges smoothly to two disconnected surfaces $(0,\infty)_\rho \x (\R/\zz)_\theta$ with metric $d\rho^2/\rho^2+\rho^2d\theta^2$, which are isometric to two disjoint elementary quotients $\hh^2/\cjg z\mapsto z+1\cjd$. 
The Laplacian $\Delta_{g_j}$ is self-adjoint with spectrum $[1/4,\infty)$ (its self-adjoint realisation is through Friedrichs extension on $C_c^\infty(\mc{F}_j)$). It is invertible on $L^2(\mc{F}_j)$ and we can consider its resolvent $R_{g_j}(\la):L^2(\mc{F}_j)\to L^2(\mc{F}_j)$ which is holomorphic for $\la\notin [1/4,\infty)$. 
This is studied in details for example in \cite[Prop. 5.2]{Bo} or \cite[Appendix]{GuZw}: writing $\la=s(1-s)$ for $s$ close to $1$, we have
\[ R_{g_j}(\la; \rho,\theta,\rho',\theta')= \sum_{k\in\zz} u_k(s; \rho,\rho')e^{2\pi ik(\theta-\theta')}\]
for some explicit functions $u_k$ analytic in $s$ for $s$ close to $1$. We denote by
$G_{g_j}$ the Green function corresponding to $\la=0$ (i.e. $s=1$). We give a more explicit bound at $\la=0$ in the following
\begin{lemma}\label{G_j} 
For $\ell_j\leq \min(|\rho|,|\rho'|)\leq \max(|\rho|,|\rho'|)\leq 1$ with $\rho\rho'>0$,
the Green function $G_{g_j}$ for the cylinder satisfies 
\begin{equation}\label{boundGj}
\begin{split}
 G_{g_j}(\rho,\theta,\rho',\theta')= &-\frac{1}{2\pi}
\log\Big|1-e^{-\frac{2\pi}{\ell_j}|\arctan(\tfrac{\rho}{\ell_j})-\arctan(\tfrac{\rho'}{\ell_j})|+2\pi i(\theta-\theta')}\Big|
\\
 & +\frac{1}{\ell_j}\min (F(\tfrac{|\rho|}{\ell_j}),F(\tfrac{|\rho'|}{\ell_j}))-\frac{1}{\pi \ell_j}F(\tfrac{|\rho|}{\ell_j})F(\tfrac{|\rho'|}{\ell_j})+\mc{O}(1)
\end{split}
\end{equation}
where $F(x):=\int_{x}^{\infty}\tfrac{du}{1+u^2}$ and the remainder is uniform.
If $\ell_j\leq \min(|\rho|,|\rho'|)\leq \max(|\rho|,|\rho'|)\leq 1$ with $\rho'\rho<0$, then 
\begin{equation}\label{boundsigndiff}
G_{g_j}(\rho,\theta,\rho',\theta')= \frac{1}{\pi\ell_j}F(\tfrac{|\rho|}{\ell_j})F(\tfrac{|\rho'|}{\ell_j})+
\mc{O}(e^{-\frac{\pi}{2\ell_j}}).
\end{equation}
If $|\rho|\in [1/2,1]$, $|\rho'|\leq 1$ and $|\rho-\rho'|>\delta$ for some $\delta>0$, then there is $C$ depending only on $\delta$ so that 
\begin{equation}\label{boundderivee}
|\pl_\rho G_{g_j}(\rho,\theta,\rho',\theta')|\leq C.
\end{equation}
If $\chi\in C_c^\infty(-1,1)$ and $\la \in [0,1)$, then we have the pointwise estimate
\begin{equation}\label{Ggjonconstant} 
\int_{\R/\Z}\int_{-1}^1G_{g_j}(\rho,\theta,\rho',\theta')\chi(\rho')|\rho'|^{-\la}d\rho'd\theta'\leq 2||\chi||_{L^\infty}
\Big(\frac{|\rho|^{-\la}}{1-\la}+\frac{|\rho|^{-\la}-1}{\la}\Big)
\end{equation}
 where by convention $\frac{|\rho|^{-\la}-1}{\la}=|\log |\rho||$ when $\la=0$.
Finally, the Robin mass of $G_{g_j}$ satisfies 
\begin{equation}\label{boundrobin} 
\Big|m_{g_j}(\rho,\theta)-\frac{1}{\pi\ell_j}\Big(\tfrac{\pi^2}{4}-\arctan(\tfrac{\rho}{\ell_j})^2\Big)-\frac{1}{2\pi}\log(\sqrt{\rho^2+\ell_j^2})\Big|\leq C.\end{equation}
\end{lemma}  
\begin{proof} 
Let $L$ be the operator $L=(\rho^2+\ell_j^2)\Delta_{g_j}$ acting 
on $T:=\R_\rho\x(\R/\Z)_\theta$ with the measure $(\rho^2+\ell_j^2)^{-1}d\rho d\theta$, 
it is symmetric on $C_c^\infty(T)$. Changing coordinates to $t=\ell_j^{-1}\arctan(\rho/\ell_j)$, 
$L$ becomes the operator 
 \[ L=-\pl_t^2-\pl_\theta^2 \textrm{ on }(-\tfrac{\pi}{2\ell_j},\tfrac{\pi}{2\ell_j})_t\x (\R/\Z)_\theta\]
 with the measure $dtd\theta$. It is not self-adjoint but we can 
we consider the Friedrichs self-adjoint extension, which amounts to set Dirichlet conditions at $t=\pm \pi/2\ell_j$. It is clearly invertible for each $\ell_j>0$ and the inverse can be computed using  Fourier decomposition in $\theta$. If $L^{-1}$ is written under the form
\[ (L^{-1}f)(t,\theta)=\int_{-\tfrac{\pi}{2\ell_j}}^{\tfrac{\pi}{2\ell_j}}\int_{\R/\zz} G_L(t,\theta,t',\theta') f(t',\theta')d\theta'dt'
\]
for some Green kernel $G_L$, then $G_{g_j}$ can be written as 
\[G_{g_j}(\rho,\theta,\rho',\theta')= G_{L}\Big(\frac{\arctan(\tfrac{\rho}{\ell_j})}{\ell_j},\theta,\frac{\arctan(\tfrac{\rho'}{\ell_j})}{\ell_j},\theta'\Big).\]
This is clear since the left-hand side maps $C_c^\infty(\mc{F}_j)$ to $L^2(\mc{F}_j)$ and is a right inverse 
for $\Delta_{g_j}$ on $C_c^\infty(\mc{F}_j)$. 
Now, computing $G_L$ is quite simple: using the Fourier decomposition 
\[ L^{-1}f(t,\theta)=\sum_{k\in \zz} e^{2\pi ik\theta}(L_k^{-1}f_k)(t)\]
where $f(t,\theta)=\sum_ke^{2\pi ik\theta}f_k(t)$ and $L_k$ is the operator on
$I_j:=(-\tfrac{\pi}{2\ell_j},\tfrac{\pi}{2\ell_j})$ given by 
$L_k=-\pl_t^2+4\pi^2k^2$ with Dirichlet condition at $\pl I_j$. For $k\not=0$, a straightforward Sturm-Liouville argument gives the expression of the Green function for $L_k$: with $\bar{k}=2\pi k$, 
\[\begin{split}
G_{L_k}(t,t')=& \frac{1}{2\bar{k}(1-e^{-2\bar{k}\pi/\ell_j})}\Big((e^{-\bar{k}t}-e^{\bar{k}(t-\pi/\ell_j)})(e^{\bar{k}t'}-e^{-\bar{k}(t'+\pi/\ell_j)})\indic_{t\geq t'}\\
 & +(e^{-\bar{k}t'}-e^{\bar{k}(t'-\pi/\ell_j)})(e^{\bar{k}t}-e^{-\bar{k}(t+\pi/\ell_j)})\indic_{t'\geq t}\Big)\\
 = &\frac{e^{-\bar{k}|t-t'|}}{2\bar{k}}+\frac{e^{-2\pi \bar{k}/\ell_j}\cosh(\bar{k}(t-t'))-e^{-\bar{k}\pi/\ell_j}\cosh(\bar{k}(t+t'))}{\bar{k}(1-e^{-2\pi \bar{k}/\ell_j})}.
\end{split}\]
If $\max(|\rho|,|\rho'|)\leq 1$, we have $|t|\leq \tfrac{\pi}{2\ell_j}-1+\mc{O}(\ell_j^2)$ and same for $t'$ thus for $\ell_j$ small,
\[ \Big|\frac{e^{-\bar{k}\pi/\ell_j}\cosh(\bar{k}(t+t'))}{\bar{k}(1-e^{-2\pi \bar{k}/\ell_j})}\Big|
\leq \frac{e^{-\bar{k}/2}}{\bar{k}} ,\quad
\Big|\frac{e^{-2\pi \bar{k}/\ell_j}\cosh(\bar{k}(t-t'))}{\bar{k}(1-e^{-2\pi \bar{k}/\ell_j})}\Big|\leq \frac{e^{-\pi \bar{k}/\ell_j}}{\bar{k}}. \]
We then get 
\[ \sum_{k\not=0}G_{L_k}(t,t')e^{i\bar{k}(\theta-\theta')}=\sum_{k\not=0}
\frac{e^{-\bar{k}|t-t'|+i\bar{k}(\theta-\theta')}}{2\bar{k}}+\mc{O}(1)=-\frac{1}{2\pi}\log\Big|1-e^{-2\pi(|t-t'|- i(\theta-\theta'))}\Big|+\mc{O}(1)
\]
when $\ell_j$ is small. 
Notice also that if $|\rho|\in [\ell_j,1]$ and $|\rho'|\in [\ell_j,1]$ with 
$\rho$ and $\rho'$ having different sign, then 
\begin{equation}\label{bounddiffsign}
\Big|\sum_{k\not=0}G_{L_k}(t(\rho),t'(\rho'))e^{i\bar{k}(\theta-\theta')}\Big|\leq Ce^{-\frac{\pi}{2\ell_j}}
\end{equation}
if $t(\rho)=\ell_j^{-1}\arctan(\rho/\ell_j)$ and similarly for $t'(\rho')$.
Next for $k=0$, the Green function is given by the expression
\[ G_{L_0}(t,t')=-\frac{1}{2} |t-t'|-\frac{\ell_j}{\pi}tt'+\frac{\pi}{4\ell_j},\]
thus  we get 
\[ \begin{split}
G_{L_0}(t(\rho),t'(\rho'))= &-\frac{1}{2\ell_j}\Big|\arctan(\tfrac{\rho}{\ell_j})-\arctan(\tfrac{\rho'}{\ell_j})\Big|+\frac{1}{\pi \ell_j}\Big(\tfrac{\pi^2}{4}-\arctan(\tfrac{\rho}{\ell_j})\arctan(\tfrac{\rho'}{\ell_j})\Big)
\end{split}.
\]
If $F(x):=\int_{x}^{\infty}\tfrac{du}{1+u^2}$, this can be rewritten as 
\begin{equation}\label{upperbdGL0}
\begin{split}
G_{L_0}(\rho,\rho')=& -\tfrac{1}{2\ell_j}\Big| F(\tfrac{\rho}{\ell_j})-F(\tfrac{\rho'}{\ell_j})\Big|
+\tfrac{1}{2\ell_j}(F(\tfrac{\rho}{\ell_j})+F(\tfrac{\rho'}{\ell_j}))-\tfrac{1}{\pi\ell_j}F(\tfrac{\rho}{\ell_j})F(\tfrac{\rho'}{\ell_j})\\
= & \left\{\begin{array}{ll}
\tfrac{1}{\ell_j}\min (F(\tfrac{\rho}{\ell_j}),F(\tfrac{\rho'}{\ell_j}))-\tfrac{1}{\pi \ell_j}F(\tfrac{\rho}{\ell_j})F(\tfrac{\rho'}{\ell_j}), & \textrm{ if } \rho>0,\rho'>0\\
\tfrac{1}{\pi \ell_j }F(\tfrac{\rho}{\ell_j})F(-\tfrac{\rho'}{\ell_j}), & \textrm{ if }\rho>0,\rho'<0
\end{array}\right.
\end{split}\end{equation}
from which \eqref{boundGj} and \eqref{boundsigndiff} follow.

Now, we also see from the expressions of $G_{L_k}$ above 
that if $|\rho|\in [1/2,1]$ and $|\rho-\rho'|>\delta$ for some fixed constant $\delta>0$, then 
\[ |\pl_\rho G_{g_j}(\rho,\theta,\rho',\theta')|\leq C\]
for some constant $C$ depending only on $\delta$ but not on $\ell_j$.

Next we prove \eqref{Ggjonconstant}. Let $F(x)=\int_x^{\infty}\frac{du}{1+u^2}$ and let $c:=||\chi||_{L^\infty}$.  For $\rho>0$, we write  
\[
\begin{split}
\int G_{g_j}(\rho,\theta,\rho',\theta')\frac{\chi(\rho')}{|\rho'|^\la}dx'=&
\frac{1}{\ell_j \pi}F(\tfrac{\rho}{\ell_j})\int_{-1}^0F(-\tfrac{\rho'}{\ell_j})\frac{\chi(\rho')}{|\rho'|^\la}d\rho'+
\int_{0}^\infty G_{g_j}(\rho,\theta,\rho',\theta')\frac{\chi(\rho')}{{\rho'}^\la}d\rho' \\
\leq & \frac{c}{\pi\rho}\int_{-1}^0 F(-\tfrac{\rho'}{\ell_j})\frac{1}{|\rho'|^\la}d\rho'+ \frac{c}{\rho}\int_{0}^\rho 
\frac{d\rho'}{{\rho'}^\la}+c
\int_\rho^{1}\frac{d\rho'}{{\rho'}^{1+\la}}\\
\leq & 2c\frac{\rho^{-\la}}{1-\la}+ 2c\frac{\rho^{-\la}-1}{\la}
\end{split}\]
where we used \eqref{upperbdGL0} in the second line and 
\[\int_{-1}^{-\rho}F(-\tfrac{\rho'}{\ell_j})\frac{d\rho'}{|\rho'|^\la}\leq \ell_j\int_{\rho}^{1}\frac{d\rho'}
{{\rho'}^{1+\la}}\leq \ell_j\frac{\rho^{-\la}-1}{\la}, \quad \int_{-\rho}^{0}F(-\tfrac{\rho'}{\ell_j})\frac{d\rho'}{|\rho'|^\la}\leq \frac{\pi}{2(1-\la)}\rho^{1-\la}
\]
for the third line. The same estimate works in the case $\rho<0$.

Using the estimates and expressions above, we also get as 
$|(\rho',\theta')-(\rho,\theta)|$ is small
\[ G_{L}(t,\theta,t',\theta')=-\frac{1}{2\pi}\log(\sqrt{(t-t')^2+(\theta-\theta')^2})-
\frac{\ell_j}{\pi}t^2+\frac{\pi}{4\ell_j}+\mc{O}(1)+\mc{O}(|t-t'|+|\theta-\theta'|)\] 
where $\mc{O}(1)$ is independent of all variables. We also have 
\[ \log(d_{g_j}(\rho,\theta,\rho',\theta'))=\log(\sqrt{(t-t')^2+(\theta-\theta')^2})+
\log(\sqrt{\rho^2+\ell_j^2})+\mc{O}(|t-t'|+|\theta-\theta'|)\]
thus the Robin mass of $G_{g_j}$ satisfies \eqref{boundrobin}.
\end{proof} 

Next we express $\int_{\pl D(0,\eps)}R_g(\la)d\la/\la$ in terms of 
$G_{g_j}$. Let $\chi_j,\chi_j'\in C^\infty(M)$ 
which are supported in $\mc{C}_j(g)$, depending only on the variable $\rho$ associated to the metric $g$,
are equal to $1$ in $|\rho|\leq 1/4$ and such that $\chi'_j=1$ on a neighborhood of $\supp(\chi_j)$. We will use the diffeomorphisms 
between $\mc{C}_j(g)$ and the subset $|\rho|<1$ of $\mc{F}_j$ to identify these sets (for notational simplicity we won't input these diffeos below).
Using that $\Delta_g=\Delta_{g_j}$ in $\mc{C}_j(g)$, 
we have, with $\chi:=\sum_{j}\chi_j$ 
\[(\Delta_{g}-\la)\sum_{j}\chi'_jR_{g_j}(\la)\chi_j=\chi+\sum_j [\Delta_g,\chi'_j]R_{g_j}(\la)\chi_j\]
and thus 
\[ \sum_{j}\chi'_jR_{g_j}(\la)\chi_j=R_g(\la)\chi+R_g(\la)\sum_j [\Delta_g,\chi'_j]R_{g_j}(\la)\chi_j.\]
Similarly, we also have 
\[\sum_{j}\chi_jR_{g_j}(\la)\chi'_j=\chi R_g(\la)+\sum_j \chi_jR_{g_j}(\la)[\chi_j',\Delta_g]R_g(\la)\]
and therefore (using also $\chi'_j\chi'_k=0$ if $i\not=k$)
\[\begin{split}
\chi R_g(\la)\chi=& \sum_{j}\chi_jR_{g_j}(\la)\chi_j-\sum_{j}\chi_jR_{g_j}(\la)\chi'_j[\Delta_g,\chi'_j]R_{g_j}(\la)\chi_j\\
& +\sum_j \chi_jR_{g_j}(\la)[\Delta_g,\chi'_j]R_g(\la)\sum_i[\chi'_i,\Delta_g]R_{g_i}(\la)\chi_i\\
=&\sum_{j}\chi_jR_{g_j}(\la)\chi_j -K_1(\la)+K_2(\la)R_{g}(\la)K_3(\la)
\end{split}\]
where $K_1(\la),K_2(\la),K_3(\la)$ are defined by the equation. Remark that, in the right hand side, only the last term has poles (first order) in $D(0,\eps)$.
 Therefore, using Cauchy formula
\[\begin{split} 
\int_{\pl D(0,\eps)}\frac{\chi R_g(\la)\chi}{2\pi i\la} d\la= & \sum_{j}\chi_jR_{g_j}(0)\chi_j
-K_1(0)+K_2(0)A_gK_3(0)\\
& -K'_2(0)\Pi_0K_3(0)-K_2(0)\Pi_0K'_3(0)\\
& -\sum_{k=1}^m\frac{(K_2(\la_k)-K_2(0))}{\la_k}\Pi_{\la_k}K_3(\la_k)-K_2(0)\Pi_{\la_k}\frac{(K_3(\la_k)-K_3(0))}{\la_k}
\end{split}\] 
where $K_i'(0):=\pl_{\la}K_i(\la)|_{\la=0}$ and $A_g:=\int_{\pl D(0,\eps)}\frac{R_g(\la)}{2\pi i\la} d\la$.
Let us analyse those terms more carefully.
\begin{proposition}\label{lastestimateGg}
Let $g_0\in \pl\mc{M}_{{\bf g}}$ be a hyperbolic surface  
with cusps in the boundary of moduli space. Then there is a neighborhood $U_{g_0}$ of $g_0$ in $\mc{M}_{\bf g}$ and $C>0$ such that for all $g\in U_{g_0}$ and $x,x'\in \cup_j \mc{C}_j(g)$
\[\begin{split}
\chi(x)\chi(x')G_{g}(x,x')=& \sum_{j}\chi_j(x)G_{g_j}(x,x')\chi_j(x')+Q_{g}(x,x')-H_0(x)J_0(x')-J_0(x)H_0(x')\\
& +\sum_{k=1}^m\Big(\frac{H_k(x)-\la_kJ_k(x)}{\sqrt{\la_k}}\Big)\Big(\frac{H_k(x')-\la_kJ_k(x')}{\sqrt{\la_k}}\Big)\end{split}\]
with $Q_g\in C^\infty(M\x M)$, $H_k=\chi \varphi_{\la_k}$, $J_k\in C^\infty(M)$ satisfying
\[\begin{gathered} 
|Q_g|_{L^\infty}\leq C , \quad |H_0|_{L^\infty}\leq C, \quad |H_k(\rho,\theta)|\leq C|\rho|^{-(1-s_k)}, \\ 
 |J_0(\rho,\theta)|\leq C|\log|\rho||, \quad  |J_k(\rho,\theta)|\leq C\Big(|\rho|^{-(1-s_k)}+
 \frac{|\rho|^{-(1-s_k)}-1}{(1-s_k)}\Big)
 \end{gathered}\]
with $s_k(1-s_k)=\la_k(g)>0$ the small eigenvalues of $\Delta_g$ converging to $0$ as $g$ approaches $\pl\mc{M}_{\bf g}$, $\varphi_{\la_k}$the associated normalized eigenfunctions and $s_k=1-\la_k(g)+\mc{O}(\la_k(g)^2)$.  Here $\rho=\ell_j\sinh(t)$ in $\mc{C}_j(g)$ with $t$ the signed distance to the geodesic $\gamma_j(g)$.
\end{proposition}
\begin{proof}
 The integral kernel of $\sum_{j}\chi_jR_{g_j}(0)\chi_j$ is $\chi_j(x)\chi_j(x')G_{g_j}(x,x')$ in $\mc{C}_j(g)$
with respect to the measure ${\rm dv}_{g_j}$. This term has an explicit bound by using Lemma \ref{G_j}. The function $Q_g(x,x')$  will be chosen to be the integral kernel of $-K_1(0)+K_2(0)A_gK_3(0)$, let us show this is smooth and uniformly bounded.
The integral kernel of $K_1(0)$ is of the form
\[ K_1(0;x,x')=\sum_{j}\chi_j(x)\int_{\mc{C}'_j}G_{g_j}(x,y)\chi_j'(y)P_j(y)G_{g_j}(y,x'){\rm dv}_{g_j}(y)\chi_j(x')
\]
in $\mc{C}_j(g)$ with respect to the measure ${\rm dv}_{g_j}$, where $P_j$ is a smooth 
differential operator of order $1$ supported in $\supp(\nabla \chi_j')$ (thus far from 
the $\gamma_j(g)=\{\rho=0\}$ curve). First $K_1$ has smooth integral kernel 
since $P_j(y)G_{g_j}(y,x')\chi_j(x')$ is smooth as $\supp(\nabla \chi_j')\cap \supp(\chi_j)=\emptyset$. Moreover, by \eqref{boundderivee} and \eqref{boundsigndiff} we directly get that for all $x,x'\in \mc{C}_j(g)$ with $|\rho(x)|>\ell_j$ 
and $|\rho(x')|>\ell_j$ 
\[ |K_1(0;x,x')|\leq C.\] 

By Proposition \ref{greenpinching}, the norm $||A_g||_{L^2(W)\to L^2(W)}$ is uniformly bounded if $W:=\supp(\nabla \chi'_j)$, and by the same argument as for $K_1(0)$, $K_2(0)$ and $K_3(0)$ have smooth integral kernels that are uniformly bounded with respect to $\ell_j$, thus there is $C,C'>0$ uniform so that for all $x,x'\in M$
\[ |(K_2(0)A_gK_3(0))(x,x')|\leq C\|K_2(0)(x,\cdot)\|_{L^2(W)}\|K_3(0)(\cdot,x')\|_{L^2(W)}\leq C'.\] 

We can rewrite $K'_2(0)\Pi_0K_3(0)$ by using that $\Pi_0=\sqrt{c}\cjg \sqrt{c},\cdot\cjd$ with $c=1/{\rm Vol}_g(M)$:
\[\begin{split}
(K'_2(0)\Pi_0K_3(0))(x,x')= & c\sum_j \chi_j(x)(\pl_\la R_{g_j}(0)\Delta_g\chi'_j)(x)\chi(x')\\
=& c\sum_j\chi_j(x)(R_{g_j}(0)\chi'_j)(x)\chi(x')=H_0(x')J_0(x)
\end{split}\]
where we used $\pl_\la R_{g_j}(0)\Delta_{g_j}=R_{g_j}(0)$, and $J_0(x):=\sqrt{c}
\sum_j\chi_j(x)\int_{\mc{C}_j}G_{g_j}(x,x')\chi_j(x'){\rm dv}_{g_j}(x')$ and $H_0:=\sqrt{c}\chi$ are smooth functions on $M$. Similarly, we get
\[(K_2(0)\Pi_0K'_3(0))(x,x')=H_0(x)J_0(x').\] 
Now the bound \eqref{Ggjonconstant} gives that $|J_0(\rho,\theta)|\leq C|\log |\rho||$ and $|H_0(x)|\leq C$ for some uniform $C>0$.

Using that $(\Delta_g-\la_k)\Pi_{\la_k}=\Pi_{\la_k}(\Delta_g-\la_k)=0$, we get 
\[\begin{split}
\frac{(K_2(\la_k)-K_2(0))}{\la_k}\Pi_{\la_k}K_3(\la_k)=&
\sum_{j}\chi_jR_{g_j}(0)\chi_j'\Pi_{\la_k}\sum_{i}[\chi_i',\Delta_g]R_{g_i}(\la_k)\chi_i\\
=& \sum_{j}\chi_jR_{g_j}(0)\chi_j'\Pi_{\la_k}\chi.
\end{split}\]
Similarly, we also get 
\[\begin{split}
K_2(0)\Pi_{\la_k}\frac{(K_3(\la_k)-K_3(0))}{\la_k}=&
\sum_{j}\chi_jR_{g_j}(0)[\Delta_g,\chi'_j]\Pi_{\la_k}\sum_{i}\chi_i'R_{g_i}(0)\chi_i\\
=& \chi\Pi_{\la_k}\sum_{i}\chi_i'R_{g_i}(0)\chi_i -\la_k \sum_{j}\chi_jR_{g_j}(0)\chi_j'\Pi_{\la_k}\sum_{i}\chi_i'R_{g_i}(0)\chi_i.
\end{split}\]
Write $\sum_k\Pi_{\la_k}(x,x')=\sum_k\varphi_{\la_k}(x)\varphi_{\la_k}(x')$ for some $\varphi_{\la_k}\in \ker (\Delta_g-\la_k)$ orthonormal basis of eigenfunctions associated to the small eigenvalues $\la_k$ (repeated with multiplicities), then \[\begin{gathered}
\sum_k\frac{(K_2(\la_k)-K_2(0))}{\la_k}\Pi_{\la_k}K_3(\la_k)+K_2(0)\Pi_{\la_k}\frac{(K_3(\la_k)-K_3(0))}{\la_k}=\\ 
\sum_k H_k(x)J_k(x')+J_k(x)H_k(x')-\la_k J_k(x)J_k(x')
\end{gathered}\]
where $H_k=\chi \varphi_{\la_k}$, $J_k(x)=\sum_j\chi_j(x)\int_{\mc{C}_j(g)}G_{g_j}(x,x')\chi_j'(x)\varphi_{\la_k}(x'){\rm dv}_{g_j}(x')$.
To conclude, we need some estimates on the eigenfunction $\varphi_{\la_k}$ associated to the small positive eigenvalue $\la_k$ of $\Delta_g$ on $M$. In the case of one pinched geodesic, this can be obtained by \cite[Proposition 7.2]{ARS}, but we provide a more general (but softer) estimate that will be useful later in the paper.
\begin{lemma}\label{boundvarphi1}
Let $g_0\in \pl\mc{M}_g$ be a surface with node with $m+1$ connected components 
and denote by $\la_1,\dots \la_m$ the positive small eigenvalues.
For each $\eps>0$, there is a neighborhood $U_{g_0}$ of $g_0$ and $C>0$ such that for each 
$g\in U_{g_0}$, the following holds: if $\varphi_{\la_i}$ is an eigenfunction for $\la_i$
which satisfies $|\varphi_{\la_i}|_{\rho=\pm 1}-a^\pm_{ij}|<\eps$ for some constants $a_{ij}^\pm\in \R$ 
in the collar $\mc{C}_j(g)$, then it satisfies 
\[\varphi_{\la_i}(\rho,\theta)=(a_{ij}^+\indic_{\rho>0}+a_{ij}^-\indic_{\rho<0})|\rho|^{s_i-1}(1+\mc{O}(\eps))+\mc{O}(\eps)\]
in the region 
$\{|\rho|>C\ell_j\}$ of the collar $\mc{C}_j(g)$, where $s_i(1-s_i)=\la_i$ and $s_i=1-\la_i+\mc{O}(\la_i^2)$ when $\la_i$ is small. In the region $|\rho|<C\ell_j$, there is $C'>0$ such that 
\[|\varphi_{\la_i}(\rho,\theta)|\leq C'|\rho|^{s_i-1}.\]
\end{lemma} 
\begin{proof} To simplify notations, we remove the $i,j$ indices from $a^\pm_{ij}$, $\la_i$ and $s_i$. 
We decompose $\varphi_{\la}$ in Fourier modes in $\theta$: there are $b_k\in C^\infty([-1,1])$ so that
\[\varphi_{\la}(\rho,\theta)=\sum_{k=-\infty}^\infty b_k(\rho)e^{2\pi ik\theta}\] 
and the series converges uniformly.
Since $\varphi_{\la}$ can be supposed real-valued, $b_0$ is real and $b_{-k}=\bbar{b_k}$.
Moreover $a_k(t):=b_k(\ell_j\sinh(t))$ satisfies the ODE
\[ \Big(-\pl_{t}^2-\tanh(t)\pl_t+\frac{4\pi^2k^2}{\ell_j^2\cosh(t)^2}-\la\Big)a_k(t)=0, \quad a_k(\pm t_j)=a_k^{\pm}\]
if $t_j$ is defined by $\ell_j\sinh(t_j)=1$. 
We write $a_k^{\pm}:=a_k|_{t=\pm t_j}=b_k|_{\rho=\pm 1}$. 
First we make the following observation for each $k\not=0:$
\begin{equation}\label{maxprinciple} 
|a_k(t)|\leq \max(|a_k^+|,|a_k^-|)<\eps.
\end{equation}
Indeed,  assume that $a_k$ achieves its maximum at $T\in (-t_j,t_j)$ with 
$a_k(T)>\max(a_k^+,a_k^-)$, then if $a_k(T)>0$
\[ -a_k''(T)=\Big(\la-\frac{4\pi^2k^2}{\ell_j^2\cosh(T)^2}\Big)a_k(T)<0\]
when $\ell_j$ is smaller than a uniform constant. This contradicts that $a_k(T)$ is a local maximum, 
thus $a_k$ achieves its maximum at $\pm t_j$ or its maximum is non-positive, in which case $a_k^+\leq 0$ and $a_k^-\leq 0$. In both cases, $|a_k(T)|\leq \max(|a_k^+|,|a_k^-|)$. 
The same argument works with the minimum and 
this shows \eqref{maxprinciple}. 
Next we analyze $u_0(t)$, and it is convenient for that to use $s\in(0,1)$ so that $s(1-s)=\la$
(then $s=1-\la+\mc{O}(\la)^2$). There are $2$ independent solutions of 
the ODE with $k=0$ (and no bounday condition), the first one $v_0$ is odd in $t$, the other one $u_0$ is even, they are  given by \cite[Chapter 5.1]{Bo}  
\[ \begin{split}
v_0(t)= & {\rm sign}(t)\frac{\Gamma(\demi-s)\Gamma(\frac{1+s}{2})^2}{\Gamma(s-\demi)\Gamma(1-\frac{s}{2})^2}
|\sinh(t)|^{-s}F\Big(\frac{1+s}{2},\frac{s}{2},\frac{1}{2}+s;\frac{-1}{\sinh(t)^2}\Big)\\
& +
{\rm sign}(t)|\sinh(t)|^{s-1}F\Big(\frac{2-s}{2},\frac{1-s}{2},\frac{3}{2}-s;\frac{-1}{\sinh(t)^2}\Big),\\
u_0(t) =&  \frac{\Gamma(\demi-s)\Gamma(\frac{s}{2})^2}{\Gamma(s-\demi)\Gamma(\frac{1-s}{2})^2} 
|\sinh(t)|^{-s}F\Big(\frac{s}{2},\frac{s+1}{2},\frac{1}{2}+s;\frac{-1}{\sinh(t)^2}\Big)
\\
& +|\sinh(t)|^{s-1}F\Big(\frac{1-s}{2},1-\frac{s}{2},\frac{3}{2}-s;\frac{-1}{\sinh(t)^2}\Big) 
\end{split}\]
where $F(a,b,c;z)$ is the hypergeometric function, holomorphic in the variables $a,b,c$ 
for $a,b,c\in\cc$ in the half-plane 
$\cc_+:=\{c\in \cc, {\rm Re}(c)>0\}$ if $z\in (-\infty,0)$, it is smooth in $z$ and for $z<0$ small
\[F(a,b,c;z)=1+\mc{O}(|z|)\]
where the remainder is uniform for $a,b,c$ in compact sets of $\cc_+$.
In particular, there is $C>0$ uniform in $g$ so that for $|t|>C$,  
\[\begin{gathered}
v_0(t)= {\rm sign}(t)|\sinh(t)|^{s-1}+\mc{O}(|\sinh t|^{-s}), \quad
u_0(t)= |\sinh(t)|^{s-1}+\mc{O}(|\sinh t|^{-s})
\end{gathered} \] 
where the remainder is uniform with respect to $\lambda$ for $\lambda>0$ small.
We obtain 
\[ a_0(t)=\frac{(a_0^++a_0^-)}{2}\frac{u_0(t)}{u_0(t_j)}+\frac{(a_0^+-a_0^-)}{2}\frac{v_0(t)}{v_0(t_j)}\]
and we deduce that 
\[ \begin{split}
a_0(t)=& (a_0^+\indic_{t>0}+a_0^-\indic_{t<0})|\ell_j\sinh(t)|^{s-1}+\mc{O}(\ell_j^{s-1}|\sinh(t)|^{-s}).
\end{split}
\]
We now use $\ell_j\sinh(t)=\rho$ and $\ell_j^{s-1}|\sinh(t)|^{-s}= \ell_j^{2s-1}|\rho|^{-s}\leq 
\eps |\rho|^{s-1}$ if $|\rho|>C\ell_j$ with $C$ large enough (depending on $\eps$). 
 
Next, consider the case $|\rho|<C\ell_j$. We can also write $u_0$ and $v_0$ under the form (see \cite[Chapter 5.5]{Bo})
\[\begin{gathered}
v_0(t)=\frac{\Gamma(\frac{1+s}{2})^2}{\Gamma(\frac{3}{2})\Gamma(s-\frac{1}{2})}\sinh(t)F\Big(\frac{1+s}{2},1-\frac{s}{2},\frac{3}{2}; -\sinh(t)^2\Big), \\
u_0(t)=\frac{\Gamma(\frac{s}{2})^2}{\Gamma(\frac{1}{2})\Gamma(\frac{1}{2}-s)}F\Big(\frac{s}{2},\frac{1-s}{2},\frac{1}{2}; -\sinh(t)^2\Big)
\end{gathered}\]
and this easily yields the desired estimate. 
\end{proof}
The proof is complete by noticing that the estimates on $H_k,J_k$ follow fom this Lemma and \eqref{Ggjonconstant}, together with the fact that each $\varphi_{\la_k}$ is bounded uniformly in 
$M\setminus \cup_j \mc{C}_j(g)$ for $g\in U_{g_0}$ by Lemma \ref{approximation}. 
\end{proof}

\section{Gaussian Free Field and Gaussian Multiplicative Chaos}\label{sec:GFF}

In this section, we shall explain how to give a mathematical sense to the formal measure 
\begin{equation}\label{liouvillemesure}
F\mapsto \int F(\varphi) e^{-S_L(g,\varphi)}D\varphi
\end{equation}
where $S_L(g,\varphi)$ is the Liouville functional defined in \eqref{QLiouville}, $g$ is a fixed metric on the surface $M$ and $\varphi$ varies among a certain space of functions so that $e^{\gamma\varphi}g$  
is parametrizing the conformal class $[g]$ of $g$.  This will allow us to define the partition function of Liouville Quantum Field Theory, and in fact $\varphi$ will be a field, i.e. a random function or random distribution, 
that we will denote by $X_g$.
The first step is to make sense of the part corresponding to the squared gradient term in $S_L(g,\varphi)$, i.e. the formal  Gaussian measure 
\begin{equation}\label{Gaussianmesure} 
F\mapsto \int F(\varphi) e^{-\frac{1}{4\pi}||d\varphi||_{L^2}^2}D\varphi.
\end{equation}
Classically, we interpret the above field $\varphi$ as a Gaussian Free Field (GFF in short): this is a  Gaussian random variable taking values in some space of distributions in the sense of Schwartz. In particular, the field $\varphi$ is not a well-defined function and giving sense to the term $e^{\gamma \varphi}$ in  \eqref{QLiouville} is thus not straightforward, but it can be done through the theory of Gaussian Multiplicative Chaos (GMC in short), which goes back to \cite{cf:Kah}.

\subsection{Gaussian Free Field}\label{SecGFF}

We describe the Gaussian Free Field on a compact Riemannian surface $(M,g)$ by using our previous description of the Green function. The definition of the GFF, as well as the definition of  its partition function, can be carried out in a direct way (see for instance \cite{dubedat,She07}). Yet, this path may not be as pedagogical as following the circle of ideas that led physicists to our current knowledge of this object, and this is what we try to summarize heuristically below to end up with a mathematically sound definition.

As a warm up, let us quickly recall that the Gaussian measure $$(2\pi)^{-n/2}\sqrt{\det(A)}e^{-{\demi} \cjg Ax,x\cjd}dx$$ on $\R^n$, when $A$ is a positive definite symmetric matrix, is the law of the random variable $X=\sum_{j=1}^n\alpha_j\varphi_j/\sqrt{\la_j}$ where $(\alpha_j)_j$ are independent Gaussian random variables in $\mc{N}(0,1)$ (mean $0$ and variance $1$),
and $(\varphi_j)_j$ is an orthonormal basis of eigenvectors for $A$ with eigenvalues $\la_j>0$. 

As the GFF is an infinite dimensional Gaussian, it is natural to expect a construction through its projections onto finite dimensional subspaces, on which one can apply the construction   described just above. Recall that the Laplacian $\Delta_g$ has an orthonormal basis of real valued eigenfunctions $(\varphi_j)_{j\in \N_0}$ in $L^2(M,g)$ with associated eigenvalues $\la_j\geq 0$; we set $\la_0=0$ and $\varphi_0=({\rm Vol}_g(M))^{-1/2}$. The Laplacian can thus be seen as a symmetric operator on an infinite dimensional space. Denote $H_n$  the finite dimensional space  spanned by the first $n$ eigenfunctions $(\varphi_j)_{j=1,\dots,n}$ of the Laplacian. Notice that for $\varphi=\sum_{j=1}^n \tilde{\alpha}_j \varphi_j$ we have
$\|d\varphi\|^2_{L^2}=\sum_{j=1}^n\tilde{\alpha}_j^2 \lambda_j$. Therefore the projection to $H_n$ of the formal measure \eqref{Gaussianmesure} is naturally understood as
\begin{align*}
 \int_{H_n} F(\varphi)e^{-\frac{1}{4\pi}\|d\varphi\|^2_{L^2}}D\varphi & =\int_{\R^n}F\big(\sum_{j=1}^n \tilde{\alpha}_j \varphi_j\big)\prod_{j=1}^n\Big(e^{-\frac{1}{4\pi} (\tilde{\alpha}_j)^2 \lambda_j} d \tilde{\alpha}_j\Big)  \\
 & =  (2\pi)^{n/2}\Big(\prod_{j=1}^{n}\lambda_j\Big)^{-1/2} \int_{\R^n}F\big(  \sqrt{2 \pi} \sum_{j=1}^n\alpha_j \frac{\varphi_j }{\sqrt{\lambda_j} } \big)\prod_{j=1}^n\Big(e^{-\frac{\alpha_j^2}{2}}  d\alpha_j\Big) \\
\end{align*}
for appropriate bounded measurable functionals $F$. The mass of this Gaussian measure is $(2\pi)^{n}\Big(\prod_{j=1}^{n}\lambda_j\Big)^{-1/2}$.  
Renormalized by its mass, this measure becomes a probability measure describing the law  of the random function
\begin{equation}\label{truncGFF}
X_n:=  \sqrt{2 \pi} \sum_{j=1}^n\alpha_j \frac{\varphi_j}{\sqrt{\la_j}}
\end{equation}
 where $(\alpha_j)_j$ are independent Gaussian random variables with law $\mc{N}(0,1)$. 

To obtain the description of the GFF, one has to take the   limit $n\to\infty$. It can be seen \cite{dubedat} that the sum \eqref{truncGFF} converges almost surely in the Sobolev space $H^{-s}(M)$ for each $s>0$. The mass $(2\pi)^{n}\Big(\prod_{j=1}^{n}\lambda_j\Big)^{-1/2}$ diverges  as $n\to \infty$ but this is not that much troublesome as it is customary in physics  
(and can be done mathematically) to remove the diverging terms provided they are "universal enough": this procedure is called renormalization. Removing the diverging terms should give a limiting total mass equal to  $ ( \det '(\tfrac{1}{2\pi}\Delta_g)) ^{-1/2}$. 
So far, this is the picture the reader should have in mind to understand the construction of the GFF. Yet, for readers who want to have more details, we stress  that  renormalizing  the product $\prod_{j=1}^{n}\lambda_j$  turns out to be very troublesome and slight adaptations are necessary to recover the  phenomenology explained above. The reader may consult the paper  \cite{bilal} where  these  renormalization issues are discussed in further details. 

\medskip
The above formal discussion thus motivates the forthcoming definitions. The Green function $G_g(x,x')$ (with vanishing mean) is a distribution on $M\x M$ which can be written as the series, converging in the sense of distributions, 
\[ G_g(x,x')=\sum_{j=1}^\infty \frac{\varphi_j(x)\varphi_j(x')}{\la_j}.\] 
Let $(a_j)_j$ be a sequence of i.i.d. real Gaussians   $\mc{N}(0,1)$, defined on some probability space   $(\Omega,\mc{F},\mathbb{P})$,  and define  the Gaussian Free Field with vanishing mean in the metric $g$ by
\begin{equation}
X_g=\sqrt{2\pi}\sum_{j\geq 1}a_j\frac{\varphi_j}{\sqrt{\la_j}} 
\end{equation}
as a random variable with values in $\mc{D}'(M)$, i.e. almost surely $X_g\in \mc{D}'(M)$ (see   \cite[Section 4.2]{dubedat} for instance). Notice that for each $\phi\in C^\infty(M)$, almost surely we have $\cjg X_g,\phi\cjd=\sqrt{2\pi}\sum_{j=1}^\infty a_j \frac{\cjg \varphi_j,\phi\cjd}{\sqrt{\la_j}}$ which is a 
converging series of random variables as $\mathbb{E}(\cjg X_g,\phi\cjd^2)<\infty$. 
 In fact, if $H_0^{-s}(M)$ is the kernel of the map $X\mapsto \cjg X,1\cjd_{L^2({\rm dv}_g)}$ on the $L^2$-based Sobolev space $H^{-s}(M)$ of order $-s\in\R^*_-$, it is easy to see (see \cite{dubedat} again) that $X_g$ makes sense as a random variable with values in $L^2(\Omega; H_0^{-s}(M))$ for all $s>0$ by using the asymptotic counting function on the eigenvalues $\la_j$ (i.e. the Weyl law). If $\phi_1,\phi_2\in C^\infty(M)$, the covariance is 
\[ \mathbb{E}[ \cjg X_g,\phi_1\cjd.\cjg X_g,\phi_2\cjd]=2\pi\sum_{j=1}^\infty \frac{\cjg \varphi_j,\phi_1\cjd \cjg\varphi_j,\phi_2\cjd}{\la_j}=2\pi
\cjg G_g,\phi_1\otimes\phi_2\cjd. 
\]
The covariance is then the Green function when viewed as a distribution: if $\phi_1\to \delta_x$ and $\phi_2\to \delta_{x'}$ for $x\not=x'$, $\mathbb{E}(\cjg X_g,\phi_1\cjd.\cjg X_g,\phi_2\cjd)\to 2\pi G_g(x,x')$ and we will write with a slight abuse of notation
\[\mathbb{E}[X_g(x).X_g(x')]=2\pi \,G_g(x,x').\]
Notice that the extra $2\pi$ factor serves to make the field $X_g$ have exact logarithmic correlations in view of Lemma \ref{greenneardiag}. As in \cite[Theorem 2.3]{She07}, there is a probability measure $\mc{P}$ on $H^{-s}_0(M)$
(for some natural $\sigma$-algebra) so that the law of $X_g$ is given by $\mc{P}$ and for each $\phi\in H^{s}(M)$, $\cjg X_g,\phi\cjd$ is a random variable on $\Omega$ with zero mean and variance $2\pi\cjg R_g(0)\phi,\phi\cjd$. The measure $\mc{P}$ represents the Gaussian measure \eqref{Gaussianmesure} (times the $\sqrt{\det'(\Delta_g)}$ term) on the space of functions orthogonal to constants, thus to
define \eqref{Gaussianmesure} on the whole $H^{-s}(M)$ space, we shall consider the tensor product $\mc{P}\otimes dc$ where $dc$ is the Lebesgue measure on $\R$ viewed as the 1-dimensional vector space of constant functions on $M$: 
in other words, we use the isomorphism 
\[ H_0^{-s}(M)\x \R \to H^{-s}(M), \qquad   (X,c)\mapsto X+c\]
to define the measure $\mc{P}'$ on $H^{-s}(M)$ as the image of $\mc{P}\otimes dc$  by this map. This measure gives a proper sense to the Gaussian measure \eqref{Gaussianmesure} times the global factor $ \big( \det '(\tfrac{1}{2\pi}\Delta_g) /{\rm Vol}_g(M)\big)^{-1/2}$. The extra term ${\rm Vol}_g(M)^{1/2}$ is a normalisation factor coming from the fact that ${\rm Vol}_g(M)^{-1/2}$ is of norm
$1$ in $L^2(M,{\rm dv}_g)$. We have 
\begin{lemma}\label{invarianceconforme}
The measure $\mc{P}'$ on $H^{-s}(M)$ obtained by tensorizing the GFF measure $\mc{P}$ by $dc$ is conformally invariant in the sense that it does not depend on the conformal representative in a conformal class $[g]$. 
\end{lemma}
\begin{proof} Let $\hat{g}=e^{\omega}g$ for some $\omega\in C^\infty(M)$. Notice that $H_{0}^{-s}(M)$ depends 
on $g$, we thus denote it $H_{0}^{-s}(M,g)$ and we denote $\cjg \cdot,\cdot\cjd_{g}$ the distribution pairing on $M$ or $M\times M$ induced by the measure ${\rm dv}_g$. First we claim that the probability
law obtained from $\hat{X}_g:=X_g-c_{\hat{g}}(X_g)$ is the same as that of $X_{\hat{g}}$, 
if $c_{\hat{g}}(X_g):=\cjg X_g, 1\cjd_{\hat{g}}/{\rm Vol}_{\hat{g}}(M)=\cjg X_g,e^{\omega}\cjd_{g}/{\rm Vol}_{\hat{g}}(M)$. The random field $\hat{X}_g$ satisfies 
$\cjg \hat{X}_g,1\cjd_{\hat{g}}=0$ and is thus in the space $H_{0}^{-s}(M,\hat{g})$, moreover 
$\mathbb{E}[\cjg \hat{X}_g,\phi\cjd_{\hat{g}}]=0$ for all $\phi\in C^\infty(M)$. The covariance of $\hat{X}_g$ is given by 
\[\begin{split}
\mathbb{E}[\cjg \hat{X}_g,\phi_1\cjd_{\hat{g}}\cjg \hat{X}_g,\phi_2\cjd_{\hat{g}}]
= &\cjg G_g,\phi_1\otimes\phi_2\cjd_{\hat{g}}+({\rm Vol}_{\hat{g}}(M))^{-2}\cjg G_g,1\otimes 1\cjd_{\hat{g}}\cjg \phi_1,1\cjd_{\hat{g}}\cjg \phi_2,1\cjd_{\hat{g}}\\
 & -({\rm Vol}_{\hat{g}}(M))^{-1}(\cjg G_g,1\otimes \phi_2\cjd_{\hat{g}}\cjg 1,\phi_1\cjd_{\hat{g}}+\cjg G_g,\phi_1\otimes 1\cjd_{\hat{g}}\cjg 1,\phi_2\cjd_{\hat{g}})\\
 =& \cjg G_g+ \alpha 1\otimes 1-u\otimes 1-1\otimes u,\phi_1\otimes \phi_2\cjd_{\hat{g}}
\end{split}\]
where $\alpha=({\rm Vol}_{\hat{g}}(M))^{-2}\cjg G_g,1\otimes 1\cjd_{\hat{g}}$, $u(x)=\int_{M}G_g(x,y){\rm dv}_{\hat{g}}(y)/{\rm Vol}_{\hat{g}}(M)$. We recognize from \eqref{greenhatg} that this kernel is just the Green function for $\hat{g}$ paired with $\phi_1\otimes \phi_2$, showing that the correlation of $\hat{X}_g$ is that of $X_{\hat{g}}$. Since both random fields are Gaussian, we deduce that the law of $\hat{X}_g$ and $X_{\hat{g}}$ are the same and thus for $F\in L^1(H^{-s}(M),\mc{P}')$,  
\[\int_{\R}\mathbb{E}[F(X_{\hat{g}}+c)]dc=\int_{\R}\mathbb{E}[F(X_g-c_{\hat{g}}(X_g)+c)]dc=\int_{\R}\mathbb{E}[F(X_g+c)]dc\] 
by making a change of variables in $c$.
\end{proof}

Finally, in view of the discussion above, the measure $F\mapsto \int_{\rr}\E(F(X_g+c))dc$ on $H^{-s}(M)$ is our mathematical definition for the formal measure 
\begin{equation}\label{mesureGFFfin}
  \big( \det {}'(\tfrac{1}{2\pi}\Delta_g) /{\rm Vol}_g(M)\big)^{1/2} \, e^{-\frac{1}{4\pi}\int_M|d\varphi|_g^2{\rm dv}_g} D\varphi
\end{equation}
and using \eqref{detpolyakov}, we can write $\sqrt{\det '(\frac{1}{2\pi}\Delta_g)}=(2\pi)^{\demi(1-\frac{\chi(M)}{6})}\sqrt{ \det'(\Delta_g)}$.

\subsection{Gaussian multiplicative chaos}\label{GMC}
To define quantities like $e^{\gamma X}$ for some $\gamma>0$ we will use a renormalization procedure after regularization of the field $X_g$. We describe the construction for $g$ hyperbolic and we shall 
remark that in fact the construction works as well for any conformal metric $\hat{g}=e^{\omega}g$ by using Lemma \ref{greenneardiag}.

First, when $\eps>0$ is very small, we define a regularization $X_{g,\eps}$ of $X_g$ by averaging on geodesic circles of radius $\eps>0$. Let $x\in M$ and let $\mc{C}_g(x,\eps)$ be the geodesic circle of center $x$ and radius $\eps>0$, and let $(f^n_{x,\eps})_{n\in \N} \in C^\infty(M)$ be a sequence with $||f^n_{x,\eps}||_{L^1}=1$ 
which is given by $f_{x,\eps}^n=\theta^n(d_g(x,\cdot)/\eps)$ where $\theta^n(r)\in C_c^\infty((0,2))$ non-negative 
supported near $r=1$ such that $f^n_{x,\eps}{\rm dv}_g$ 
is converging in $\mc{D}'(M)$ to the uniform probability measure 
$\mu_{x,\eps}$
on $\mc{C}_g(x,\eps)$ as $n\to \infty$ (for $\epsilon$ small enough, the geodesic circles form a submanifold and the restriction of $g$ along this manifold gives rise to a finite measure, which corresponds to the uniform measure after renormalization so as to have mass $1$).
Then we have the standard 

\begin{lemma}\label{Xeps}
Assume $g$ is hyperbolic. The random variable $\cjg X_g,f^n_{x,\eps}\cjd$ converges in $L^2(\Omega)$  to a random variable as $n\to \infty$, which has a modification $X_{g,\eps}(x)$ with   continuous sample paths with respect to   $(x,\eps)\in M\x (0,\eps_0)$, with covariance 
\begin{equation}\label{covXge}
\mathbb{E}[X_{g,\eps}(x)X_{g,\eps}(x')]=2\pi \int G_g(y,y')d\mu_{x,\eps}(y)d\mu_{x',\eps}(y')
\end{equation}
and we have as $\eps\to 0$
\begin{equation}\label{devptE}
\mathbb{E}[X_{g,\eps}(x)^2]=- \log(\eps)+W_g(x)+o(1)
\end{equation}
where  $W_g$ is the smooth function on $M$ given by $W_g(x)=2\pi m_g(x,x)$ if $m_g$ is the smooth function of Lemma \ref{greenneardiag}.
\end{lemma}

\begin{remark}
As a continuous Gaussian process,  the law of $X_{g,\epsilon}$ is characterized by expectation and covariance. In particular, \eqref{covXge} shows that the law of $X_{g,\epsilon}$ does not depend on the regularization scheme, namely the choice of the functions $(\theta^n)_n$.
\end{remark}

\begin{proof} Let us fix $x,\eps$, then if $Y_n:=\cjg X_{g},f^n_{x,\eps}\cjd$, it suffices to show that $\mathbb{E}(Y_nY_{n'})$ has a limit as $(n,n')\to \infty$ to prove that $Y_n$ is a Cauchy sequence in $L^2(\Omega)$.  
Using Lemma \ref{greenneardiag} (and its notation): 
\begin{equation}\label{integral0}\begin{split}
\mathbb{E}(Y_nY_{n'})=&2\pi\int_{M\x M} G_g(y,y')f^n_{x,\eps}(y)f^{n'}_{x,\eps}(y'){\rm dv}_g(y)
{\rm dv}_g(y')\\
=& \int_{M\x M} (- \log(d_g(y,y'))+2\pi m_g(y,y'))\theta^n_\eps (d_g(x,y))\theta_\eps^{n'}(d_g(x,y')){\rm dv}_g(y)
{\rm dv}_g(y').
\end{split}\end{equation}
Clearly the term 
\[\int_{M\x M}m_g(y,y')\theta^n_\eps (d_g(x,y))\theta_\eps^{n'}(d_g(x,y')){\rm dv}_g(y)
{\rm dv}_g(y')\] is uniformly bounded in $(n,n',\eps)$ and, as $(n,n')\to \infty$, it converges
to 
\[\int_{\mc{C}(x,\eps)}\int_{\mc{C}(x,\eps)}m_g(y,y')d\mu_{x,\eps}(y)d\mu_{x,\eps}(y')\]
which in turn is smooth in $x$ and converges, as $\eps\to 0$, to $m_g(x,x)$ uniformly in $x$. 
For $\eps>0$ small enough, we can use an isometry $\psi$ between a small geodesic ball 
$B_g(x,3\eps)$ of radius $3\eps$ and the ball $B_{\hh^2}(0,3\eps)$ in $\hh^2$ viewed as the disk model, 
so that the integral \eqref{integral0} above reduces to an integral in $B_g(x,3\eps)$ in both $y,y'$. Using the coordinates $z\in\hh^2$ induced by $\psi$ and \eqref{logdg}, 
\[\begin{gathered}
\int_{M\x M} \log(d_g(y,y'))\theta^n_\eps (d_g(x,y))\theta_\eps^{n'}(d_g(x,y')){\rm dv}_g(y)
{\rm dv}_g(y')=\\
\int_{[0,1]^2\x [0,2\pi]^2} \Big(\log |2\tanh(\tfrac{r}{2})e^{i(\alpha-\alpha')}-2\tanh(\tfrac{r'}{2})|
+L\Big)\theta^n(\tfrac{r}{\eps})
\theta^{n'}(\tfrac{r'}{\eps})d\alpha d\alpha'\sinh(r)\sinh(r')drdr'.
\end{gathered}\]
where $L=L(r,r',e^{i\alpha},e^{i\alpha})$ is continuous and $L(0,0,\cdot,\cdot)=0$. The term involving $L$ is clearly uniformly bounded in 
$(n,n')$ and $\eps$ and converges just like for $m_g$ above, and its limit as $\eps\to 0$ is $0$. The part with the log term is also straightforward to deal with
and is also uniformly bounded in $(n,n')$ for fixed $\eps>0$ and we get 
\[\begin{gathered}
\int_{[0,1]^2\x [0,2\pi]^2} \log |2\tanh(\tfrac{r}{2})e^{i(\alpha-\alpha')}-2\tanh(\tfrac{r'}{2})|
\theta^n(\tfrac{r}{\eps})\theta^{n'}(\tfrac{r'}{\eps})d\alpha d\alpha'\sinh(r)\sinh(r')drdr'\\
\underset{(n,n')\to \infty}{\longrightarrow} \log |2\tanh(\tfrac{\eps}{2})|+\frac{1}{4\pi^2}\int_{[0,2\pi]^2} \log|e^{i(\alpha-\alpha')}-1|\, d\alpha d\alpha'=
\log |2\tanh(\tfrac{\eps}{2})|.
\end{gathered} \]
We then have shown the convergence of $\cjg X_g,f^n_{x,\eps}\cjd$ towards a random variable $\til{X}_{g,\eps}(x)$ in $L^2(\Omega)$. To show it has a modification $X_{g,\eps}(x)$ that is sample continuous in $(x,\eps)\in M\x (0,\eps_0)$, it suffices to apply Kolmogorov  multi-parameter continuity theorem exactly like in the proof of \cite[Prop. 3.1]{cf:DuSh}, we do not repeat the argument. 
The variance $\mathbb{E}(X_{g,\eps}(x)^2)$
is smooth in $x$ and behaves like $- \log(\eps)+2\pi m_g(x,x)+o(1)$ 
as $\eps\to 0$, uniformly in $x$.
\end{proof}

Next from Lemma \ref{Xeps}, we will be able to define the Gaussian Multiplicative Chaos (GMC) first considered by Kahane \cite{cf:Kah}. The reader may also consult \cite{cf:DuSh,review,RoVa,shamov} on the topic, and in particular we recommend \cite{berestycki} for the simplicity of the approach.

\begin{proposition}\label{GMCprop} Assume $g$ is hyperbolic. Then the following hold true:\\
1) Let $\gamma>0$, the random measures $\mc{G}_{g,\eps}^{\gamma}:= \eps^{\frac{\gamma^2}{2}}e^{\gamma X_{g,\eps}(x)}{\rm dv}_g(x)$ converge in probability and weakly in the space of Radon measures towards a random measure $\mc{G}_g^\gamma(dx)$. The measure $\mc{G}_g^\gamma(dx)$ is non zero if and only if $\gamma \in (0,2)$.   \\
2) One obtains a non trivial random measure that we will denote $\mc{G}_{g}^{2}$ in the case $\gamma=2$  by considering the limit in probability and in the sense of weak convergence of measures of the family of random measures $\mc{G}_{g,\eps}^{\gamma}:= (-\ln\eps)^{1/2}\eps^{\frac{\gamma^2}{2}}e^{\gamma X_{g,\eps}(x)}{\rm dv}_g(x)$.
\end{proposition}

\begin{remark}
An important feature of GMC theory is that the law of the limiting random measure $\mc{G}_g^\gamma(dx)$ does not depend on the regularization scheme, namely the way the GFF $X_g$ has been smoothened   to produce a regularized field $X_{g,\epsilon}$. Universality of GMC has been investigated at various degree of generality in the papers \cite{cf:Kah,RoVa,cf:DuSh,review,shamov,berestycki}.
\end{remark}

\begin{proof}
The proof is standard for convolution based regularizations  of log-correlated Gaussian fields (first considered in \cite{RoVa}, see \cite{shamov} for latest results) in the case $\gamma<2$. The case $\gamma=2$ is treated in \cite[Section 5]{DRV} in the case of tori, relying on the result proved in \cite{Rnew7,Rnew12} for GFF with Dirichlet boundary conditions. The same argument applies for general compact $2d$-surfaces up to cosmetic modifications.
 
Here we give a simple argument in the case $\gamma<\sqrt{2}$ for the convenience of readers who are not familiar with GMC. Using the expression \eqref{devptE},
it suffices to study the convergence of the measures
\[e^{\gamma X_{g,\eps}(x)-\frac{\gamma^2}{2}\E[X_{g,\eps}(x)^2]}e^{\frac{\gamma^2}{2}W_g(x)}{\rm dv}_g(x).\]
Then by Fubini we directly get for each Borel set $A\subset M$, with $d\sigma:=e^{\frac{\gamma^2}{2}W_g(x)}{\rm dv}_g(x)$
\[ \E[\mc{G}_{g,\eps}^{\gamma}(A)]=\int_{A}\E[e^{\gamma X_{g,\eps}(x)-\frac{\gamma^2}{2}\E[X_{g,\eps}(x)^2]}
]d\sigma(x)= \sigma(A).\]
Using that there is $C>1$ such that for all $z\in\cc, |z|<1$   and $\eps>0$ small
\[ 1/C+|\log(|z|+\eps)|\leq   \int_{0}^{2\pi} \Big|\log |z+\eps e^{i\alpha}|\Big| d\alpha \leq C+|\log(|z|+\eps)| \]
then the arguments in the proof of Lemma \ref{Xeps} and the expression \eqref{logdg} imply that there is $C'$ such that 
\begin{equation}\label{uniformest}   
1/C'+ |\log (d_g(x',x)+\eps)| \leq \E[X_{g,\eps}(x)X_{g,\eps'}(x')]\leq C'+ |\log (d_g(x',x)+\eps)|.
\end{equation}
for all $\eps'\leq \eps$, and all $x,x'\in M$. In particular we get by using Fubini and the fact that $X_g$ is a Gaussian free field
\[ \begin{split}
 \E[\mc{G}^\gamma_{g,\eps}(A)^2]=&\E \Big[ \Big(\int_A e^{\gamma X_{g,\eps}(x)-\frac{\gamma^2}{2}\E[X_{g,\eps}(x)^2]}d\sigma(x)\Big)^2\Big]\\
 =& \E \Big[ \int_{A}\int_A e^{\gamma (X_{g,\eps}(x)+X_{g,\eps}(x'))-
 \frac{\gamma^2}{2}(\E[X_{g,\eps}(x)^2]+\E[X_{g,\eps}(x')^2])}d\sigma(x)d\sigma(x')\Big]\\
 = & \int_A\int_A e^{\gamma^2\E[X_{g,\eps}(x)X_{g,\eps}(x')]}d\sigma(x)d\sigma(x')
\end{split}\] 
which converges to $\int_{A}\int_A e^{\gamma^2 2\pi G_g(x,x')}d\sigma(x)d\sigma(x')<\infty
$ as $\eps\to 0$ by using \eqref{uniformest} and Lebesgue theorem - the condition $\gamma^2<2$ appear here due to the log divergence of $2\pi G_g(x,x')$ at $x=x'$, see Lemma \ref{greenneardiag}. A similar argument and \eqref{uniformest} also show that  $\E[(\mc{G}^\gamma_{g,\eps}(A)-\mc{G}^\gamma_{g,\eps'}(A))^2]\to 0$ if 
$(\eps,\eps')\to 0$, thus $\mc{G}^\gamma_{g,\eps}(A)$ is a Cauchy sequence, which therefore converges in $L^2(\Omega)$ to a random variable $Z(A)$, of mean $\sigma(A)$. By standard arguments, $\mc{G}^\gamma_{g,\eps}(dx)$ converges to a random measure $\mc{G}^\gamma_{g}$ satisfying $\E[\mc{G}^\gamma_{g}(A)]=\sigma(A)$. The case $\gamma\in [\sqrt{2},2)$ is more complicated and several methods have been proposed in the literature. We refer to    \cite{berestycki} for a simple argument.
\end{proof}

In fact, the whole construction above is not so particular to choosing the hyperbolic metric: indeed it uses only the fact that the covariance of $X_g$ is the Green function, the fact that near the diagonal 
$2\pi G_g(x,x')=-\log d_g(x,x')+ F(x,x')$ with $F$ continuous, and finally the fact that in local isothermal coordinates $z$ so that $g=e^{2f(z)}|dz|^2$ 
\[ \log d_g(z,z')=\log |z-z'|+2f(z)+o(1), \quad |z-z'|\to 0.\]
This allows to define a random measure $\mc{G}_{\hat{g}}^\gamma$ just as above for any other 
metric $\hat{g}=e^{\omega}g$ conformal to the hyperbolic metric $g$.
For later purpose we will need to make the following observation. 
If $\hat{g}=e^\omega g$, define 
\begin{equation}\label{Xhat}
\hat{X}_{g,\eps}(x):=\lim_{\eps\to 0}\cjg X_g, \hat{f}_{x,\eps}^n\cjd_{\hat{g}}
\end{equation} 
for each $x\in M$ where $\hat{f}_{x,\eps}^n:=\theta^n(d_{\hat{g}}(x,\cdot)/\eps)$ with $\theta^n$ like above, so that $\hat{f}_{x,\eps}^n {\rm dv}_{\hat{g}}$ converge as $n\to \infty$ to the uniform probability measure $\hat{\mu}_{x,\eps}$
on the geodesic circle $\mc{C}_{\hat{g}}(x,\eps)$ of center $x$ and radius $\eps$ with respect to $\hat{g}$.
In isothermal coordinates at $x$ so that $z=0$ correspond to the point $x$ and the metric is 
$g=|dz|^2/{\rm Im}(z)^2$, the circle $\mc{C}_{\hat{g}}(x,\eps)$ is parametrized by
\[  \eps e^{-\demi\omega(z)+\eps h_\eps(\alpha)}e^{i\alpha}, \alpha \in [0,2\pi] 
\]
for some continuous function $h_\eps(\alpha)$ uniformly bounded in $\eps$.
Then one has 
\[\mathbb{E}(\hat{X}_{g,\eps}(x)\hat{X}_{g,\eps}(x'))=2\pi \int G_g(y,y')d\hat{\mu}_{x,\eps}(y)d\hat{\mu}_{x',\eps}(y')\]
and by the arguments in the proof of Lemma \ref{Xeps}, we have as $\eps\to 0$  
\begin{equation}\label{devptE'}
\mathbb{E}(\hat{X}_{g,\eps}(x)^2)=- \log(\eps)+W_g(x)+\demi \omega(x)+o(1).
\end{equation}
 Then by the arguments of Proposition \ref{GMCprop}, the random measure  (add an extra push $(-\ln \eps)^{1/2}$ when $\gamma=2$)
 \begin{equation}\label{Ggamma}
 \hat{\mc{G}}^\gamma_{g,\eps}:= \eps^{\frac{\gamma^2}{2}}e^{\gamma \hat{X}_{g,\eps}(x)}{\rm dv}_{\hat{g}}(x)\end{equation}
converges weakly as $\eps\to 0$ to some measure $\hat{\mc{G}}^\gamma_{g}$ which satisfies
\begin{equation}\label{relationentrenorm} 
d\hat{\mc{G}}^\gamma_{g}(x)=e^{(1+\frac{\gamma^2}{4})\omega(x)}d\mc{G}^\gamma_g(x).
\end{equation}

\section{Liouville Quantum field theory with fixed modulus}
In this section we  define Liouville Quantum Field Theory (LQFT) with fixed conformal class (also called \emph{modulus}) and describe its main properties. 
It follows the approach of \cite{DKRV} in the case of the Riemann sphere. 
 Liouville Quantum Gravity (LQG) with fixed genus is a sum, called \emph{partition function}, 
over all possible metrics on a surface with fixed genus. The space of metrics splits into conformal classes 
and we want to decompose the partition function accordingly. Each conformal class has 
a unique hyperbolic metric, which plays the role of a background metric.

\subsection{Axiomatic of CFT}\label{CFT}
Here we give a brief account of the axiomatic of Conformal Field Theories in order to motivate the forthcoming results. Our purpose will then be to construct   the quantum Liouville theory and show that it satisfies this axiomatic. The reader is referred  to \cite{gaw} for more details related to this formalism. 

A CFT (on the surface $(M,g)$) is described by its partition function $Z(g)$ as well as the correlation functions of the (spinless) primary fields $(\theta_i)_{i\in I}$ denoted by
$$Z(g,\theta_{i_1}(x_1),\dots,\theta_{i_n}(x_n))$$
where $n\geq 1$, $\{i_1,\dots,i_n\}\in I$ and $x_1,\dots , x_n$ are arbitrary points on $M$.
Let us just roughly say that a CFT is supposed to give sense to "random fields" defined on $M$, here the primary fields  $(\theta_i)_i$, and the correlation functions can be thought of as the cumulants of these random fields. These correlation functions are supposed to satisfy the following conditions:\\
$\bullet$ \textbf{Diffeomorphism covariance:} for any orientation preserving diffeomorphism $\psi$
\begin{align}
  Z( g)&=Z(\psi^*g)\label{diffeo1}\\
  Z(g,\theta_{i_1}(\psi(x_1)),\dots,\theta_{i_n}(\psi(x_n)))&=Z(\psi^*g,\theta_{i_1}(x_1),\dots,\theta_{i_n}(x_n))\label{diffeo2}
\end{align}
$\bullet$ \textbf{Conformal anomaly:} for any smooth function $\omega$ on $M$
\begin{align}
 \ln\frac{Z( e^{\omega}g)}{Z( g)}  &= \frac{\mathbf{c}}{96\pi} \int_{M} (|d_g\omega |_{g}^2+2K_{g}\omega) {\rm dvol}_g   \label{scale1}\\
  \ln\frac{ Z(e^{\omega} g,\theta_{i_1}(x_1),\dots,\theta_{i_n}(x_n))}{Z(g,\theta_{i_1}(x_1),\dots,\theta_{i_n}(x_n))}&= -\sum_{k=1}^n \Delta_{i_k}\omega(x_k) +\frac{\mathbf{c}}{96\pi}  \int_{M} (|d_g\omega |_{g}^2+2K_{g}\omega) {\rm dvol}_g\label{scale2}
\end{align}
where the constant $\mathbf{c}$ is the so-called {\it central charge} of the CFT and each real number  $\Delta_i$ (for  $i\in I$) is called  the {\it conformal weight } of the primary field $\theta_i$.

One of the interesting feature of CFTs is their strong algebraic structure, which make them fall under the scope of techniques for integrable systems, leading to the possibility of  obtaining exact expressions for the correlation functions.

\subsection{The partition function of LQFT}
The first step is to describe LQFT with fixed modulus. 
LQFT will describe the probability law of some  random conformal factor, i.e. we consider the random metrics $e^{\gamma X}g$ where $g$ is a fixed   metric and $X$ is a random function.   The law of $X$ will be mathematically described by the   measure \eqref{liouvillemesure}. So, let $g\in{\rm Met}(M)$ be a fixed metric on $M$. The mathematical definition of  the  LQFT measure   (i.e. \eqref{liouvillemesure}) is the following. Fix  $\gamma\in(0,2]$. For $F:  H^{-s}(M)\to\R$ (with $s>0$) a bounded continuous functional, set  
\begin{align}\label{partLQFT}
 \Pi_{\gamma, \mu}(g,F):=& ({\det}'(\Delta_{g})/{\rm Vol}_{g}(M))^{-1/2}  \\
 &\times \int_\R  \E\Big[ F( c+  X_{g}) \exp\Big( -\frac{Q}{4\pi}\int_{M}K_{g}(c+ X_{g} )\,{\rm dv}_{g} - \mu  e^{\gamma c}\mc{G}_{g}^\gamma(M)  \Big) \Big]\,dc .\nonumber
\end{align}
This quantity, if it is finite, gives a mathematical sense to the formal integral 
\[  \int F(\varphi)e^{-S_L(g,\varphi)}D\varphi\]
where $S_L(g,\varphi)$ is the Liouville action \eqref{QLiouville}. The partition function is the total mass of this measure, i.e $\Pi_{\gamma, \mu}(g,1)$. 

\begin{proposition}\label{totalmass}
For $g\in  {\rm Met}(M)$ and  $\gamma\in(0,2]$, we have $0<\Pi_{\gamma, \mu}(g,1)<+\infty$ and 
the mapping $$F\in C_b(H^{-s}(M),\R)\mapsto  \Pi_{\gamma, \mu}(g,F)$$ defines a positive finite measure. When renormalized by its total mass, it describes the law of a random variable living in $H^{-s}(M)$ called the {\bf Liouville field}. When $g\in{\rm Met}_{-1}(M)$ is hyperbolic, we further have 
\begin{equation}\label{explicit}
 \Pi_{\gamma, \mu}(g,1)= \Big(\frac{{\det}'(\Delta_{g})}{{\rm Vol}_{g}(M)}\Big)^{-1/2}   \gamma^{-1}\mu^{\frac{Q\chi(M)}{\gamma} }\Gamma(-\tfrac{Q\chi(M)}{\gamma}) \E\Big[  \mc{G}_{g}^\gamma(M)  ^{ \frac{Q\chi(M)}{\gamma}} \Big]  
 \end{equation}
where $\Gamma(z)$ is the standard Euler Gamma function.
\end{proposition} 
\begin{proof} The proof of this proposition follows the same lines as in \cite[section 3.1]{DKRV}. 
We consider the case of a metric $g\in {\rm Met}_{-1}(M)$, since 
the general case follows from this case, as is explained below in Proposition \ref{limitcorel} for the correlations functions. In constant curvature, the Gauss-Bonnet theorem entails
$$\frac{Q}{4\pi}\int_{M}K_{g}(c+ X_{g} )\,{\rm dv}_{g}=Qc \chi(M) $$
where $\chi(M)$ is the Euler characteristic  of $M$ and we get
$$\Pi_{\gamma, \mu}(g,1)= \Big(\frac{{\det}'(\Delta_{g})}{{\rm Vol}_{g}(M)}\Big)^{-1/2}    \int_\R  e^{-Qc \chi(M) } \E\Big[   \exp\Big( -\mu  e^{\gamma c}\mc{G}_{g}^\gamma(M)  \Big) \Big]\,dc .$$
After inverting expectation and integration, and using the change of variables $y=\mu e^{\gamma c}\mc{G}_{\hat{g}}^\gamma(M) $, we get \eqref{explicit}. Finiteness of this quantity is ensured by the fact that GMC has finite moments of negative orders as $\chi(M)<0$ - finiteness of negative moments is proved  for example in \cite[Proposition 3.6]{RoVa} for $\gamma<2$ and in \cite[Corollary 14]{Rnew12} in the case $\gamma=2$.
\end{proof}

\subsubsection{Conformal anomaly and diffeomorphism invariance}
Here we investigate the symmetries of the measure \eqref{partLQFT} and in particular how the partition function reacts to changes of background metrics. The following proposition is the quantum counterpart of \eqref{conformS}.
\begin{proposition}\label{covconf}{\bf (Conformal anomaly)}
Let $Q=\frac{\gamma}{2}+\frac{2}{\gamma}$ with  $\gamma\in(0,2]$  and $g\in {\rm Met}_{-1}(M)$ be a hyperbolic metric on $M$.
The partition function satisfies the following conformal anomaly:  
if $\hat{g}=e^{\omega}g$ for some $\omega\in C^\infty(M)$, we have
\[ \Pi_{\gamma, \mu}(\hat{g},F)=\Pi_{\gamma, \mu}(g,F(\cdot\,-\tfrac{Q}{2}\omega))\exp\Big(\frac{1+6Q^2}{96\pi}\int_{M}(|d\omega|_g^2+2K_g\omega) {\rm dv}_g\Big).\]
\end{proposition}

\begin{proof}  We focus on the integral part in \eqref{partLQFT} (and hence let   the determinant of Laplacian  apart as its contribution is clear from \eqref{detpolyakov}). First, by Lemma \ref{invarianceconforme}, we can replace $X_{\hat{g}}$ by $X_g$ in the expression defining 
$\Pi_{\gamma, \mu}(\hat{g},F)$ and are thus left with considering the following quantity (with  $\hat{\mc{G}}_g^\gamma$ is the measure defined by \eqref{Ggamma})  
\[A:=\int_\R  \E\Big[ F( c+  X_{g})  \exp\Big( -\frac{Q}{4\pi}\int_{M}K_{\hat{g}}(c+ X_{g}  )\,{\rm dv}_{\hat{g}} - \mu  e^{\gamma c} \hat{\mc{G}}^\gamma_g(M)    \Big) \Big]\,dc.\]
By \eqref{GB}, the term $-\frac{Qc}{4\pi}\int_{M}K_{\hat{g}}\,{\rm dv}_{\hat{g}}$ can be written as $-Qc\chi(M)$ where $\chi(M)$ is the Euler characteristic. Define the Gaussian random variable 
\[ Y:= -\frac{Q}{4\pi}\int_{M}K_{\hat{g}}X_{g} \,{\rm dv}_{\hat{g}}=-\frac{Q}{4\pi}\cjg X_{g},K_{\hat{g}}e^{\omega}\cjd_g.\] 
Let $R_g(0)$ be the resolvent operator whose Schwartz kernel is $G_g$ with respect to ${\rm dv}_g$. 
Since $ K_{\hat{g}}e^{\omega}=\Delta_g\omega+K_g$, we compute, using that $R_g(0)K_g=0$ (as $K_g=-2$),
\[\begin{split} 
\E[ \cjg X_{g},K_{\hat{g}}e^{\omega}\cjd_g^2]=& 2\pi\cjg G_g,(\Delta_g\omega+K_g)\otimes(\Delta_g\omega+K_g)\cjd_g\\
=& 2\pi \cjg \omega-\tfrac{\cjg\omega,1\cjd_g}{{\rm Vol}_g(M)}, \Delta_g\omega-2\cjd_g= 2\pi \int_M|d\omega|_g^2{\rm dv}_g
\end{split}\]
and similarly we have 
\[
\E[YX_g]=-\frac{Q}{2} R_g(0)(K_{\hat{g}}e^\omega)=-\frac{Q}{2}(\omega-c_g(\omega))
\]
if $c_g(\omega):=\tfrac{\cjg\omega,1\cjd_g}{{\rm Vol}_g(M)}$. Thus we get 
\begin{equation}\label{rayas2005}
\demi \E[Y^2]= \frac{Q^2}{16\pi}\int_M|d\omega|_g^2{\rm dv}_g, \quad \E[YX_{g}]= -\frac{Q}{2}(\omega-c_g(\omega)).
\end{equation}
Therefore by applying Girsanov transform to the random variable $Y$, we can rewrite 
\[A=\int_\R  e^{\demi \E[Y^2]-Qc\,\chi(M)}\E\Big[ F( c+  X_{g}+\E(Y X_{g})) \exp\Big( - \mu  e^{\gamma (c+\frac{Q}{2}c_g(\omega))}  \int_M e^{-\frac{\gamma Q}{2}\omega}\,d \hat{\mc{G}}_g^\gamma \Big) \Big]\,dc.\]
With the help of the relation  \eqref{relationentrenorm} and $Q=\frac{\gamma}{2}+\frac{2}{\gamma}$, we see that $ \int_M e^{-\frac{\gamma Q}{2}\omega}\,d \hat{\mc{G}}_g^\gamma= \mc{G}_g^\gamma(M)$. Using \eqref{rayas2005},  $A$ can be written as 
\[A=\int_\R  e^{\frac{Q^2}{16\pi}||d\omega||^2_{L^2_g}-Qc\chi(M)}\E\Big[ F( c+  X_{g} -\tfrac{Q}{2}\omega+\tfrac{Q}{2}c_g(\omega))) \exp\Big( - \mu  e^{\gamma (c+\tfrac{Q}{2}c_{g}(\omega))}\mc{G}_g^\gamma(M) \Big) \Big]\,dc.\]
It remains to make the change of variable $c\to c-\tfrac{Q}{2}c_g(\omega)$ and we deduce that 
\[A=\int_\R  e^{\frac{Q^2}{16\pi}||d\omega||^2_{L^2_g}-Qc\chi(M)+\demi Q^2\chi(M)c_g(\omega)}\E\Big[ F( c+  X_{g}-\frac{Q}{2}\omega) \exp\Big( - \mu  e^{\gamma c}\mc{G}_g^\gamma(M) \Big) \Big]\,dc.\]
Since $K_g=-2$ and ${\rm Vol}_g(M)=-2\pi\chi(M)$ we have 
\[-\frac{Q}{4\pi}\int_{M}K_{g}(c+ X_{g} )\,{\rm dv}_{g}=-Qc\chi(M), \quad  
c_{g}(\omega)\chi(M)=\frac{1}{4\pi}\int_M K_g\omega\, {\rm dv}_g\]
which shows that $A=\Pi_{\gamma, \mu}(g,F(\cdot\,-\tfrac{Q}{2}\omega))\sqrt{\det'(\Delta_{g})/{\rm Vol}_{g}(M)}e^{\frac{6Q^2}{96\pi}\int_{M}(|d\omega|_g^2+2K_g\omega) {\rm dv}_g}$. Combining with \eqref{detpolyakov}, the proof is complete.
\end{proof}

The constant ${\bf c}_L:=1+6Q^2$ describing the conformal anomaly is called the \emph{central charge} of the Liouville Theory.
Since all the objects in the construction of the Gaussian Free Field and the Gaussian multiplicative chaos are geometric (defined in a natural way from the metric), it is direct to get the following diffeomorphism invariance:
\begin{proposition}\label{diffeoinv} {\bf (Diffeomorphism invariance)}
Let  $g\in {\rm Met}(M)$  be a   metric on $M$ and let $\psi:M\to M$ be an orientation preserving diffeomorphism. Then we have for each bounded measurable $F:H^{-s}(M)\to \R$ with $s>0$
\[ \Pi_{\gamma, \mu} (\psi^*g ,F)= \Pi_{\gamma, \mu}(g,F(\cdot \circ \psi)) .\]
\end{proposition}

\begin{proof} This follows directly from the fact that all the object considered in the construction of the measure are natural with respect to the metric $g$, thus invariant by isometries: more precisely, it follows from the identities
$$G_{\psi^*g}(x,y)=G_g(\psi(x),\psi(y)),\quad K_{\psi^*g}(x)=K_g(\psi(x)),\quad X_{\psi^*g} \stackrel{law}{=}X_g\circ\psi,$$
which are standard.
\end{proof}

The two above results show that the axioms \eqref{diffeo1}+\eqref{scale1} are satisfied with central charge $\mathbf{c}_L=1+6Q^2$. Yet we still have to define the primary fields and their correlation functions. This is the purpose of the next subsection.

\subsection{The correlation functions}\label{correlfct}
The correlation functions of  LQFT can be thought of  as the exponential moments  $e^{\alpha \varphi(x)}$ of the random function $\varphi$, the law of which is ruled by the path integral \eqref{liouvillemesure},  evaluated at some location $x\in M$ with weight $\alpha$.  Yet, the field $\varphi$ is not a well-defined function as it belongs to $H^{-s}(M)$ for $s>0$, so that the construction requires some care.  

As before let $g\in {\rm Met}(M)$. We fix $n$ points $x_1,\dots,x_n$ ($n\geq 0$) on $M$ with respective associated weights $\alpha_1,\dots,\alpha_n\in\R$. We denote ${\bf x}=(x_1,\dots,x_n)$ and ${\boldsymbol \alpha} =(\alpha_1,\dots,\alpha_n)$. 
The rigorous definition of the primary fields will require a regularization scheme. We introduce the following $\eps$-regularized functional
\begin{align}\label{actioninsertion}
 &\Pi_{\gamma, \mu}^{{\bf x},{\boldsymbol  \alpha}}  (g,F,\eps):=\big({\det }'(\Delta_{g})/{\rm Vol}_{g}(M)\big)^{-1/2} \\
 & \int_\R  \E\Big[ F(c+  X_{g} ) \big(\prod_i V^{\alpha_i}_{g,\eps}(x_i)\big)    \exp\Big( -\frac{Q}{4\pi}\int_{M}K_{g}(c+ X_{g}  )\,{\rm dv}_{g} - \mu  e^{\gamma c}
 \mc{G}_{g,\eps}^\gamma(M)  \Big) \Big]\,dc 
 \nonumber
\end{align}
where we have set, for fixed $\alpha\in\R$ and $x\in M$,
\begin{equation*}
V^{\alpha}_{g,\eps}(x)=\eps^{\alpha^2/2}  e^{\alpha  (c+X_{g,\epsilon}(x)) } .
\end{equation*}
Here the regularization is the one described in Lemma \ref{Xeps}. Such quantities are called {\rm vertex operators}.
 Notice that $V_{g,\epsilon}^\alpha$ also depends on the variable $c$ but we have dropped this dependence in the notations.   
 
Then, the point is to determine whether the limit 
$$  \Pi_{\gamma, \mu}^{{\bf x},\boldsymbol \alpha}  (g,F):=\lim_{\eps\to 0} \Pi_{\gamma, \mu}^{{\bf x},\boldsymbol \alpha}  (g,F,\eps)$$
exists  and defines a non trivial functional  on those mappings $F:H^{-s}(M)\to \R$. If it does, the quantity $ \Pi_{\gamma, \mu}^{{\bf x},\boldsymbol \alpha}  (g):= \Pi_{\gamma, \mu}^{{\bf x},\boldsymbol \alpha}  (g,1)$ stands for the $n$-point correlation function of the primary fields $(e^{\alpha_i \varphi})_{1\leq i \leq n}$ respectively evaluated at $(x_i)_{1\leq i\leq n}$. Furthermore, another quantity of interest is the probability law on $H^{-s}(M)$ defined by the measure
$$F\in C_b(H^{-s}(M))\mapsto \Pi_{\gamma, \mu}^{{\bf x},\boldsymbol \alpha}  (g,F)/\Pi_{\gamma, \mu}^{{\bf x},\boldsymbol \alpha}  (g),$$
which describes the law of some formal "random function" (it is in fact a distribution). 

We obtain a result similar to \cite{DKRV} (done for the sphere).
\begin{proposition}\label{limitcorel} Let ${\bf x}=(x_1,\dots,x_n)\in M^n$ and ${\bf \alpha}=(\alpha_1,\dots,\alpha_n)\in\R^n$.
Then for all bounded continuous functionals $F: h \in  H^{-s}(M) \to F(h) \in\R$ with $s>0$, the limit
 \begin{equation*}
  \Pi_{\gamma, \mu}^{{\bf x},\alpha }  (g,F) := \lim_{\eps\to 0}  \Pi_{\gamma, \mu}^{{\bf x},\alpha}  (g,F,\eps) ,
\end{equation*}
 exists and is finite with $\Pi_{\gamma, \mu}^{{\bf x},\alpha }  (g,1)>0$, if and only if:
 \begin{align}\label{seiberg1}
 & \sum_{i}\alpha_i + 2 Q ({\bf g}-1)>0,\\ 
 &\forall i,\quad \alpha_i<Q\label{seiberg2}.
 \end{align}
In the case $g\in{\rm Met}_{-1}(M)$, we have the following expression
 \begin{equation*}
 \Pi_{\gamma, \mu}^{{\bf x},\alpha}  (g) = \big(\frac{\det '(\Delta_{g})}{{\rm Vol}_{g}(M)}\big)^{-\frac{1}{2}}   e^{C(\mathbf{x})}     \:   \mu^{-\sum_{i}\alpha_i-2Q({\bf g}-1)}\Gamma( \sum_i\alpha_i+ 2 Q ({\bf g}-1)) \:   \E\Big[      \mathcal{G}^\gamma_{g, {\bf  x},{\bf \alpha}}(M)^{-\frac{\sum_i\alpha_i +2 Q ({\bf g}-1)}{\gamma}} \Big] 
 \end{equation*} 
 where $\Gamma$ is Euler gamma function and, if $W_g$ is the function appearing in Lemma \ref{Xeps}, 
 \begin{equation}\label{Z0Cx}
\begin{gathered} 
\mathcal{G}^\gamma_{g, {\bf  x},{\bf \alpha}}(dx):= e^{\gamma \sum_i \alpha_i 2\pi G_{g}(x_i,x)}  \mc{G}_{g}^{\gamma}(dx), \\
C(\mathbf{x}) :=\sum_i\frac{\alpha_i^2}{2}W_{g}(x_i)+2\pi\sum_{i<j}\alpha_i\alpha_j G_{g}(x_i,x_j).
\end{gathered} \end{equation}
 \end{proposition}
 
 \begin{remark}
The reader may compare with the correlation functions of free scalar fields, see \cite[Equations (2.90) and (2.93)]{DhPh3} for instance.
 \end{remark}
 
 \begin{proof}
 The   argument goes essentially as in the proof of \cite[Theorems 3.2 \& 3.4]{DKRV}, 
 while having in mind that  the Gauss-Bonnet theorem is \eqref{GB} on general compact Riemann surfaces.
 We recall the main steps. It suffices to prove the claim for $F=1$. We fix $g$ hyperbolic and we will also consider the case of $\hat{g}=e^{\omega}g$ for $\omega\in C^\infty(M)$ to understand the behaviour of the correlation functions under conformal change. Consider  \eqref{actioninsertion} for the metric $\hat{g}$. By Lemma \ref{invarianceconforme}, we can replace $X_{\hat{g}}$ by $X_{g}$ in this expression -- in particular $V^{\alpha_i}_{\hat{g},\eps}(x_i)$ becomes 
$\hat{V}^{\alpha_i}_{g,\eps}(x_i):=\eps^{\alpha_i^2/2}  e^{\alpha_i  (c+\hat{X}_{g,\epsilon}(x_i))}$, with $\hat{X}_{g,\epsilon}$ defined by \eqref{Xhat}.
First we notice by \eqref{devptE'} that 
\[ \hat{V}^{\alpha_i}_{g,\eps}(x_i)= e^{\alpha_ic+\frac{\alpha_i^2}{4}\omega(x_i)+\frac{\alpha_i^2}{2}W_{g}(x_i)} e^{\alpha_i \hat{X}_{g,\eps}(x_i)-\frac{\alpha_i^2}{2}\E[\hat{X}_{g,\eps}(x_i)^2]}(1+o(1))\]
as $\eps\to 0$, with the remainder being deterministic. Here we have used the notation
$\hat{X}_{g,\eps}(x_i)=\cjg X_g,\hat{\mu}_{x_i,\eps}\cjd$ as before, where $\hat{\mu}_{x_i,\eps}$ is the uniform probability measure on the Riemannian circle $\mc{C}_{\hat{g}}(x_i,\eps)$.
Then applying Girsanov transform   in the expression 
\[ A_\eps:= \int_\R  \E\Big[ \big(\prod_i \hat{V}^{\alpha_i}_{g,\eps}(x_i)\big)    \exp\Big( -\frac{Q}{4\pi}\int_{M}K_{\hat{g}}(c+ X_{g} )\,{\rm dv}_{\hat{g}} - \mu  e^{\gamma c}
 \hat{\mc{G}}_{g,\eps}^\gamma(M)  \Big) \Big]\,dc\]
to the Radon-Nikodym derivative $\prod_{i=1}^ne^{\alpha_i \hat{X}_{g,\eps}(x_i)-\frac{\alpha_i^2}{2}\E[\hat{X}_{g,\eps}(x_i)^2]}$, we get
\[\begin{split} 
A_\eps= & e^{C_\eps({\bf x})}
  \int_\R  e^{c(\sum_i\alpha_i-Q\chi(M))}\E\Big[  \exp\Big( -\frac{Q}{4\pi}\cjg X_{g},K_{\hat{g}}\cjd_{\hat{g}}  - \mu  e^{\gamma c}
 \hat{Z}_\eps  \Big) \Big]\,dc\, (1+o(1))
\end{split}\]
where 
\[\begin{gathered} 
\hat{Z}_\eps:= \eps^{\frac{\gamma^2}{2}}\int_M e^{\gamma (\hat{X}_{g,\eps}+H_{g,\eps})}{\rm dv}_{\hat{g}}\\
H_{g,\eps}(x):=\sum_i 2\pi \alpha_i \int_{\mc{C}_{\hat{g}}(x_i,\eps)}G_{g}(y,x)d\hat{\mu}_{x_i,\eps}(y), \\
\quad C_\eps({\bf x}):=2\pi\sum_{i\not=j}\alpha_i\alpha_j G_{g}(x_i,x_j)-\frac{Q}{4\pi}\int_{M}K_{\hat{g}}H_{g,\eps}{\rm dv}_{\hat{g}}+\sum_i \frac{\alpha_i^2}{4}(\omega(x_i)+2W_g(x_i)).
\end{gathered}\]
Notice that, since $K_{\hat{g}}{\rm dv}_{\hat{g}}=(\Delta_g\omega-2){\rm dv}_{g}$, we have  as $\eps\to 0$
\[C_\eps({\bf x})\to \pi\sum_{i\not=j}\alpha_i\alpha_j G_{g}(x_i,x_j)+\sum_{i}(\frac{\alpha_i^2}{4}-
\frac{Q\alpha_i}{2})\omega(x_i)+\frac{Q}{2}\sum_i\alpha_ic_{g}(\omega)+\demi\sum_{i}\alpha_i^2W_g(x_i).\]
By applying Girsanov transform again just like in the proof of Proposition \ref{covconf}, we can get rid of the 
$\cjg X_{g},K_{\hat{g}}\cjd_{\hat{g}}$ term and this shifts the field $\hat{X}_{g,\eps}$ in $\hat{Z}_\eps$ 
by $F(x)=-\tfrac{Q}{2}(\omega(x)-c_g(\omega))$: 
\[\begin{split} 
A_\eps= & e^{C_\eps({\bf x})+\frac{Q^2}{16\pi}||d\omega||^2_{L^2_g}}  \int_{\R} e^{c(\sum_i\alpha_i-Q\chi(M))}
\E\Big[  \exp\Big( - \mu  e^{\gamma c}
 \til{Z}_\eps  \Big) \Big]\,dc\, (1+o(1))
\end{split}\]
where $\til{Z}_\eps:=\eps^{\frac{\gamma^2}{2}}\int_M e^{\gamma (\hat{X}_{g,\eps}+H_{g,\eps}+F)}{\rm dv}_{\hat{g}}$; here we have denoted $c_g(\omega)=\cjg \omega,1\cjd_g/{\rm Vol}_g(M)$.
By Lemma \ref{Xeps},  $||H_{g,\eps}||_{L^\infty}<\infty$ thus by Proposition \ref{GMCprop}, we get that 
$\E[\til{Z}_\eps]<\infty$. Therefore we can find $B>0$ such that $\mathbb{P}(\til{Z}_\eps\leq B)>0$. We therefore get 
\[ A_\eps\geq \beta_{\eps,{\bf x}} \int_{-\infty}^0 e^{c(\sum_i\alpha_i-Q\chi(M)) - \mu  e^{\gamma c}} \mathbb{P}(\til{Z}_\eps\leq B)\,dc
\]
for some $\beta_{\eps,{\bf x}}>0$, and this is infinite if $\sum_i\alpha_i -Q\chi(M)\leq 0$. Then we assume  \eqref{seiberg1}. 
We also have as in \eqref{relationentrenorm} the relation
\[ \til{Z}_\eps= e^{\frac{\gamma Q}{2}c_g(\omega)}  \mathcal{G}^{\gamma,\epsilon}_{g,{\bf  x},{\bf \alpha}}(M)(1+o(1)), \quad  \mathcal{G}^{\gamma,\epsilon}_{g,{\bf  x},{\bf \alpha}}(M)=\eps^{\frac{\gamma^2}{2}}\int_M e^{\gamma H_{g,\eps}}d\mc{G}^\gamma_{g,\eps}.\]
Making the change of variables $c\to c-\tfrac{Q}{2}c_g(\omega)$, we obtain that $A_\eps$ is equal to
\[e^{C({\bf x})+\sum_{i}(\frac{\alpha_i^2}{4}-
\frac{Q\alpha_i}{2})\omega(x_i)+\frac{Q^2}{16\pi}||d\omega||^2_{L^2_g}+\frac{Q^2}{2}\chi(M)c_g(\omega)} 
\int_{\R} e^{c(\sum_i\alpha_i-Q\chi(M))}
\E\Big[  \exp\Big( - \mu  e^{\gamma c}
  \mathcal{G}^{\gamma,\epsilon}_{g,{\bf  x},{\bf \alpha}}(M)  \Big) \Big]\,dc\,\]
times $1+o(1)$ as $\eps\to 0$, where $C({\bf x})$ is given by \eqref{Z0Cx}. In particular this implies \eqref{confan} if we can show that 
for the case $\omega=0$ the limit of $A_\eps$ is finite. We now assume $\omega=0$, or equivalently we consider $\hat{g}=g$ the hyperbolic metric. We make the change of variables 
$c=\mu e^{\gamma c} \mathcal{G}^{\gamma,\epsilon}_{g,{\bf  x},{\bf \alpha}}(M)$ in the $c$-integral defining $A_\eps$ (recall that $ \mathcal{G}^{\gamma,\epsilon}_{g,{\bf  x},{\bf \alpha}}(M)>0$ almost surely), and we get 
\[ A_\eps=\gamma^{-1}e^{C({\bf x})}\mu^{\frac{-\sum_{i}\alpha_i+Q\chi(M)}{\gamma}}\Gamma\Big(\frac{\sum_{i}\alpha_i-Q\chi(M)}{\gamma}\Big)\E[  \mathcal{G}^{\gamma,\epsilon}_{g,{\bf  x},{\bf \alpha}}(M)^{-\frac{\sum_{i}\alpha_i-Q\chi(M)}{\gamma}}].
\]
It remains to show that if $\alpha_i<Q$ for all $i$ and $\delta:=\frac{\sum_{i}\alpha_i-Q\chi(M)}{\gamma}>0$, then 
\begin{equation}\label{EZeps}
\lim_{\eps\to 0}\E[  \mathcal{G}^{\gamma,\epsilon}_{g,{\bf  x},{\bf \alpha}}(M)^{-\delta}]=\E[\mathcal{G}^\gamma_{g,{\bf  x},{\bf \alpha}}(M)^{-\delta}] \in (0,\infty), \quad \mathcal{G}^\gamma_{g,{\bf  x},{\bf \alpha}}(M)=\int_Me^{\gamma H_{g}}d\mc{G}^\gamma_{g}
\end{equation}  
with $H_{g}(x):=\lim_{\eps\to 0}H_{g,\eps}(x)=2\pi \sum_i\alpha_iG_{g}(x_i,x)$, and that if $\alpha_i\geq Q$ for some $i$, then $\E[ \mathcal{G}^{\gamma,\epsilon}_{g,{\bf  x},{\bf \alpha}}(M)^{-\delta}]\to 0$. But this part is only a local argument and therefore Lemma 3.3. of \cite{DKRV} applies directly. The argument goes essentially as follows. The term $e^{\gamma H_{g}}$ behaves like  $\frac{1}{d_g(x,x_i)^{\gamma\alpha_i}}$ in the neighborhood of $x_i$ and thus we need to determine whether the measure  $\mc{G}_{g}^{\gamma}(dx)$ integrates the singularity $\frac{1}{d_g(x,x_i)^{\gamma\alpha_i}}$ in the neighborhood of $x_i$ to get non-triviality of the random variable $\mathcal{G}^\gamma_{g,{\bf  x},{\bf \alpha}}(M)$ (if the singularity is not integrable, we get $Z_0=+\infty$ a.s. and $\lim_{\eps\to 0}\E[ \mathcal{G}^{\gamma,\epsilon}_{g,{\bf  x},{\bf \alpha}}(M)^{-\delta}]=0$). Standard multifractal analysis shows that for any $\delta>0$ one can find a constant $C_\delta$ such that
 $$\sup_{r<1}r^{-\gamma Q+\delta}\mc{G}_{g}^{\gamma}(B_r(x_i))\leq C_\delta$$ where $B_r(x_i)$ stands for the geodesic ball of radius $r$ centered at $x_i$. This gives the condition $\alpha_i<Q$ for non-triviality of  $\mathcal{G}^\gamma_{g,{\bf  x},{\bf \alpha}}(M)$. 
 Finally it remains to determines whether the quantity
 $$\int_\R e^{c\big(\sum_i\alpha_i-2 Q(1-\mathbf{g})\big)}\E[e^{-\mu e^{\gamma c} \mathcal{G}^\gamma_{g,{\bf  x},{\bf \alpha}}(M)}]\,dc$$ is finite. As we have seen that $\mathcal{G}^\gamma_{g,{\bf x},{\bf \alpha}}(M)$ is a well defined non trivial random variable under the condition \eqref{seiberg2}, one may think of it as a macroscopic quantity and replace it by a constant quantity, say $1$, so as to be left with the integral
\[\int_\R e^{c\big(\sum_i\alpha_i-2 Q(1-\mathbf{g})\big)-\mu e^{\gamma c} }\,dc,\]
 which is easily seen to be converging if and only if \eqref{seiberg1} holds. This is only a sketch of proof but details are exposed in  \cite[Lemma 3.3]{DKRV}. 
 \end{proof}
 
The proof of the previous proposition (adding a functional $F$ does not change anything) also shows the 
\begin{proposition}\label{covconf2}{\bf (Conformal anomaly and diffeomorphism invariance)}
Let 
$g$ be a hyperbolic metric on $M$ and $\hat g=e^{\omega}g$ for some $\omega\in C^\infty(M)$, and let ${\bf x}=(x_1,\dots,x_n)\in M^n$ and ${\bf \alpha}=(\alpha_1,\dots,\alpha_n)\in\R^n$. Then we have
\begin{equation}\label{confan} 
\log\frac{\Pi_{\gamma, \mu}^{{\bf x},{\bf \alpha}}  (\hat{g},F)}{\Pi_{\gamma, \mu}^{{\bf x},\alpha}  (g,F(\cdot-\tfrac{Q}{2}))}= 
\frac{1+6Q^2}{96\pi}\int_{M}(|d\omega|_g^2+2K_g\omega) {\rm dv}_g-\Delta_{\alpha_i}\omega(x_i)
\end{equation}
where the real numbers $\Delta_{\alpha_i}$, called {\it conformal weights}, are defined by the relation $\Delta_{\alpha}:=\frac{Q\alpha}{2}-\frac{\alpha^2}{4} $ for $\alpha\in\R$\footnote{The reader may compare \eqref{confan} with the general axiomatic of CFTs exposed in Subsection \ref{CFT}.}.
Let $\psi:M\to M$ be an orientation preserving diffeomorphism. Then  
\[ \Pi_{\gamma, \mu}^{{\bf x},\boldsymbol\alpha}   (\psi^*g ,F)= \Pi_{\gamma, \mu}^{{\bf \psi(x)},\boldsymbol\alpha}  (g,F(\cdot \circ \psi)) .\]
\end{proposition}

\section{Liouville Quantum gravity}
\subsection{The full partition function}\label{fullpartfct}

The partition function of Liouville quantum gravity is a weighted integral over the moduli space of the Liouville quantum field theory  coupled to a Conformal Field Theory  (sometimes  called matter field in physics in this context). The weight  of each modulus is given by some explicit functional $ Z_{\textrm{Ghost}}(g) $ (this weight depends on the underlying surface $(M,g) $), called the ghost system in physics. This takes into account the factorization of the space of metrics by the action of the group of diffeomorphisms  of the surface (as 
explained for example in \cite{DhPh}).

Let us first recall the physics heuristics that leads to the partition function, by following \cite{Pol, DhPh, DistKa, cf:Da}; the following discussion is not mathematically rigorous but is rather a ``state of the art" in physics literature.
The partition function for (Euclidean) quantum gravity in 2D, coupled with matter, is
\[ Z= \int_{\mc{R}}e^{-S_{{\rm EH}}(g)}\Big(\int e^{-S_{\rm M}(g,\phi_m)}D_g\phi_m\Big) Dg\]
where $\mc{R}={\rm Met}(M)/{\rm Diff}(M)$ is the space of Riemannian structures, the action $S_{{\rm EH}}(g)=\mu_0{\rm Vol}_g(M)$ is the Einstein-Hilbert action (or gravity action) with $\mu_0\in\mathbb{R}$ the cosmological constant and the matter fields $\phi_m$ are elements of an infinite dimensional space $E$ of fields living over $M$ (typically $\phi_m$ are sections of some bundles over $M$)
with $S_{\rm M}(g,\phi_m)$ being the action for matter which depends on $g$ in a conformally invariant way. Notice that, in comparison with \eqref{introEH}, we got rid of the term $\int_MK_g\,d{\rm v}_g$ as it is a topological invariant in $2d$ because of the Gauss-Bonnet theorem: this is an important feature of $2d$-quantum gravity.
The quantity 
 \[Z_{\rm M}(g):=\int e^{-S_{\rm M}(g,\phi_m)}D_g\phi_m\]
is supposed to be a CFT with central charge ${\bf c_{\rm M}}$, $D_g\phi_m$ the formal Riemannian measure induced by the $L^2$ Riemannian metric on the space on fields $E$  and $Dg$ is the formal Riemannian measure induced by the $L^2$ Riemannian metric on ${\rm Met}(M)$ given by \eqref{L2metric} (the group ${\rm Diff}(M)$ acts by isometries on ${\rm Met}(M)$ thus the $L^2$-metric on ${\rm Met}(M)$ descends to $\mc{R}$).

Each metric can be decomposed as $g=\psi^*(e^{\varphi}g_\tau)$
where $\tau$ is a parameter on moduli space $\mc{M}_{\bf g}$, $g_\tau$ is a family of metrics representing moduli space and $\psi\in {\rm Diff}(M)$, and the formal measure $Dg$ can be accordingly 
decomposed as  
\[ Dg=Z_{\rm Ghost}(e^{\varphi}g_\tau) D_{e^\varphi g_\tau}\varphi D\tau \]
where $Z_{\rm Ghost}$ is the ghost determinant which comes from the Jacobian of the quotient of ${\rm Met}(M)$ by the group of diffeomorphism ${\rm Diff}(M)$ (see for example \cite{DhPh}), and given by 
\[ Z_{\textrm{Ghost}}(g)= \Big(\frac{{\det}(P_{g}^*P_{g})}{\det J_{g}}\Big)^{1/2}\]
where $P_g,J_g$ are defined in Section \ref{detoflap}. The ghost determinant satisfies the conformal anomaly formula \eqref{scale1} with central charge
 \begin{equation}\label{CCG} 
 \mathbf{c}_{\mathbf{ghost}}=-26.
 \end{equation}
Here $D\tau$ is a measure on the slice of metrics $g_\tau$ chosen to represent moduli space, whose value is $D\tau:=(\det J_{g_\tau})^{1/2} d\tau$ with $d\tau$ being the Weil-Petersson volume form on the moduli space $\mathcal{M}_{\mathbf{g}}$ (somehow $Z_{\rm Ghost}(e^{\varphi}g_\tau) D\tau$ is the quantity that makes invariant sense, as it does not depend on the matrix $J_g$). The formal measure 
$D_{e^{\varphi} g_\tau}\varphi$ should be induced by the $L^2$ Riemannian metric on metrics, which on the tangent space to the conformal orbit $[g_\tau]=\{e^{\varphi}g_\tau; \varphi\in C^\infty(M)\}$ is given by 
\begin{equation}\label{measDt}
 ||f||_{e^{\varphi}g_\tau}^2= \int_M \omega^2 e^{\varphi}{\rm dv}_{g_\tau}, \quad f=\omega e^{\varphi}g_\tau \in T_{e^{\varphi} g_\tau}[g_\tau].
\end{equation}
This measure depends non-linearly on $\varphi$ and it is difficult to ``do the functionnal integral'' for this measure, as written in \cite{DistKa}. Therefore  David and Distler-Kawai \cite{cf:Da,DistKa} made the 
``well-motivated'' assumption that 
\[  e^{-S_{\rm EH}(g)}Z_{\rm M}(g)Dg= Z_{\rm M}(g_\tau)Z_{\rm Ghost}(g_\tau) e^{-S_L(g_\tau,\varphi)} D\tau 
D_{g_\tau}\varphi
 \]
where $S_L(g,\varphi)$ is the Liouville action defined by \eqref{QLiouville} for some parameter 
$Q,\gamma,\mu$ to be chosen and $D_{g_\tau}\varphi$ is the formal Riemannian measure induced by the $L^2$-metric \eqref{measDt} at $g_\tau$. A formal justification of this fact was written down later in \cite{MaMi} and \cite{DhPh,DhKu}. 
Invariance of the theory by choice of slice $g_\tau$ representing moduli space forces the partition function
$\int e^{-S_L(g_\tau,\varphi)} 
D_{g_\tau}\varphi$ to be a CFT  
with central charge ${\bf c_{\rm L}}=1+6Q^2$ such that the total conformal anomaly   vanishes:
\[ \mathbf{c}_{\mathbf{ghost}}+ \mathbf{c}_{\mathbf{M}}+ \mathbf{c}_{\mathbf{L}}=0.\]
Recalling that $\gamma\in ]0,2]$ is  related to $Q$ by $Q=\gamma/2+2/\gamma$, this forces  $\mathbf{c}_{\mathbf{M}}\leq 1$ and we obtain another KPZ relation \cite{cf:KPZ}
\begin{equation}\label{KPZ2}
\gamma=\frac{\sqrt{25-\mathbf{c}_{\mathbf{M}}}-\sqrt{1-\mathbf{c}_{\mathbf{M}}}}{\sqrt{6}},
\end{equation}
which fixes the value of $\gamma$ in terms of $\mathbf{c}_{\mathbf{M}} $.

Now we stop the physics parenthesis and come back to mathematics. 
For the matter field, we take the particular case the most studied in the physics literature, namely 
\begin{equation}\label{Zm}
Z_{\textrm{M}}(g):= \Big(\frac{{\det }'(\Delta_{g})}{{\rm Vol}_{g}(M)}\Big)^{-\mathbf{c}_{\mathbf{M}}/2}
\end{equation}
where $\mathbf{c}_{\mathbf{M}}$ is a constant in $(-\infty,1]$. 
Note that this has the central charge ${\bf c_{\rm M}}$ by \eqref{detpolyakov}.  Furthermore, there are at least two important particular cases: pure gravity where ${\bf c_{\rm M}}=0$ and the $2d$ bosonic string in the case ${\bf c_{\rm M}}=1$. Because it is the critical situation of this approach, the latter case is especially interesting and raises serious additional difficulties.  
One could consider also other CFT partition functions provided that we get an expression explicit enough to determine how it behaves at the boundary of the moduli space.
For each modulus $\tau\in \mc{M}_{\bf g}$, we can associate a hyperbolic metric $g_\tau$ and we will denote by $(g_\tau)_\tau$ the family of hyperbolic metrics representing the moduli space.
By definition, the partition function of Liouville quantum gravity is given by the following formula:
 \begin{equation}\label{QGpartition} 
Z:=  \int_{\mathcal{M}_{\mathbf{g}}}    Z_{\textrm{Ghost}}(g_\tau)   \times Z_{\textrm{M}}(g_\tau) \times  \Pi_{\gamma,\mu}(g_\tau)    \:  D\tau 
 \end{equation}
 where  $D\tau:=(\det J_{g_\tau})^{1/2} d\tau$ with $d\tau$ the Weil-Petersson volume form on the moduli space $\mathcal{M}_{\mathbf{g}}$, and $\Pi_{\gamma,\mu}(g)$  is the partition function of the Liouville quantum field theory. 
This can be reduced to   
  \begin{equation}\label{QGpartition2} 
Z=  C_{{\bf g}}\int_{\mathcal{M}_{\mathbf{g}}}    {\det}(P_{g_\tau}^*P_{g_\tau})^{1/2}\x   {\det }'(\Delta_{g_\tau})^{-{\bf c_M}/2}
 \times  \Pi_{\gamma,\mu}(g_\tau)    \:  d\tau  
 \end{equation}
with ${\rm C}_{{\bf g}}$ a constant depending only on the genus of $M$. We point out that the reduction of the partition function under the form \eqref{QGpartition2} was first derived by \cite{DhPh}, at least for the critical string case.

Now, the main result of this section is the following: (we denote by ${\rm Rad}(M) $ the space of Radon measures over $M$ in the statement below)
\begin{theorem}\label{defLiouvillemeas}
If $\gamma \in (0,2]$ and $ \mathbf{c}_{\mathbf{M}}$ satisfies relation \eqref{KPZ2}, the partition function $Z$ given by \eqref{QGpartition} is finite. Hence it gives rise to  a finite measure $\nu$ on   $ {\rm Rad}(M)\x \mc{M}_{\bf g}$ defined as follows: 
 if $(g_\tau)_{\tau}$ is a family of hyperbolic metrics parametrizing the moduli space $\mc{M}_{\bf g}$, then
 \[\begin{aligned}
 \nu(F):= 
 \int_{\mc{M}_{\bf g}\x \mathbb{R}}\Big(\frac{{\det}(P_{g_\tau}^*P_{g_\tau})}{({\det }'\Delta_{g_\tau})^{\mathbf{c}_{\mathbf{M}}+1}}\Big)^{\demi} \E\big[F(e^{\gamma c} \mc{G}_{g_\tau}^\gamma(dz),\tau)e^{-Q\chi(M) c-\mu   e^{\gamma c}\mc{G}_{g_\tau}^\gamma(M)} \Big]  \,d\tau\,dc 
 \end{aligned}\]
for all continuous functionals $F: {\rm Rad}(M)\x \mc{M}_{\bf g}\to \mathbb{R}$. When renormalized by its total mass $Z=\nu(1)$, it becomes a probability measure which we call  $\P_{  (g_\tau)_{\tau} , \mu}$ (with expectation $\E_{  (g_\tau)_{\tau} , \mu} $) and the couple $(e^{\gamma c} \mc{G}_{g_\tau}^\gamma(dz),\tau)$   becomes a random variable on  $ {\rm Rad}(M)\x \mc{M}_{\bf g}$, which we denote by $(\mathcal{L}_{\gamma},R)$ and stands for the volume form of the space (called Liouville quantum gravity measure) and its modulus (called quantum modulus). 

\medskip

Furthermore, for  all continuous functionals $F: {\rm Rad}(M)\x \mc{M}_{\bf g}\to \mathbb{R}$ 
\begin{align}
& \E_{  (g_\tau)_{\tau} , \mu}[   F(  \mathcal{L}_{\gamma}(dz)  , R ) ] \label{defLQGmeasure}    \\
& = \frac{\Gamma(\frac{2Q ({\bf g}-1)}{\gamma})}{   \gamma Z \mu^{\frac{2Q ({\bf g}-1)}{\gamma}} }    \int_{\mathcal{M}_{\mathbf{g}}}     \Big(\frac{{\det}(P_{g_\tau}^*P_{g_\tau})}{({\det }'\Delta_{g_\tau})^{\mathbf{c}_{\mathbf{M}}+1}}\Big)^{\demi}  \E  \left [ F  \left (  \xi_\gamma \frac{\mc{G}_{g_\tau}^\gamma(dz)}{\mc{G}_{g_\tau}^\gamma(M)  } , \tau  \right )   \mc{G}_{g_\tau}^\gamma(M)^{\frac{Q}{\gamma}\chi(M) }     \right ]   \: d\tau   \nonumber
\end{align}
where $\xi_\gamma$ is a random variable with Gamma law of density $\frac{\mu^{\frac{2Q ({\bf g}-1)}{\gamma}}}{\Gamma(\frac{2Q ({\bf g}-1)}{\gamma})}e^{-\mu x}  x^{\frac{2Q ( {\bf g}-1  )}{\gamma}-1} \mathbf{1}_{x\geq 0}$ and the random modulus $R$ has density 
$$  {\det}(P_{g_\tau}^*P_{g_\tau})^{\demi}\Big(\frac{{\det }'(\Delta_{g_\tau})}{{\rm Vol}_{g_\tau}(M)}\Big)^{-\frac{(\mathbf{c}_{\mathbf{M}}+1)}{2}} \E  \left [    \mc{G}_{g_\tau}^\gamma(M)^{\frac{Q}{\gamma}\chi(M) }     \right ]  $$
with respect to the $d\tau$ measure. 
\end{theorem}

Let us make some comment on the above result. The LQG measure  depends on the family of hyperbolic metrics $(g_\tau)_{\tau}$ but this  is not an issue since it enjoys the following invariance by reparametrization: if $(\psi_{\tau})_{\tau}$ is a family of orientation preserving diffeomorphisms, we get the following equality for all $\tau$
\begin{equation}\label{invdiffeo}
\E_{  (\psi_{\tau}^{*}g_\tau)_{\tau} }[   F(  \mathcal{L}_{\gamma} \circ \psi_\tau  )  | R = \tau ]= \E_{  (g_\tau)_{\tau} }[   F(  \mathcal{L}_{\gamma}  ) |  R = \tau ].
\end{equation}

\subsection{Proof of Theorem \ref{defLiouvillemeas} in the case $\gamma=2$ }
 
  Our purpose is to determine the behaviour of the partition function \eqref{QGpartition} of LQFT (for $\gamma=2$) at the boundary of the moduli space with $Z_{{\rm M}}(g)$ defined by \eqref{Zm} and  $\mathbf{c}_{  \mathbf{M}}=1$. According to the relation \eqref{explicit}, this amounts to showing that the integral 
\begin{equation}\label{partg2}
 \int_{\mathcal{M}_{\mathbf{g}}}     \Big(\frac{{\det}(P_{g_\tau}^*P_{g_\tau})}{({\det }'\Delta_{g_\tau})^{\mathbf{c}_{\mathbf{M}}+1}}\Big)^{\demi}  \E  \left [  \mc{G}_{g_\tau}^\gamma(M)^{\frac{Q}{\gamma}\chi(M) }     \right ]   \: d\tau 
\end{equation}
  is finite.  The singularities in this integral appear at the boundary of the moduli space, namely when the surface $(M,g)$ gets close to a surface with nodes $(M_0,g_0)$    by pinching $n_p$  geodesics with respective  lengths $(\ell_j)_j$ on $(M,g)$   (see section \ref{sectionsmall}). According to the explicit bounds  \eqref{Wolpert1} and \eqref{Wolpert2} for the product $ \Big(\frac{{\det}(P_{g_\tau}^*P_{g_\tau})}{({\det }'\Delta_{g_\tau})^{\mathbf{c}_{\mathbf{M}}+1}}\Big)^{\demi} $ and the expression for the Weil-Petersson measure \eqref{WPvol}, we can give an upper bound
\begin{equation}\label{product}
 C \Big(\prod_{j'=1}^{n_p}\ell_{j'}^{-2}\prod_{\lambda_i<1/4}\lambda_i ^{-1}\Big)\Big(\prod_{j=1}^{3\mathbf{g}-3}\ell_j\, d\ell_j d\theta_j\Big)
 \end{equation} 
   for the quantity
  $ \Big(\frac{{\det}(P_{g_\tau}^*P_{g_\tau})}{({\det }'\Delta_{g_\tau})^{\mathbf{c}_{\mathbf{M}}+1}}\Big)^{\demi} \,d\tau$   in the coordinate system $(\ell_j,\theta_j)_{j=1,\dots,3\mathbf{g}-3}$ associated to a pant decomposition.
Hence   it suffices to   check the integrability  of the expectation $\E\big[\mathcal{G}_g^\gamma(M)^{\frac{Q\chi(M)}{\gamma}}\big]$ with respect to the measure \eqref{product} near  $(M_0,g_0)$. It turns out that the mass of the  measure $\mathcal{G}_g^\gamma(M)$ will become very large when $g$ gets close to $g_0$ in the pinched region of the surface $(M,g)$, making the expectation $\E\big[\mathcal{G}_g^\gamma(M)^{\frac{Q\chi(M)}{\gamma}}\big]$ very small. We have to quantify the rate of decay of this expectation   to show integrability. Proposition \ref{lastestimateGg}    describes how the Green function behaves near the pinched geodesics. The purpose of what follows is to explain how these estimates transfer to the above expectation.  

\medskip
We assume $(M_0,g_0)$ possesses $m+1$ connected components, and we consider a neighborhood of $(M_0,g_0)$; there is a subpartition  $\{\gamma_1,\dots,\gamma_{n_p}\}$   of $M$ corresponding to the pinched geodesics. Let us denote by $(S_j)_{j=1,\dots,m+1}$ the connected components of $M\setminus (\cup_{j'=1}^{n_p}\gamma_{j'})$. Each $\gamma_{j'}$ has a collar neighborhood denoted $\mc{C}_{j'}$ (see \eqref{collarCj}, for simplicity of notation  we  remove the $g$ dependance in $\mc{C}_{j'}(g)$). We denote by $S'_j$ the set obtained by removing from $S_j$  all the collars
$$S'_j:=\bigcap_{j'=1}^{n_p}(S_j\setminus \mc{C}_{j'})$$
in such a way that  $M=\cup_{j=1}^{m+1}S'_j\cup_{j'=1}^{n_p}\mathcal{C}_{j'}$. Furthermore, for each $j=1,\dots,m+1$, we define $I_j=\{j'\in\big\{1,\dots,n_p\};S_j\cap \mc{C}_{j'}\not=\emptyset \big\}$ the set of indices $j'$ such that the collar $ \mc{C}_{j'}$ encounters $S_j$.

We define the following quantities for $x>0$ and $1\leq j'\leq n_p$
\begin{align*} 
 \mathcal{C}_{j'}^+:=\mathcal{C}_{j'}\cap\{\rho\geq 0\} &  &\mathcal{C}_{j'}^-:=\mathcal{C}_{j'}\cap\{\rho\leq 0\}\\
  \mathcal{C}_{j'}(x)^+:=\mathcal{C}_{j'}\cap\{x\leq \rho \}  & & \mathcal{C}_{j'}(x)^-:=\mathcal{C}_{j'}\cap\{  \rho\leq -x\},
\end{align*} 
here $\rho:\cup_{j'}\mc{C}_{j'}\to [-1,1]$ is the function in the collars so that the metric is given by \eqref{collarrho}. 
 Let us denote by $(\varphi_{i})_{1\leq i\leq m}$ the eigenfunctions associated to the small (non zero) eigenvalues   $(\lambda_i)_{1\leq i\leq m}$ and write   the Green function $G_g$ as $\sum_{1\leq i\leq m}\tfrac{1}{\lambda_i}\Pi_{\lambda_i}+A_g$. Consider  $X_g'$ the Gaussian field with covariance $2\pi A_g$, which is nothing but 
\begin{equation}\label{defX'}
X_g'=X_g-\sum_{1\leq i\leq m}(2\pi/\lambda_i)^{1/2}\frac{\langle X_g,\varphi_i \rangle}{(2\pi)^{1/2}}\varphi_i.
\end{equation}  
  The finite sequence $(\frac{\langle X_g,\varphi_i \rangle}{(2\pi)^{1/2}})_{1\leq i\leq m}$ is a sequence of i.i.d. standard Gaussian random variables, namely  with law $\mathcal{N}(0,1)$, which we denote $(a_i)_{1\leq i\leq m}$. Furthermore, $(a_i)_{1\leq i\leq m}$ and $X_g'$ are independent. In case the surface $(M_0,g_0)$ is not disconnected, write $A_g=G_g$ and $X'_g=X_g$. We introduce the random measure: $\forall A\subset M$ Borel set 
 $$\mathcal{G}'(A):=\lim_{\eps\to 0}(-\ln\eps)^{1/2}\eps^2\int_A e^{2 X'_{g,\eps}} \,d{\rm v}_g,$$
where $X'_{g,\eps}$ is the regularization of $X'_g$ as in Lemma \ref{Xeps}.   Notice that the convergence in probability of this measure is ensured by the convergence in probability of the same measure involving the field $X_g$ instead of $X'_g$ (both field coincide up to an additive continuous field so that convergence of one measure is equivalent to convergence of the other one). The main technical estimate we need in the proof is the following

\begin{lemma}\label{driftout}
 Let $U_0\subset \mc{M}_{\bf g}$ be a neighborhood of some metric with nodes $g_0$.  For any $j'$, $\delta>0$, $q>0$, $\lambda\in [0,1]$ there exists some constant $C$ such that for all $g\in U_0$,  $\forall\ell,\ell'\geq \ell_{j'}$ and $A,B>0$ and $\psi:\R\to\R$ of class  $C^1$ such that $\psi\circ F_{j'}(-1)=\psi\circ F_{j'}(1)=0$
  \begin{align*}
\E\Big[    \Big(\int_{\mathcal{C}_{j'}}\phi e^{\psi\circ F_{j'}}\,d\mathcal{G}'\Big)^{-q} \Big] \leq C_q A^{-q\lambda}B^{-q(1-\lambda)}(\ell\ell')^{\tfrac{1}{2}-\delta}\exp\Big(C\int_{1\leq |r|\leq \tfrac{2\pi}{\min(\ell,\ell')^{1-\delta}}}|\psi'(r)|^2\,dr \Big)
\end{align*}
where the function $\phi$ is defined on the collar $\mathcal{C}_{j'} $ by $\phi(\rho)=A\mathbf{1}_{\mathcal{C}_{j'}(\ell)^+}+B\mathbf{1}_{\mathcal{C}_{j'}(\ell')^-}$, and   $F_{j'}(\rho):=\tfrac{2\pi}{\ell_{j'}}\arctan(\tfrac{\ell_{j'}}{\rho})$. The constant $C_q$ depends on $q$ and the mapping $q\in [0,+\infty[\mapsto C_q$ is locally bounded.
\end{lemma}

 The proof of this lemma is postponed to the end of this subsection. As a direct consequence we claim   
 \begin{corollary}\label{Qpunc}
 For any $j'$, $\delta>0$ $q>0$, there exists some constant $C$ such that for all $g\in U_0$ and $\ell,\ell'\geq \ell_{j'}$ and $A,B>0$
\begin{align*}
&1) \quad \E\Big[   \mathcal{G}'( \mathcal{C}_{j'}(\ell)^+)^{-q}\Big]+\E\Big[   \mathcal{G}'( \mathcal{C}_{j'}(\ell)^-)  ^{-q}\Big]\leq C\ell^{1/2-\delta},\\
&2) \quad \E\Big[    (A \mathcal{G}'( \mathcal{C}_{j'}(\ell)^+)+B \mathcal{G}'( \mathcal{C}_{j'}(\ell')^-))^{-q}\Big] \leq C(AB)^{-q/2}(\ell  \ell')^{1/2-\delta}.
\end{align*}
 \end{corollary}

 Now we complete the proof while considering two main situations: either the surface $(M_0,g_0)$ is disconnected or not. 

$\bullet$ \textbf{Case $(M_0,g_0)$ is not  disconnected: } In that case the measure \eqref{product} can be estimated from above by the measure
 \begin{equation}\label{productnodisc}
 C \Big(\prod_{j'=1}^{n_p}\ell_{j'}^{-1}\Big) \prod_{j=1}^{3\mathbf{g}-3}d\ell_j d\theta_j.
 \end{equation} 
 Concerning the contribution of the GMC expectation in \eqref{partg2}, we claim
 \begin{lemma}\label{teclem1}
 Assume $(M_0,g_0)$ is not  disconnected. For any $\delta>0$,  there exists some constant $C$ such that for all $g\in U_0$, 
 \begin{align}\label{eqteclem1}
\E\big[\mathcal{G}_g^\gamma(M)^{\frac{Q\chi(M)}{\gamma}}\big]\leq &    C\prod_{j'=1}^{n_p}\ell_{j'}^{ 1-\delta}
\end{align}
 \end{lemma}
It is then clear that the estimate \eqref{eqteclem1} is integrable with respect to \eqref{productnodisc}, which completes our argument in the case when  $(M_0,g_0)$ is not  disconnected.
 
\noindent {\it Proof of Lemma \ref{teclem1}.} Recalling that $\chi(M)<0$ and using the elementary inequality $(a+b)^{-1}\leq b^{-1}$ for $a,b>0$ we have
 \begin{align*}
\E\big[\mathcal{G}_g^\gamma(M)^{\frac{Q\chi(M)}{\gamma}}\big]\leq & \E\big[  \big(\sum_{j'=1}^{n_p}\mathcal{G}'(\mathcal{C}_{j'}) \big)^{\frac{Q\chi(M)}{\gamma}}\big] .
\end{align*}
Now we observe that the cross covariances of the field $X'_g$ in the various regions $\mathcal{C}_{j'}$  are bounded by some uniform constant, i.e. $\sup_{g\in U_0}\sup_{j_1'\not=j_2'}\sup_{x\in \mathcal{C}_{j_1'},y\in \mathcal{C}_{j_2'}}|\E[X'_g(x)X'_g(y)]|\leq C$ (see Proposition \eqref{greenpinching}). Kahane's inequality \cite[Lemma 1]{cf:Kah} then tells us that, considering   independent copies $(\widehat{\mathcal{G}}'_{j'})_{1\leq j'\leq n_p}$ of $\mathcal{G}'$, the latter expectation is less than (for some irrelevant constant $C$)
$$C\E\big[  \big(\sum_{j'=1}^{n_p}\widehat{\mathcal{G}}'_{j'}(\mathcal{C}_{j'}) \big)^{\frac{Q\chi(M)}{\gamma}}\big].$$
Then, we use  the following elementary inequality valid for $b_1,\dots,b_n\geq 0$ and $w_1,\dots,w_n\geq 0$ with $\sum_{j=1}^n w_j=1$ and $q>0$  
\begin{equation}\label{elementary}
(\sum_{i=1}^nb_i)^{-q}\leq \prod_{i=1}^nb_i^{-w_i q}  
\end{equation} 
to deduce that the above expectation is less  than $\prod_{j'=1}^{n_p}\E\big[ \big(\widehat{\mathcal{G}}'_{j'}(\mathcal{C}_{j'}) \big)^{\frac{Q\chi(M)}{\gamma n_p}}\big]$. Combining with  Corollary \ref{Qpunc}, we obtain
\begin{align*}
\E\big[\mathcal{G}_g^\gamma(M)^{\frac{Q\chi(M)}{\gamma}}\big]\leq &    C\prod_{j'=1}^{n_p}\ell_{j'}^{ 1-\delta}
\end{align*}
 for some arbitrary $\delta>0$. \qed

\bigskip 
 
$\bullet$ \textbf{Case $(M_0,g_0)$ is  disconnected: }   In this case, finiteness of the integral   \eqref{partg2} restricted to a neighborhood of $(M_0,g_0)$ results from the combination of   \eqref{product} together with  the crude estimate \eqref{estimeevpSWY} on the eigenvalues $\la_j(g)$ in terms of the lengths 
 $\ell_k(g)$ and the following lemma:
 \begin{lemma}\label{teclem2}
 Assume $(M_0,g_0)$ is  disconnected. For any $\delta>0$,  there exists some constant $C$ such that for all $g\in U_0$, 
 \begin{align}\label{eqteclem2}
\E\big[\mathcal{G}_g^\gamma(M)^{\frac{Q\chi(M)}{\gamma}}\big]\leq &      C\prod_{i=1}^m\lambda_i^{1/2}\prod_{j'=1}^{n_p}\ell_{j'}^{1-\delta}.
\end{align}
 \end{lemma}

 \noindent {\it Proof. } Recall that the eigenfunctions $(\varphi_i)_{i=1,\dots,m}$ converge uniformly on the compact subsets of each $S_j$ respectively towards some fixed value denoted $v_{ij}$. From  Lemma \ref{boundvarphi1}, we have the estimate in the region ${|\rho|>C\ell_{j'}}$ of the cusp $\mathcal{C}_{j'}\cap S_j$ 
\begin{equation}\label{supbound}
\varphi^+_{ij} \leq \varphi_{i} \leq \varphi^-_{ij} 
\end{equation}
where we have set
\begin{align}
\varphi^+_{ij}  :=&v_{ij}(1-\mathcal{S}(v_{ij})C\epsilon)|\rho|^{-\lambda_i +C\mathcal{S}(v_{ij})\lambda_i^2}-C\epsilon\label{phij1}\\
\varphi^-_{ij}  :=&v_{ij}(1+\mathcal{S}(v_{ij})C\epsilon)|\rho|^{-\lambda_i -C\mathcal{S}(v_{ij})\lambda_i^2}+C\epsilon \label{phij2}
\end{align} 
 for some constant $C>0$; here we have denoted by  $\mathcal{S}$    the function $x\in\R\mapsto  \mathcal{S}(x):={\rm sign}(x)$.
 Restricting the integral   to the cusp regions and then using \eqref{supbound}, we have the estimate conditionally on the $  (a_i)_{1\leq i \leq m}$
\begin{align*}
\E\big[\mathcal{G}_g^\gamma(M)^{\frac{Q\chi(M)}{\gamma}}| (a_i)_{1\leq i \leq m}\big]\leq &    \E\big[\mathcal{G}_g^\gamma(\cup_{j'} \mathcal{C}_{j'})^{\frac{Q\chi(M)}{\gamma}}| (a_i)_{1\leq i \leq m}\big]\\
=&  \E\big[\big(\sum_{j'}\int_{\mathcal{C}_{j'}}e^{\sum_i2(2\pi/\lambda_i)^{1/2}a_i\varphi_i}d\mathcal{G}'\big)^{\frac{Q\chi(M)}{\gamma}}| (a_i)_{1\leq i \leq m}\big]\\
\leq &  \E\big[\big(\sum_{j,j'}\int_{\mathcal{C}_{j'}\cap S_j}e^{\sum_i2(2\pi/\lambda_i)^{1/2}a_i\varphi_{ij}^{\mathcal{S}(a_i)}}d\mathcal{G}'\big)^{\frac{Q\chi(M)}{\gamma}}| (a_i)_{1\leq i \leq m}\big],
\end{align*}
where $\varphi_{ij}^{\mathcal{S}(a_i)}=\varphi_{ij}^{+} $ if $a_i\geq 0$ and $\varphi_{ij}^{\mathcal{S}(a_i)}  =\varphi_{ij}^{-}$ if $a_i<0$. Once again we can use Kahane's convexity inequality to show that there exists a collection of mutually independent random measure $(\mathcal{G}'_{(j')})_{1\leq j'\leq n_p}$ (and independent of the $ (a_i)_{1\leq i \leq m} $) such that, for some irrelevant constant $C>0$
\begin{align*}
 \E\big[&\big(\sum_{j,j'}\int_{\mathcal{C}_{j'}\cap S_j}  e^{\sum_i2(2\pi/\lambda_i)^{1/2}a_i\varphi_{ij}^{\mathcal{S}(a_i)}}d\mathcal{G}'\big)^{\frac{Q\chi(M)}{\gamma}}| (a_i)_{1\leq i \leq m} \big]\\
& \leq   C \E\big[\big(\sum_{j,j'}\int_{\mathcal{C}_{j'}\cap S_j}e^{\sum_i2(2\pi/\lambda_i)^{1/2}a_i\varphi_{ij}^{\mathcal{S}(a_i)}}d\mathcal{G}'_{(j')}\big)^{\frac{Q\chi(M)}{\gamma}}| (a_i)_{1\leq i \leq m} \big].
\end{align*}

Now we choose a collection of real-valued random variables $(r_{j})_{1\leq j\leq m+1}$ that are  non negative with $\sum_{j} r_{j}=1$ and measurable with respect to the family $(a_i)_{1\leq i \leq m} $. The precise choice of the family $(r_{j})_{1\leq j\leq m+1}$ will be made later when appropriate.  Given those $(r_{j})_{1\leq j\leq m+1}$, we construct new weights $(w_{j'})_{1\leq j'\leq n_p}$ as follows.  For each $j=1,\dots,m+1$, we define $I^+_j=\{j'\in\big\{1,\dots,n_p\};S_j\cap \mc{C}^+_{j'}\not=\emptyset \big\}$ the set of indices $j'$ such that the collar $ \mc{C}_{j'}^+$ encounters $S_j$ and we define similarly $I^-_j$. Also, for each $j'=1,\dots,n_p$, there exists a unique $j$ such that $ \mc{C}^+_{j'}\subset S_j$ and we denote by $j'_+$ that index. Similarly for $j'_-$. Finally for $j'=1,\dots,n_p$, we define
\begin{equation}
w_{j'}^+=\frac{r_{j'_+}}{|I^+_{j'_+}|+|I^-_{j'_+}|},\quad w_{j'}^-=\frac{r_{j'_-}}{|I^+_{j'_-}|+|I^-_{j'_-}|}\quad \text{and }\quad w_{j'}=w_{j'}^++w_{j'}^-.
\end{equation}
The first observation is that
\begin{equation}\label{convex}
\sum_{j'=1,\dots,n_p}w_{j'}=1.
\end{equation}
Indeed 
\begin{align*}
\sum_{j'=1,\dots,n_p}w_{j'}=&\sum_{j=1}^{m+1}\sum_{j'\in I^+_j}w_{j'}^++\sum_{j=1}^{m+1}\sum_{j'\in I^-_j}w_{j'}^-\\
=&\sum_{j=1}^{m+1}\sum_{j'\in I^+_j}\frac{r_{j}}{|I^+_{j}|+|I^-_{j}|}+\sum_{j=1}^{m+1}\sum_{j'\in I^-_j}\frac{r_{j}}{|I^+_{j}|+|I^-_{j}|}\\
=&\sum_{j=1}^{m+1}r_j=1.
\end{align*}
Relation \eqref{convex} allows us to use   inequality \eqref{elementary}. Together with independence of the measures $(\mathcal{G}'_{(j')})_{1\leq j'\leq n_p}$ conditionally on the $(a_i)_{1\leq i \leq m}$, this yields
\begin{align*}
\E\big[&\mathcal{G}_g^\gamma(M)^{\frac{Q\chi(M)}{\gamma}} | (a_i)_{1\leq i \leq m} \big]\\
\leq &C\prod_{j'=1}^{n_p}  \E\big[ \big( \sum_j\int_{\mathcal{C}_{j'} \cap S_j }e^{\sum_i2(2\pi/\lambda_i)^{1/2}a_i\varphi_{ij}^{\mathcal{S}(a_i)}}d\mathcal{G}'  \big)^{w_{j'}\frac{Q\chi(M)}{\gamma}}| (a_i)_{1\leq i \leq m}\big].
\end{align*} 
 Notice that the sum over $j$ in the latter expression contains at most two non trivial terms as a cusp $\mathcal{C}_{j'} $ possesses at most two non trivial intersections with the $S_j$'s. More precisely 
 \begin{align*}\sum_j\int_{\mathcal{C}_{j'} \cap S_j }e^{\sum_i2(2\pi/\lambda_i)^{1/2}a_i\varphi_{ij}^{\mathcal{S}(a_i)}}d\mathcal{G}' = & \int_{\mathcal{C}_{j'}^-  }e^{\sum_i2(2\pi/\lambda_i)^{1/2}a_i\varphi_{ij'_-}^{\mathcal{S}(a_i)}}d\mathcal{G}'
+ \int_{\mathcal{C}_{j'}^+}e^{\sum_i2(2\pi/\lambda_i)^{1/2}a_i\varphi_{ij'_+}^{\mathcal{S}(a_i)}}d\mathcal{G}'.
\end{align*}

Now introduce $t_{j'}^+= F_{j'}^{-1}(1)$ and $t_{j'}^-= F_{j'}^{-1}(-1)$ and rewrite   the above integrals as
$$\int_{\mathcal{C}_{j'} }\phi e^{\psi\circ F_{j'}}d\mathcal{G}'$$
 in view of applying Lemma \ref{driftout}, with    $\ell=\ell'=\ell_{j'}$ 
 \begin{align*}
 \phi&= e^{\sum_i2(2\pi/\lambda_i)^{1/2}a_i\varphi_{ij'_+}^{\mathcal{S}(a_i)}(t_{j'}^+)}\mathbf{1}_{\mathcal{C}_{j'}^{+}}+e^{\sum_i2(2\pi/\lambda_i)^{1/2}a_i\varphi_{ij_-'}^{\mathcal{S}(a_i)}(t_{j'}^-)}\mathbf{1}_{\mathcal{C}_{j'}^-}\\
 \psi &=\sum_i2(2\pi/\lambda_i)^{1/2}a_i\Big((\varphi_{ij_-'}^{\mathcal{S}(a_i)}\circ F_{j'}^{-1}-\varphi_{ij_-'}^{\mathcal{S}(a_i)}(t_{j'}^-))\mathbf{1}_{\mathcal{C}_{j'}^-}+ (\varphi_{ij_+'}^{\mathcal{S}(a_i)}\circ F_{j'}^{-1}-\varphi_{ij_+'}^{\mathcal{S}(a_i)}(t_{j'}^{+}))\mathbf{1}_{\mathcal{C}_{j'}^+}\Big).
 \end{align*}
Notice that $\psi\circ F_{j'}(1)=\psi\circ F_{j'}(-1)=0$.  Then Lemma \ref{driftout} with $\lambda=\frac{w_{j'}^+}{w_{j'}}$ gives the estimate
 \begin{align*}
  \E\big[ & \big( \sum_j\int_{\mathcal{C}_{j'} \cap S_j }e^{\sum_i2(2\pi/\lambda_i)^{1/2}a_i\varphi_{ij}^{\mathcal{S}(a_i)}}d\mathcal{G}')\big)^{w_{j'}\frac{Q\chi(M)}{\gamma}}| (a_i)_{1\leq i \leq m}\big] \\
  \leq & C\Big(  e^{ \tfrac{Q\chi(M)}{ \gamma}\sum_i2(2\pi/\lambda_i)^{1/2}a_i  \big(w_{j'}^-\varphi_{ij_-'}^{\mathcal{S}(a_i)}(t_{j'}^-)+w_{j'}^+\varphi_{ij'_+}^{\mathcal{S}(a_i)}(t_{j'}^+)\big)}\Big)\ell_{j'}^{1-\delta}e^{C\sum_{i}\lambda_ia_i^2 }. 
\end{align*}
To estimate the quantity $\int_{1\leq |r|\leq \frac{2\pi}{\ell_{j'}^{1-\delta}}}|\psi'(r)|^2\,dr$ appearing in the conclusion of Lemma \ref{driftout},  we have used the chain rule formula for derivatives combined with the following  elementary estimates for $ F_{j'}^{-1}(r)=\frac{\ell_{j'}}{\tan\frac{\ell_{j'}r}{2\pi}}$ valid    for all $r\in [1,\frac{2\pi}{\ell_{j'}^{1-\delta}}]$ and for some irrelevant constant $c>0$ 
$$\frac{c^{-1}}{r}\leq F_{j'}^{-1}(r)\leq \frac{c }{r} \quad \text{ and }-\frac{c }{r^2}\leq (F_{j'}^{-1})'(r)\leq -\frac{c^{-1} }{r^2}.$$
Combining with the expressions \eqref{phij1}+\eqref{phij2} we obtain the estimate $$|\psi'(r)|^2\leq C(\sum_i\lambda_i^{1/2}a_i|r|^{\lambda_i-1})^2,$$ yielding in turn, after integrating (and recalling that $0<\lambda_i<1/4$),
 $$\int_{1\leq |r|\leq \frac{2\pi}{\ell_{j'}^{1-\delta}}}|\psi'(r)|^2\,dr\leq C\sum_{i,i'}a_ia_{i'}(\lambda_i\lambda_{i'})^{1/2}\leq C\sum_{i}\lambda_ia_i^2 $$
for some global constant $C$ (which may change along lines). Let us now make a remark. In the statement of Lemma \ref{driftout}, the exponent $q$ is deterministic whereas here it is  random as a measurable function of the $w_{j'}$ (hence of the $a_i$'s), namely $q=w_{j'}\frac{Q\chi(M)}{\gamma}$. Yet conditionally on the $a_i$'s the exponent can be seen as deterministic. We can thus apply Lemma \ref{driftout}. The resulting constant $C_q$ is therefore a measurable function of the $a_i$'s. Yet, because $C_q$ is locally bounded as a function of $q$ and because $w_{j'}\in[0,1]$ for all $j'$, we can obtain an overall  deterministic constant $C$ in the above inequality by taking $C=\sup_{q\in [0,\frac{Q\chi(M)}{\gamma}]} C_q$.

So far, we have established that
\begin{align*}
\E\big[\mathcal{G}_g^\gamma(M)^{\frac{Q\chi(M)}{\gamma}}\big]\leq C\E\Big[  e^{C\sum_{i}\lambda_ia_i^2 } e^{ \tfrac{Q\chi(M)}{ \gamma}\sum_{i,j'}2(2\pi/\lambda_i)^{1/2}a_i  \big(w_{j'}^-\varphi_{ij_-'}^{\mathcal{S}(a_i)}(t_{j'}^-)+w_{j'}^+\varphi_{ij'_+}^{\mathcal{S}(a_i)}(t_{j'}^+)\big)}   \Big]\prod_{j'}\ell_{j'}^{1-\delta}.
\end{align*}
We can convert the sums over $j'$ in the above expectation   into   sums of the type  $\sum_j\sum_{j'\in I^{\pm}_j}$ and use the relation $\varphi_{ij'_{\pm}}^{\mathcal{S}(a_i)}(t_{j'}^\pm)=v_{ij } +\mc{O}(\lambda_i)+\mc{O}(\epsilon)$ for $j'\in I_j^{\pm} $ (the $\mc{O}$ entering this relation are uniform w.r.t. the $a_i$'s as easily seen from the expressions \eqref{phij1} and \eqref{phij2}) to get for some $C>0$
\begin{align}
\E\big[&\mathcal{G}_g^\gamma(M)^{\frac{Q\chi(M)}{\gamma}}\big]\label{estgq1}\\
\leq & C\E\Big[  e^{C\sum_{i}\lambda_ia_i^2 } 
e^{C \sum_i\frac{|a_i|}{\sqrt{\lambda_i}}( \mc{O}(\epsilon)+\mc{O}(\lambda_i))}
e^{ \tfrac{Q\chi(M)}{ \gamma}\sum_{i,j}2(2\pi/\lambda_i)^{1/2}a_i   r_j v_{ij}\big)}   \Big]\prod_{j'}\ell_{j'}^{1-\delta}.\nonumber
\end{align}

\medskip To complete the proof we use the structure of the coefficients $(v_{ij})_{ij}$. Recall that the vectors $v_i\in\R^{m+1}$  with components $(v_{ij})_{j}$ form an orthonormal family for the inner product $(u,v) =\sum_j u_jv_j{\rm Vol}_{g_0}(S_j)$ and the orthogonal complement of ${\rm span} (v_i)_i$  is  the vector $\mathbf{1}$ with all components equal to $1$. Now we define the precise values of the coefficients $r_{j}$ involved in \eqref{estgq1}. Set $V=\max_{ij}|v_{ij}|$ and define
\begin{equation}\label{defrj}
r_{j}:=R\big(1+\tfrac{1}{2mV}\sum_i\mathcal{S}(a_i)v_{ij} \big){\rm Vol}_{g_0}(S_j)
\end{equation}
 where $R$ is a normalizing constant such that $\sum_j r_j=1$. Observe that $0<r_j<1$ for all $j$ and with this choice
 $$\forall i,\quad \sum_j r_jv_{ij}=\frac{R}{2mV}\mathcal{S}(a_i).$$
Now we plug this relation  into the estimate  \eqref{estgq1}  to get for some $C>0$
\begin{align*}
\E\big[\mathcal{G}_g^\gamma(M)^{\frac{Q\chi(M)}{\gamma}} \big]\leq &C\E\Big[ e^{ -C \sum_i\frac{|a_i|}{\sqrt{\lambda_i}}(1+\mc{O}(\epsilon)+\mc{O}(\lambda_i))+C\sum_{i}\lambda_ia_i^2 } \Big]   \prod_{j'=1}^{n_p}\ell_{j'}^{1-\delta}.\nonumber
\end{align*} 
 
  We can choose the neighborhood of $(M_0,g_0)  $ in such a way that the term $|\mc{O}(\epsilon)+\mc{O}(\lambda_i)|$ is less than $1/2$ uniformly with respect to $i,(a_i)_i$, in which case taking expectation of the above expression yields
$$\E\big[\mathcal{G}_g^\gamma(M)^{\frac{Q\chi(M)}{\gamma}}  \big]\leq C\E\big[  e^{ -C \sum_i\frac{|a_i|}{\sqrt{\lambda_i}} +C\sum_{i}\lambda_ia_i^2 }\big]\Big(  \prod_{j'=1}^{n_p}\ell_{j'}^{1-\delta}\Big)\leq C\prod_{i=1}^m\lambda_i^{1/2}\prod_{j'=1}^{n_p}\ell_{j'}^{1-\delta}.$$
 This completes the proof.\qed

\subsection{Proof of Theorem \ref{defLiouvillemeas} in the case $\gamma<2$}
The proof of the case $\gamma<2$ is much simpler than the case $\gamma=2$. The main reason is that the quantity $ \Big(\frac{{\det}(P_{g_\tau}^*P_{g_\tau})}{({\det }'\Delta_{g_\tau})^{\mathbf{c}_{\mathbf{M}}+1}}\Big)^{\demi} \,d\tau$ appearing in the integral \eqref{partg2} presents a more gentle behaviour at the boundary of the moduli space due to the lower central charge $\mathbf{c}_{\mathbf{M}}<1$. More precisely   \eqref{Wolpert1} and \eqref{Wolpert2} now gives the following estimate for this term  
\begin{equation}\label{product<}
 C\Big(\prod_{j=1}^{3\mathbf{g}-3}\ell_j\Big)\Big(\prod_{j'=1}^{n_p}\ell_{j'}^{-\frac{5-\mathbf{c}_{\mathbf{M}}}{2}}e^{-\frac{\pi^2(1-\mathbf{c}_{\mathbf{M}})}{6\ell_{j'}} }\prod_{\lambda_i<1/4}\lambda_i ^{-\frac{\mathbf{c}_{\mathbf{M}}+1}{2}}\Big)\prod_{j}d\ell_j d\theta_j,
 \end{equation} 
hence an additional exponential decay in comparison with the case $\mathbf{c}_{\mathbf{M}}=1$.

We stick to the notations introduced in section \ref{sec:Green} and in the proof of Theorem \ref{defLiouvillemeas} in the case $\gamma=2$. We introduce $(a_i)_{1 \leq i \leq m}$ i.i.d. standard Gaussian random variables and consider the following Gaussian field:
\begin{equation*}
Y(x)= \sum_{i=1}^m f_{v_i}(x) \frac{a_i}{\sqrt{\nu_i}}    
\end{equation*}
where the $\nu_i$'s are defined in Theorem \ref{thburger} and $f_{v_i}$ in \eqref{deffa}, with $v_i=(v_{ij})_{1\leq j\leq m+1}\in\R^{m+1}$ from Lemma \ref{approximation}. 
Now, recall  that on each $S'_j$ the field $Y(x)$ is the constant random variable $Y_j= \sum_{i=1}^m v_{ij}\frac{a_i}{\sqrt{\nu_i}}   $.

By combining Proposition \ref{greenpinching} and Lemma \ref{approximation}, the Green function $G_g$ is such that for all $x,x' \in \cup_{1 \leq j \leq m+1} S'_j $ 
\begin{equation*}
G_g(x,x')  \leq \sum_{i=1}^m \frac{f_{v_i}(x) f_{v_i}(x')}{\nu_j(g)}   +A_g(x,x')  +   \frac{C \delta^{1/L}}{\nu_1}
\end{equation*}
where $C,L>0$ are global constants. Hence, if we introduce a standard normalized Gaussian variable $Z$ and an independent Gaussian field $X'_g$ (living in the space of distributions) with covariance $A_g$, we get that
for all $x,x' \in \bigcup_{1 \leq j \leq m+1} S'_j  $ 
\begin{equation*}
G_g(x,x')  \leq \E[ Y(x) Y(x')    ]  +\E[X_g'(x) X'_g(x')  ]  + C\delta^{1/L} \nu_1^{-1} \E Z   ^2 
\end{equation*}
in such a way that Kahane's convexity inequality \cite{cf:Kah} ensures that for all $q>0$
\begin{align*}
& \E  \left [  \mc{G}_{g}^\gamma( \cup_{1 \leq j \leq m+1} S'_j )^{-q} \right ]   \\
& =    \E  \Big [    \Big  ( \sum_{j=1}^{m+1}  \int_{S'_j}  e^{\gamma X_g(x)  -\frac{\gamma^2}{2}\E[  X_g(x)^2  ] } e^{\frac{\gamma^2}{2} W_g(x)} {\rm dv}_g(x)   \Big)^{-q} \Big]    \\
&  \leq   \E \Big[  \Big( \sum_{j=1}^{m+1}  \int_{S'_j}  e^{\gamma Y(x) -\frac{\gamma^2}{2} \E[ Y(x)^2 ]  +X'_g(x)-\frac{\gamma^2}{2} \E[ X'_g(x)^2 ] + \gamma  (C\delta^{1/L}/\nu_1)^{1/2} Z   -\frac{\gamma^2}{2}  C\delta^{1/L}/\nu_1 } e^{\frac{\gamma^2}{2} W_g(x)} {\rm dv}_g(x)  \Big)^{-q}\Big]    \\
& \leq e^{(q+q^2/2) \frac{\gamma^2}{2}\frac{C\delta^{1/L}}{\nu_1}  }  \E \Big[ \Big( \sum_{j=1}^{m+1}  \int_{S'_j}  e^{\gamma Y(x) -\frac{\gamma^2}{2} \E[ Y(x)^2 ]  +\gamma X'_g(x)-\frac{\gamma^2}{2} \E[ X'_g(x)^2 ]  } e^{\frac{\gamma^2}{2} W_g(x)} {\rm dv}_g(x)   \Big)^{-q}\Big] .
\end{align*}
Now, by Proposition \ref{greenpinching} and Lemma \ref{approximation}, there is some global constant $C>0$ such that for all $j$ and $x \in S'_j$ we have
\begin{equation}\label{infBatman}
W_g(x) \geq E[Y(x)^2]  -C
\end{equation} 
so that for some other global constant $C>0$
 \begin{align*}
& \E  \left [ \mc{G}_{g}^\gamma( \cup_{1 \leq j \leq m+1} S'_j)^{-q} \right ]    \leq   C e^{ \frac{C\delta^{1/L}}{\nu_1}  }   \E  \Big[   \Big( \sum_{j=1}^{m+1}  \int_{S'_j}  e^{\gamma Y(x) + \gamma X'_g(x)-\frac{\gamma^2}{2} E[ X'_g(x)^2 ]  }  d{\rm v}_g(x)   \Big)^{-q} \Big] .
 \end{align*}
Notice that
\begin{equation*}
\sum_{j} {\rm Vol}_{g}(S_j) Y_j= \sum_{j} \text{Vol}_{g}(S_j) (\sum_i v_{ij} \frac{a_i}{\sqrt{\nu_i}} )=   \sum_{i} \frac{a_i}{\sqrt{\nu_i}}  (\sum_j {\rm Vol}_{g}(S_j) v_{ij} ) =0
\end{equation*} 
since the vectors $(v_i)_i$ are orthogonal to $1$ in the $|| .||_g$ norm \eqref{normgraph}. 
Hence there exists almost surely some $j$ such that $Y_j \geq 0$. Therefore, gathering the above considerations, we have 
\begin{align*}
 \E  \Big[  \mc{G}_{g}^\gamma( \cup_{j=1}^{m+1} S'_j )^{-q} \Big]   & \leq  C e^{ \frac{C\delta^{1/L}}{\nu_1}  } 
  \E  \Big [   \Big( \sum_{j=1}^{m+1}  \int_{S'_j}  e^{\gamma Y(x) +\gamma X'_g(x)-\frac{\gamma^2}{2} \E[X'_g(x)^2 ]  }  {\rm dv}_g(x)   \Big)^{-q} \Big]    \\
  & \leq    C e^{ \frac{C\delta^{1/L}}{\nu_1}  } \sum_j  \E  \Big[   1_{Y_j \geq 0}    \Big( \int_{S'_j}  e^{\gamma Y_j +\gamma X'_g(x)-\frac{\gamma^2}{2} \E[ X'_g(x)^2 ]  }  {\rm dv}_g(x)   \Big)^{-q} \Big]    \\
& \leq   C e^{ \frac{C\delta^{1/L}}{\nu_1}  } \sum_j \E  \Big[    \Big( \int_{S'_j}  e^{\gamma X'_g(x)-\frac{\gamma^2}{2} \E[ X'_g(x)^2 ]  }  {\rm dv}_g(x)   \Big)^{-q} \Big]    \\
& \leq   C e^{ \frac{C\delta^{1/L}}{\nu_1}  } 
 \end{align*}
where in the last inequality we have used the fact that the covariance of $X'_g$ can be controlled independently of the size of the small eigenvalues (this is a consequence of Proposition \ref{greenpinching} and Lemma \ref{approximation}). 
Combining with \eqref{product<}, this estimate shows integrability with respect to the measure \eqref{product<} by recalling that $\nu_1  \geq C \ell_1$ for some some constant $C>0$ and by choosing $\delta$ such that $C\delta^{1/L}<\frac{(1-\mathbf{c}_{\mathbf{M}}) \pi^2}{6} $.

Finally, it remains to identify the relation \eqref{defLQGmeasure}. Starting from the definition of $\nu$, we have
\begin{align*}
\E_{  (g_\tau)_{\tau} , \mu}&[   F(  \mathcal{L}_{\gamma}(dz)  , R ) ]  =\nu(F)/\nu(1)\\
=&\frac{1}{Z}\int_{\mc{M}_{\bf g}}\int_{\R}  \Big(\frac{{\det}(P_{g_\tau}^*P_{g_\tau})}{({\det }'\Delta_{g_\tau})^{\mathbf{c}_{\mathbf{M}}+1}}\Big)^{\demi} \E\Big[ F(e^{\gamma c} \mc{G}_{g_\tau}^\gamma(dz),\tau)e^{-Q\chi(M) c-\mu  e^{\gamma c}\mc{G}_{g_\tau}^\gamma(M)}   \Big]\,d\tau dc.
\end{align*}
It suffices to make the change of variables $y=e^{\gamma c}\mc{G}_{g_\tau}^\gamma(M)$ to get
\begin{align*}
&\E_{  (g_\tau)_{\tau} , \mu} [   F(  \mathcal{L}_{\gamma}(dz)  , R ) ]   \\
=&\frac{1}{\gamma Z}\int_{\mc{M}_{\bf g}}  \Big(\frac{{\det}(P_{g_\tau}^*P_{g_\tau})}{({\det }'\Delta_{g_\tau})^{\mathbf{c}_{\mathbf{M}}+1}}\Big)^{\demi}\E\Big[ F\Big(y\frac{\mc{G}_{g_\tau}^\gamma(dz)}{\mc{G}_{g_\tau}^\gamma(M)},\tau\Big) \mc{G}_{g_\tau}^\gamma(M)^{-\frac{2Q(\mathbf{g}-1)}{\gamma}} \Big] y^{\frac{2Q(\mathbf{g}-1)}{\gamma}-1} e^{ -\mu y}\,dy,
\end{align*}
from which our claim follows.\qed

 \subsection{Proof of Lemma  \ref{driftout}.} 
  We will use the results in Lemma \ref{boundrobin} and Proposition \ref{lastestimateGg}. So it is convenient to introduce the notations
\begin{align*}
 &\mathcal{C}_{j'}^{\pm}:= \mathcal{C}_{j'}(\ell)^+\cup \mathcal{C}_{j'}(\ell')^-, &  F_{j'}(\rho):=\tfrac{2\pi}{\ell_{j'}}\arctan(\tfrac{\ell_{j'}}{\rho}) , \\
 &  h(\rho):=\sum_i|\rho|^{-2\lambda_i}\big(1+   \ln(1/|\rho|)\big).
\end{align*}
We further denote by   $F_{j'}^{-1}$ the inverse of the function $F_{j'}$.  Finally, in what follows,  $C$ stands for a generic irrelevant positive  constant, the value of which  may change along lines.

Proposition \ref{lastestimateGg} ensures that the Green function $A_g$ (which is the covariance of the field $X'_g$  defined in \eqref{defX'}) satisfies for $x=\rho e^{i\theta}$ and $x'=\rho' e^{i\theta}$ in the collar $\mathcal{C}_{j'}$
\begin{equation}\label{errorG}
A_g(x,x')\leq G_{g_{j'}}(x,x') +C^2h(\rho)h(\rho').
\end{equation}
Otherwise stated, we have bounded the errors terms in Proposition \ref{lastestimateGg} with the help of the function $h$. We stress that the function $\chi$ in Proposition \ref{lastestimateGg} is worth $1$ only for $|\rho|\leq 1/4$ so that the above bound is rigorously valid only for $|\rho|\leq 1/4$. Yet in the following we assume  that it is valid for $|\rho|\leq 1$ for notational convenience because it does not change the validity of the argument.

Then we need to decompose the Gaussian field with covariance function $G_{g_{j'}}$ according to its radial/angular parts. So  we consider two independent centered Gaussian fields $X^{r}$ and $X^a$ defined on the collar $\mathcal{C}_{j'}=[-1,1]_\rho\times (\R\setminus\Z)_\theta$  with respective covariance kernels $G^{r}$ and $G^a$ defined by
 \begin{align}
G^{r}(\rho,\theta,\rho',\theta') :=& \left\{\begin{array}{ll}   \min(F_{j'}(|\rho|),F_{j'}(|\rho'|))-\frac{\ell_{j'}}{2\pi^2}F_{j'}(|\rho|)F_{j'}(|\rho'|)+C^2h(\rho)h(\rho') &\text{if }\rho\rho'>0\\ 
\frac{\ell_{j'}}{2\pi^2}F_{j'}(|\rho|)F_j(|\rho'|)  +C^2h(\rho)h(\rho')   &\text{if }\rho\rho'\leq 0
\end{array}\right.\label{gr}
\\
G^a(\rho,\theta,\rho',\theta') :=&-\ln\Big|1-e^{- |F_{j'}(\rho)-F_{j'}(\rho')|+2i\pi (\theta-\theta')}\Big|\mathbf{1}_{\{\rho\rho'\geq 0\}}.
 \end{align}
The  two quantities $G^{r}$ and $G^a$ are   positive definite, hence the existence  of such fields (pointwise for $X^r$ and distributional for $X^a$).
Combining the above two relations we get that for $x=\rho e^{i\theta}$ and $x'=\rho' e^{i\theta}$ in the collar $\mathcal{C}_{j'}$
\begin{equation}
\E[X'_g(x)X'_g(x')]\leq \E[X^r(x)X^r(x')]+\E[X^a(x)X^a(x')]+\E[(  \delta^{(0)}h(\rho))(\delta^{(0)}h(\rho'))]
\end{equation}
where $\delta^{(0)}$ is a standard Gaussian random variable independent of everything. Hence we can use   Kahane's inequality \cite[Lemma 1]{cf:Kah} to get the estimate (recall that $m_{g_{j'}}$ is the Robin constant defined in Lemma \ref{boundrobin}) for $q>0$
\begin{align}\label{estinterm0}
 \E\Big[\Big(\int \phi & e^{\psi\circ F_{j'}}\,d\mathcal{G}'\Big)^{-q}\big]\\
 \leq &C_q\E\big[\big(\int_{\mathcal{C}_{j'}^\pm}\phi \,e^{\psi\circ F_{j'}}e^{2C\delta^{(0)}h-2h^2}e^{ 2X^{r}-2\E[(X^r)^2]}e^{4\pi m_{g_{j'}}}\,d \mathcal{G}_g^a\big)^{-q}\big]\nonumber
 \end{align}
 where $C_q$ is an explicit constant such that $\ln C_q$ is quadratic in $q$.
Here   we have defined the random measure  $\mathcal{G}_g^a$ as the limit in law as $\epsilon \to 0$ (eventually up to some subsequence) of the family of random measures $(-\ln \eps)^{1/2}e^{2X^a_\eps-2\E[(X^a_\eps)^2]}\,d{\rm v}_g$. 
From \eqref{boundrobin}+Proposition  \ref{lastestimateGg}, we get that    $4\pi m_{g_{j'}} -2\E[X_r^2]\geq 2\ln|\rho|-2C^2h(\rho)^2$   for $|\rho|>\ell_{j'}$  for some constant $C>0$.  Hence, for $|\rho|>\ell_{j'}$, the measure $e^{ 2C\delta^{(0)}h-2h^2+2X^{r}-2\E[(X^r)^2]}e^{4\pi m_{g_{j'}}}\,d \mathcal{G}_g^a$ is greater than the measure $e^{ 2X^{r} }e^{g(\rho)}\rho^{-2}\,d \mathcal{G}_g^a$ 
 where we have set 
 \begin{equation}\label{grho}
g(\rho):= 2C\delta^{(0)}h(\rho)+4\ln |\rho|-4C^2h(\rho)^2.
\end{equation}

Now that we have simplified the deterministic part of the measure we analyze the random part. For this, we write a path decomposition result for the process $X^r$
\begin{lemma}\label{lemXr} 
 Let us consider two standard Gaussian r.v. $\delta^{(1)},\delta^{(2)} \sim\mathcal{N}(0,1)$ and  two standard Brownian bridges $({\rm Br}^+_\rho)_{\rho\in[0,1]}$ $({\rm Br}^-_\rho)_{\rho\in[0,1]}$ , all of them mutually independent. We have the following equality in law in the sense of processes for $\rho\in[-1,1]$
 $$X^r(\rho) = \mathbf{1}_{\{\rho>0\}} \pi/\sqrt{\ell_{j'}}{\rm Br}^+_{\tfrac{\ell_{j'}}{\pi^2}F_{j'}(|\rho|)} + \mathbf{1}_{\{\rho<0\}} \pi/\sqrt{\ell_{j'}}{\rm Br}^-_{\tfrac{\ell_{j'}}{\pi^2}F_{j'}(|\rho|)} +\tfrac{\sqrt{\ell_{j'}/2}}{\pi}F_{j'}(|\rho|)\,\delta^{(1)}+Ch(\rho)\delta^{(2)} .$$
 \end{lemma}
 
 \noindent {\it Proof.} Recall the covariance structure of the Brownian bridge $\E[{\rm Br}^+_\rho{\rm Br}^+_{\rho'}]=\rho\wedge\rho'-\rho\rho' $. Hence for arbitrary constants $a,c>0$ and $\rho<1/c$ 
 $$\E[(a{\rm Br}^+_{c\rho})(a{\rm Br}^+_{c\rho'})]=a^2c\rho\wedge\rho'-a^2c^2\rho\rho' .$$
 Adjusting the constants $a,c$ to fit with the covariance function $ \min( \rho,\rho')-\frac{\ell_{j'}}{\pi^2}\rho\rho'$ gives $a=\pi/\sqrt{\ell_{j'}}$ and $c=\tfrac{\ell_{j'}}{\pi^2}$. One completes easily the proof of the claim by  time changing with $F_{j'} $ and adding the covariance structure of the term $\tfrac{\sqrt{\ell_{j'}/2}}{\pi}F_{j'}(|\rho|)\,\delta^{(1)}+Ch(\rho)\delta^{(2)} $.\qed

Now, observe that if we restrict to those $\ell_{j'}\leq  |\rho|\leq 1$ then $\tfrac{\ell_j}{\pi^2}F_{j'}(|\rho|)\in  [0,\tfrac{1}{2}]$ and the law of the Brownian bridge ${\rm Br}$ on $[0,\tfrac{1}{2}]$ is absolutely continuous with respect to the law of Brownian motion $B$ with density $2e^{-|B_{1/2}|^2}$. Using this relation with ${\rm Br}^+$ and  ${\rm Br}^-$ together with the scaling relation for Brownian motion $aB_{t/a^2}\stackrel{law}{=}B_t$ for fixed $a>0$, we deduce  using \eqref{estinterm0} and Lemma \ref{lemXr}  that
 \begin{align}\label{estinterm2}
 \E\Big[&\Big(\int \phi e^{\psi\circ F_{j'}} \,d\mathcal{G}'\Big)^{-q}\big]\\
 \leq &C_q\E\Big[\Big(A\int_{\mathcal{C}_{j'} (  \ell)^+ }e^{\psi\circ F_{j'}}e^{2B^+_{F_{j'}(\rho)} +\Theta(\rho) }\rho^{-2}\,d \mathcal{G}_g^a +B\int_{\mathcal{C}_{j'}(  \ell')^- }e^{\psi\circ F_{j'}}e^{2B^-_{F_{j'}(\rho)}+\Theta(\rho)}\rho^{-2}\,d \mathcal{G}_g^a\Big)^{-q}\Big],\nonumber
 \end{align}
 where $B^+,B^-$ are two independent Brownian motions, independent of everything, and the function $\Theta$ is defined by $\Theta(\rho):=\sqrt{2\ell_{j'}}/\pi F_{j'}(\rho)\,\delta^{(1)}+2Ch(|\rho|)\delta^{(2)}+g(\rho)$. Now we would like to get rid of the drift term $\Theta$. The point is that the behaviour of $\Theta$ is rather tricky for those    $|\rho|$ that are very close to (or less than) $\ell_{j'}$ whereas the contribution of the $\Theta(\rho)$, say for  $|\rho|\geq \ell_{j'}^{1-\delta}$ for any $\delta>0$,  turns out to be easily controlled.   So we fix an arbitrary $\delta\in ]0,1[$ and remark that the expectation in the r.h.s. of \eqref{estinterm2} is less than the same expectation with  integration restricted to $\mathcal{C}_{j'}(  \ell^{' 1-\delta})^-$ and $\mathcal{C}_{j'} (  \ell^{1-\delta})^+$. Furthermore, we introduce the (random) function $Y_\rho$ through the relation
\begin{equation}\label{defY}
\forall \rho\in]0,1],\quad  \Theta(\rho)= 2\int_{F_{j'}(1)}^{F_{j'}(\rho)}  Y_u\,du+\Theta(1).
\end{equation}
We set     $ \kappa(\ell):=F_{j'}(  \ell^{1-\delta})=\tfrac{2\pi}{\ell_{j'}}\arctan(\ell_{j'}\ell^{\delta-1})$. Then the Girsanov theorem   tells us that, under the probability measure 
$$R\,d\P,\quad \text{ with } R:= e^{\int_{ \kappa(1)}^{ \kappa(\ell)}Y_r\,dB^+_r+\int_{ \kappa(1)}^{ \kappa(\ell')}Y_r\,dB^-_r- \tfrac{1}{2}\int_{ \kappa(1)}^{ \kappa(\ell)}Y_r^2\,dr-\tfrac{1}{2}\int_{ \kappa(1)}^{ \kappa(\ell')}Y_r^2\,dr},$$ 
 the processes $\rho\in [1,\ell^{1-\delta}]\mapsto 2 B^+_{F_{j'}(\rho)} $ and $\rho\in [1,\ell^{'1-\delta}]\mapsto 2 B^-_{F_{j'}(\rho)} $ have respectively the same laws as   the processes   $\rho\in [1,\ell^{1-\delta}]\mapsto 2 B^+_{F_j(\rho)} +\Theta(\rho)-\Theta(1)$ and  $\rho\in [1,\ell^{1-\delta}]\mapsto 2 B^-_{F_j(\rho)} +\theta(\rho)-\Theta(1)$ under $\P$. Therefore, using the Girsanov transform in the expectation \eqref{estinterm2}, we get  
  \begin{align*}
 \E\Big[&\Big(\int \phi e^{\psi\circ F_{j'}} \,d\mathcal{G}'\Big)^{-q}\big]\\
 \leq &4\E\Big[R e^{-q\Theta(1)}\Big(A\int_{\mathcal{C}_{j'}(   \ell^{1-\delta})^+ }e^{\psi\circ F_{j'}}e^{2B^+_{F_{j'}(\rho)}  }\rho^{-2}\,d \mathcal{G}_g^a 
 +B\int_{\mathcal{C}_j(    \ell_{j'}^{'1-\delta})^- }e^{\psi\circ F_{j'}}e^{2B^-_{F_{j'}(\rho)}}\rho^{-2} \,d \mathcal{G}_g^a\Big)^{-q}\Big]\nonumber\\
 \leq &4E_{\ell,\ell'}(q)\E\Big[ \Big(A\int_{\mathcal{C}_{j'}(   \ell^{1-\delta})^+ }e^{\psi\circ F_{j'}}e^{2B^+_{F_{j'}(\rho)}  }\rho^{-2}\,d \mathcal{G}_g^a 
 +B\int_{\mathcal{C}_{j'}(    \ell^{'1-\delta})^- }e^{\psi\circ F_{j'}}e^{2B^-_{F_{j'}(\rho)}}\rho^{-2} \,d \mathcal{G}_g^a\Big)^{-pq}\Big]^{\frac{1}{p}}.
 \end{align*}
 where we have  used the H\"older inequality to get the last inequality with $p,m,r>1$ and $\tfrac{1}{p}+\tfrac{1}{m}+\tfrac{1}{r}=1$ and set 
\begin{align*}
E_{\ell,\ell'}(q):=&\E\Big[\Big(e^{\int_{ \kappa(1)}^{ \kappa(\ell)}Y_u\,dB^+_u+\int_{ \kappa(1)}^{ \kappa(\ell')}Y_u\,dB^-_u- \tfrac{1}{2}\int_{ \kappa(1)}^{ \kappa(\ell)}Y_u^2\,du-\tfrac{1}{2}\int_{ \kappa(1)}^{ \kappa(\ell')}Y_u^2\,du}\Big)^m \Big]^{1/m}  \E[e^{-qr\Theta(1)}]^{1/r} \\
=& \E\Big[ e^{  \tfrac{m^2-m}{2}\int_{ \kappa(1)}^{ \kappa(\ell)}Y_r^2\,dr+\tfrac{m^2-m}{2}\int_{ \kappa(1)}^{ \kappa(\ell')}Y_r^2\,dr} \Big]^{1/q} \E[e^{-qr\Theta(1)}]^{1/r}.
\end{align*}
So we have got rid of the drift term $\Theta(\rho)$ with the Girsanov trick. The cost is the constant $E_{\ell,\ell'}(q)$ but an easy computation shows that  $\sup_{ \ell,\ell'\geq \ell_j}E_{\ell,\ell'}(q)<+\infty$ for any $q>1$. This computation is left to the reader but we   give a brief convincing heuristic  argument.
The process $Y_u$  is defined by \eqref{defY}. We have already explained  in the proof of Theorem \ref{defLiouvillemeas} that for small $\ell_{j'}$ and for $|\rho|\geq\ell_{j'}$, $F_{j'}(\rho)$ behaves like $1/\rho$. Hence  $Y(\rho)$ is with good approximation   given by $-\Theta'(1/\rho)\rho^{-2}$. Then it is readily seen that  $-\Theta'(1/\rho)\rho^{-2}$ is a sum of terms of the type $1/\rho$, $|\rho|^{-c\lambda_i-1}(\ln |\rho|)^n$ (for $n=0,1,2$) or $\ell_{j'}^{1/2}$. It is then obvious to see that the square of every possible linear combination of such terms  has its $\int_{\kappa(1)}^{ \kappa(\ell)}$-intergal bounded by constant independently of $\ell\geq \ell_{j'}$.

We can use the same argument to explicitly determine the effect of the drift term $\psi\circ F_{j'}$.  The variance of the Girsanov transform to get rid of this term is less than
$$C\int_{1\leq |r|\leq 2\pi /\min(\ell^{ 1-\delta},\ell^{' 1-\delta})}|\psi'(r)|^2dr$$  (here we have used the fact that $\kappa(1)\geq 1$ and $\kappa(\ell)\leq 2\pi/\ell^{1-\delta}$).  
All in all, this entails that for arbitrary $p>1$ there exists some constant $C_p$ such that 
 \begin{align}
  \E\Big[&\Big( \int \phi e^{\psi\circ F_{j'}} \,d\mathcal{G}'\Big)^{-q}\big]\leq C_p \exp\Big(\int_{0}^{2\pi \max(\ell,\ell')^{\delta-1}}|\psi'(r)|^2dr\Big)\\
  &\times\E\Big[\Big(A\int_{\mathcal{C}_j(  \ell^{1-\delta})^+ } e^{2B^+_{F_j(\rho)}  }\rho^{-2}\,d \mathcal{G}_g^a
 +B\int_{\mathcal{C}_j (  \ell^{'1-\delta})^- } e^{2B^-_{F_j(\rho)}}\rho^{-2} \,d \mathcal{G}_g^a\Big)^{-pq}\Big]^{1/p},\nonumber
 \end{align}
which in turn less than
 \begin{align}\label{estinterm5}
   C_p(AB)^{-\frac{q}{2}}\exp\Big(&\int_{0}^{2\pi \max(\ell,\ell')^{\delta-1}}|\psi'(r)|^2dr\Big)\\
   &\times \E\Big[\Big(\int_{\mathcal{C}_{j'}(  \ell^{1-\delta})^+ } e^{2B^+_{F_{j'}(\rho)}  }\rho^{-2}\,d \mathcal{G}_g^a\Big)^{-pq\lambda} \Big(\int_{\mathcal{C}_{j'}(  \ell^{'1-\delta})^- } e^{2B^-_{F_{j'}(\rho)}}\rho^{-2} \,d \mathcal{G}_g^a\Big)^{-pq(1-\lambda)}\Big]^{\frac{1}{p}}\nonumber
 \end{align}
after using the elementary inequality $(a+b)^{-m}\leq a^{-\lambda m}b^{-(1-\lambda)m}$ for $a,b\geq 0$, $\lambda\in[0,1]$ and $m>0$.

  It remains to evaluate the latter expectation. So we introduce the  sets for $n,k\geq 0$ and $\ell>0$
 \begin{align*}
 A_n^+(\ell)=&\{\sup_{u\in [\ell^{1-\delta} ,1]}  B^+_{F_{j'}(u)} -B^+_{F_{j'}(1)} \in ]n,n+1]\}\\
  A_k^-(\ell')=&\{\sup_{u\in [\ell^{'1-\delta} ,1]}  B^-_{F_{j'}(u)}-B^-_{F_{j'}(1)}  \in ]k,k+1]\}
 \end{align*}
as well as the stopping times
\begin{align*}
 T_n^+ =&\{\inf_{u\in  ]0,1]} B^+_{F_{j'}(u) }-B^+_{F_{j'}(1) }=n\}&T_n^-=&\{\inf_{u\geq 0}  B^-_{F_{j'}(u) }-B^-_{F_{j'}(1) } \in ]n,n+1]\}.
  \end{align*}
 Partitioning the probability space according to the events $ A_n^+(\ell)$ and $ A_k^-(\ell')$ and using sub-additivity of the mapping $x\in\R_+\mapsto x^{1/p}$ for $p>1$, we get the estimate
 \begin{equation*}
 \E\Big[\Big(\int_{\mathcal{C}_{j'}(  \ell^{1-\delta})^+ } e^{2B^+_{F_{j'}(\rho)}  }\rho^{-2}\,d \mathcal{G}_g^a\Big)^{-pq\lambda} \Big(\int_{\mathcal{C}_{j'}(  \ell^{'1-\delta})^- } e^{2B^-_{F_{j'}(\rho)}}\rho^{-2} \,d \mathcal{G}_g^a\Big)^{-pq(1-\lambda)} \Big]^{1/p}\leq \sum_{k,n} E_{n,k}
\end{equation*}
where we have set
$$E_{n,k}:=\E\Big[\mathbf{1}_{A_n^+(\ell)}\mathbf{1}_{A_k^-(\ell')}\Big(\int_{\mathcal{C}_{j'}(  \ell^{1-\delta})^+ }e^{2B^+_{F_{j'}(\rho)}  }\rho^{-2}\,d \mathcal{G}_g^a\Big)^{-pq\lambda}\times\Big(\int_{\mathcal{C}_{j'}(  \ell^{'1-\delta})^- }e^{2B^-_{F_{j'}(\rho)}}\rho^{-2} \,d \mathcal{G}_g^a\Big)^{{-pq (1-\lambda)} }\Big]^{\frac{1}{p}}.$$

 The idea is now the following: the fluctuations of the Brownian motion $B^+$ over an interval of length $1$ are of order $1$. Hence over the interval    $\mathcal{I}_n^+:=[F_{j'}(T_n^+),F_{j'}(T_n)+1]$, $B^+$ is approximately equal to $n$. Put in other words, the process $B^+_{F_{j'}(u)}$ is worth $n$ on the interval $[F_{j'}^{-1}(F_{j'}(T_n^+)+1),T_n^+]$. Same remark for $B^-$. Hence $E_{n,k}$ should be estimated by
\begin{equation}\begin{split}
\label{finalest}
E_{n,k}
\leq & e^{ - (n+k)pq}\P(A_n^+)^{\frac{1}{p}} \P(A_k^-)^{\frac{1}{p}} \\
& \times  \sup_{x,x'\in]\ell_j,1]}\E\Big[\Big(\int_{\{\rho\in I^+(x)\}} \rho^{-2}\,d \mathcal{G}_g^a\Big)^{-pq\lambda}\times\Big(\int_{\{\rho\in I^-(x')\}}  \rho^{-2}\,d \mathcal{G}_g^a\Big)^{{-pq(1-\lambda)} }\Big]^{\frac{1}{p}} .
\end{split}\end{equation}
where for $x\in ]0,1]$, we denote $I^+(x):=[F_j^{-1}(F_j(x)+1),x]$ and for $x'\in [-1,0[$, $I^-(x'):=[-x',-F_j^{-1}(F_j(x')+1)]$.
We will conclude with the two following lemmas
 \begin{lemma}\label{arg1}
For any $q>0$, we have
 $$\sup_{\ell_j'\leq 1}\sup_{x\in]\ell_{j'},1]}\E\Big[\Big(\int_{\{\rho\in I^+(x)\}} \rho^{-2}\,d \mathcal{G}_g^a\Big)^{-q} \Big]<+\infty.$$
 The same property for $I^-(x')$ and $x'\in [-1,-\ell_j[$.
 \end{lemma}

  \begin{lemma}\label{arg2}
There is some constant $C>0$ such that  for any $\delta>0$, $n,k\geq 0$ and $\ell,\ell'\geq\ell_{j'}$
 $$\P(A_n^+(\ell))\leq C n\ell^{\tfrac{1}{2}(1-\delta)}\hspace{2cm}\P(A_k^-(\ell'))\leq C k\ell^{'\tfrac{1}{2}(1-\delta)}.$$
 \end{lemma}  
Indeed, using H\"older inequality and Lemma \ref{arg1} we see that the expectation involved in the r.h.s. of \eqref{finalest} is less than some constant independent of everything. Hence  we get the bound $E_{n,k}\leq Ce^{ -(n+k)qp}\P(A_n^+)\P(A_k^-)$. Lemma \ref{arg2} and summability of the series $\sum_{n,k\geq 0}kn e^{ -(n+k)qp }$ complete the argument.
  
 The only remaining detail to fix is to show \eqref{finalest}. This is an easy task as, using the independence of $B^+,B^-, \mathcal{G}_g^a$ as well as the strong Markov property of the  Brownian motion, we have
 \[\begin{split}
E_{n,k}\leq & \E\Big[\mathbf{1}_{A_n^+(\ell)\cap A_k^-(\ell')}e^{-2\lambda pq B^+_{F_{j'}(1)}-2(1-\lambda) pqB^-_{F_{j'}(1)}   }\\
 & \quad \quad   \quad \times  \Big(\int_{\mathcal{I}_n^+ }e^{2B^+_{F_{j'}(\rho)}-2B^+_{F_{j'}(1)}  }\rho^{-2}\,d \mathcal{G}_g^a\Big)^{-pq\lambda}\times\Big(\int_{\mathcal{I}_n^- }e^{2B^-_{F_{j'}(\rho)}-2B^-_{F_{j'}(1)}}\rho^{-2} \,d \mathcal{G}_g^a\Big)^{-pq(1-\lambda)}\Big]^{\frac{1}{p}} \\
\leq & \E[e^{-2\lambda pq B^+_{F_{j'}(1)}-2(1-\lambda) pqB^-_{F_{j'}(1)}   }]^{\frac{1}{p}}e^{ -pq(n+k)}(\P(A_n^+(\ell))\P(A_k^-(\ell')))^{\frac{1}{p}} \\
& \E[e^{-pq\lambda\min_{u\in\mathcal{I}_n^s}B^+_{F_{j'}(u)}-B^+_{F_j{j'}(1)}}]^{\frac{1}{p}} \E[e^{-pq(1-\lambda) \min_{u\in\mathcal{I}_n^s}B^-_{F_{j'}(u)}-B^-_{F_j{j'}(1)}}]^{\frac{1}{p}}
\\ 
&  \sup_{x,x'\in]\ell_{j'},1]}\E\Big[\Big(\int_{\{\rho\in I^+(x)\}} \rho^{-2}\,d \mathcal{G}_g^a\Big)^{-pq\lambda}\times\Big(\int_{\{\rho\in I^-(x')\}}  \rho^{-2}\,d \mathcal{G}_g^a\Big)^{-pq(1-\lambda)}\Big]^{\frac{1}{p}}.
\end{split}\]
 Standard estimates about the supremum of the Brownian motion over an interval of size 1 show that the $\E[e^{-pq\lambda\min_{u\in\mathcal{I}_n^+}B^+_{F_j(u)}-B^+_{F_j(1)}}]^{1/p}$ and $ \E[e^{-pq(1-\lambda)\min_{u\in\mathcal{I}_n^-}B^-_{F_j(u)}-B^-_{F_j(1)}}]^{1/p}$ are bounded by some constant independent of $\ell,\ell'$. Hence our claim.\qed
  
 \bigskip   
  
  \noindent {\it Proof of Lemma \ref{arg1}.}  Recall that Gaussian multiplicative chaos at criticality (i.e. $\gamma=2$) possesses moments of negative order  (see \cite[Prop. 5]{Rnew7}). This entails that for any $x>0$, $\E\Big[\Big(\int_{\{x\leq \rho\leq 1\}} \rho^{-2}\,d \mathcal{G}_g^a\Big)^{-q} \Big]<+\infty$. Then we observe that we have the relation   $F_{j'}^{-1}(F_{j'}(x)+1)>C\ell_{j'}$ for some irrelevant constant $C$ and for all $x>\ell_{j'}$. Therefore, for some $C$ and all $x>\ell_{j'}$
  $$\E\Big[\Big(\int_{I^+(x)} \rho^{-2}\,d \mathcal{G}_g^a\Big)^{-q} \Big]\leq C\E\Big[\Big(\int_{I^+(x)} (\rho^2+\ell_{j'}^2)^{-1}\,d \mathcal{G}_g^a\Big)^{-q} \Big].$$
 To conclude, observe that  the measure  $(\rho^2+\ell_{j'}^2)^{-1}\,d \mathcal{G}_g^a$ is the pushforward  of the measure $e^{2X^a(F_{j'}^{-1}(\rho))-2\E[X^a(F_{j'}^{-1}(\rho))^2]}\,d\rho$ (implicitly understood as the limit of a regularized sequence) under the mapping $\rho\mapsto F_{j'} (\rho)$. The Gaussian random distribution $\rho\mapsto X^a(F_{j'}^{-1}(\rho))$ is stationary and its law does not depend on $\ell_{j'}$ in such a way that the process $x\in\R_+\mapsto \int_x^{x+1}e^{2X^a(F_{j'}^{-1}(\rho))-2\E[X^a(F_{j'}^{-1}(\rho))^2]}\,d\rho$ is stationary (and its law does not depend on $\ell_{j'}$ either) and has the same law as $\int_{I(x)} (\rho^2+\ell_{j'}^2)^{-1}\,d \mathcal{G}_g^a$, hence our claim.\qed
 
 \bigskip
  \noindent {\it Proof of Lemma \ref{arg2}.} Recall the standard computation related to Brownian motion
  $$\P(\sup_{r\in [0,t]}B_r\leq \beta)\leq \beta t^{-1/2}.$$ Therefore 
  $$\P(A_n^+(\ell))\leq \P\big(\sup_{u\in [\ell^{1-\delta} ,1]}  B^+_{F_{j'}(u)-F_j(1)} \leq n+1\big)\leq (n+1)(F_{j'}(\ell^{1-\delta})-F_{j'}(1))^{-1/2}.$$
 We conclude by noticing that $\ell_{j'}^{-1}\arctan (\ell_{j'}/x)\geq C /x$  for some constant $C$ and all $x\geq \ell_{j'}$. Same argument for $A_k^-(\ell')$.\qed

\subsection{Relation with random planar maps  }\label{randplanar}

The purpose of this subsection is to write a precise mathematical conjecture relating LQG to the scaling limit of large planar maps. 
Following Polyakov's work \cite{Pol}, it was soon acknowledged by physicists that LQG should describe the scaling limit of discretized 2d quantum gravity given by finite triangulations of a given surface,  eventually coupled with a model of statistical physics (often called matter field in the physics language), see for example the classical textbook from physics \cite{Amb} for a review on this problem. We will describe two situations in what follows: pure gravity (no matter) or the bosonic string embedded in $D=1$ dimension.

 We consider a fixed family $(g_\tau)_\tau$ of hyperbolic metrics on a compact surface $M$ (without boundary) with genus \textbf{g} as previously and the associated Liouville measure $\mc{L}_\gamma$ under  $\E_{  (g_\tau)_{\tau} , \mu}[  \cdot ]$.  Let $\mathcal{T}_{N, {\bf g}}$ be the set of triangulations  with $N$ faces with the topology of a surface of genus {\bf g}. Since these triangulations are seen up to orientation preserving homeorphisms, there are only a finite number of such triangulations.  We equip $T \in \mathcal{T}_{N,  {\bf g}}$ with a standard metric structure $h_T$ where each triangle is given volume $a^2$. The metric structure consists in gluing flat equilateral triangles: the exact definition of the metric structure is given in Les Houches lecture notes \cite{Houches} in the case of the sphere and the case we consider here does not present additional difficulties for the definition. The uniformization theorem tells us that there exists a unique $\tau_T \in \mathcal{M}_{\bf g}$ along with an orientation preserving diffeomorphism $\psi_T: T \to M$ and a conformal factor $\varphi_T$ (with logarithmic singularities at the images of the vertices of the triangles) such that 
\begin{equation}\label{decompmetric}
h_T= \psi_T^*(e^{\varphi_T} g_{\tau_T}  ).
\end{equation}
 Recall that in the decomposition \eqref{decompmetric}, the functions $\varphi_T$ and $\psi_T$ are unique except if the metric $g_{\tau_T}$ possesses non trivial isometries. In that case, the isometry group is finite of the form $(\psi^{(i)})_{1 \leq i \leq n}$ and starting with a decomposition \eqref{decompmetric} all the other decompositions of $h_T$ are $ ( (\psi^{(i)})^{-1} \circ \psi_T)^{*}(e^{\varphi_T \circ \psi^{(i)} } g_{\tau_T}  )$. Therefore, in the following discussion, we will suppose that the functions $\varphi_T$ and $\psi_T$ are uniquely determined by the triangulation $T$ and if this is not the case (i.e. there exists a non trivial isometry group), we replace $e^{\varphi_T} g_{\tau_T}$ in what follows by the average $\frac{1}{n}\sum_{i=1}^n   e^{\varphi_T \circ \psi^{(i)} } g_{\tau_T} $: these special metrics should play no role anyway as their equivalence classes are of measure $0$ with respect to the Weil-Petersson volume form.   
 
 \bigskip
 {\bf Pure gravity.}
 It is proved in \cite{bender} that the following asymptotic holds:

\begin{equation}\label{Asymptriang}
| \mathcal{T}_{N, {\bf g}}   |  \underset{N \to \infty}{\sim}  C_{\mathcal{T}}  e^{ \mu_c N} N^{\frac{5}{2}   ({\bf g}-1)-1}
\end{equation}
where $C_{\mathcal{T}}>0$ and $\mu_c>0$ are constants. The constants $C_{\mathcal{T}}, \mu_c$ are non universal in the sense that one can consider quadrangulations say in the place of triangulations: in this setting, the number of quadrangulations $\mathcal{Q}_{N, {\bf g}}$ of size $N$ will satisfy the asymptotic $| \mathcal{Q}_{N, {\bf g}}   |  \underset{N \to \infty}{\sim} C_{\mathcal{Q}} e^{\tilde{\mu}_c N} N^{\frac{5}{2}   ({\bf g}-1)-1}$ where $C_{\mathcal{Q}}$ is different from $C_{\mathcal{T}}$ and $\tilde{\mu}_c>0$ is different from $\mu_c$.

 We set 
\begin{equation}\label{Asympbarmu}
\bar{\mu}= \mu_c+ a^2 \mu,
\end{equation} 
where $\mu>0$ is fixed, and we consider the following random volume form on the surface $M$, defined in terms of its functional expectation

\begin{equation}\label{deffinitemap}
\E^{a}[  F( \nu_{a} )  ]= \frac{1}{Z_{a}} \sum_{N \geq 1}  e^{-\bar{\mu}  N} \sum_{T \in \mathcal{T}_{N,{\bf g}}} F(e^{\varphi_T} \,d{\rm v}_{g_{\tau_T}} ),
\end{equation} 
 for positive bounded functions $F$ where $Z_a$ is a normalization constant ensuring that $\E^{a} [ \cdot]$ is the expectation of a probability measure. We denote by $\P^{a}$ the probability law associated to $\E^{a}$. 

\medskip
We can now state a precise mathematical conjecture:

\begin{conjecture}\label{conjcartes}
Under $\P^{a}$, the random measure  $\nu_{a}$ converges in law as $a \to 0$ with $\bar{\mu}$ given by \eqref{Asympbarmu} in the space of Radon measures equipped with the topology of weak convergence towards  the Liouville measure $\mc{L}_\gamma$ under  $\E_{  (g_\tau)_{\tau} , \mu}[  \cdot ]$  with parameter $\gamma=\sqrt{\frac{8}{3}}$.
\end{conjecture}

The fact that $\gamma= \sqrt{\frac{8}{3}}$ can be read of the total volume of space; indeed, thanks to \eqref{Asymptriang}, it is easy to show that in the above asymptotic the total volume $\nu_a(M)$ converges to the Gamma law with density $\frac{\mu^{\frac{5}{2} ({\bf g}-1) }}{\Gamma(\frac{5}{2} ({\bf g}-1))}e^{-\mu x} x^{\frac{5}{2} ({\bf g}-1)  -1}\mathbf{1}_{x\geq 0}$. This law matches the law of the total volume $\xi_\gamma$ of $ \mc{L}_\gamma$ in Theorem \ref{defLiouvillemeas} for $\frac{2 Q}{\gamma}= \frac{5}{2}$, i.e. $\gamma= \sqrt{\frac{8}{3}}$.

Finally, let us mention that conjectures similar to \ref{conjcartes} have appeared in other topologies: the sphere \cite{DKRV}, the disk \cite{HRV} and the torus \cite{DRV}. However, in these other topologies, the corresponding conjectures are still completely open. Let us nevertheless mention some partial progress by   Curien in \cite{Curien} where appealing convergence results are proven assuming a reasonable  condition that has unfortunately not been proven yet. 

\bigskip
 {\bf Bosonic string.} Given a triangulation $T$, let us denote by $V_T$ the vertex set of the dual lattice. We consider the partition function of the bosonic string on $T$  by
 $$Z(T):=\int e^{-\tfrac{1}{2}\sum_{v\sim v'}	(x_{v}-x_{v'})^2}\prod_{v \in V}dx_{v} $$
 where $ \sim$ denotes adjacent vertices of the dual lattice. It is expected that
 
\begin{equation}\label{Asympbos}
\sum_{T\in  \mathcal{T}_{N, {\bf g}}  } Z(T)  \underset{N \to \infty}{\sim}  C'_{\mathcal{T}}  e^{ \mu_c' N} N^{2 ({\bf g}-1)-1}
\end{equation}
where $C'_{\mathcal{T}}>0$ and $\mu'_c>0$ are (non universal) constants.    We set 
\begin{equation}\label{Asympbarbos}
\bar{\mu}= \mu'_c+ a^2 \mu,
\end{equation} 
where $\mu>0$ is fixed, and we consider the following random volume form on the surface $M$, defined in terms of its functional expectation

\begin{equation}\label{deffinitemap2}
\E^{a}[  F( \nu_{a} )  ]= \frac{1}{Z_{a}} \sum_{N \geq 1}  e^{-\bar{\mu}  N} \sum_{T \in \mathcal{T}_{N,{\bf g}}} F(e^{\varphi_T} \,d{\rm v}_{g_{\tau_T}} ) Z(T),
\end{equation} 
 for positive bounded functions $F$ where $Z_a$ is a normalization constant ensuring that $\E^{a} [ \cdot]$ is the expectation of a probability measure. We denote by $\P^{a}$ the probability law associated to $\E^{a}$.

\begin{conjecture}\label{conjcartes2}
Assume $\mathbf{g}=2$. Under $\P^{a}$, the random measure  $\nu_{a}$ converges in law as $a \to 0$ with $\bar{\mu}$ given by \eqref{Asympbarbos} in the space of Radon measures equipped with the topology of weak convergence towards  the Liouville measure $\mc{L}_\gamma$ under  $\E_{  (g_\tau)_{\tau} , \mu}[  \cdot ]$  with parameter $\gamma=2$.
\end{conjecture}

The reader can find much more material on $2d$-string theory in the review \cite{Kleb}   or the lecture notes \cite{polchinski}.


\hspace{10 cm}


\begin{thebibliography}{20}



\bibitem[ARS]{ARS} P. Albin, F. Rochon, D. Sher, \emph{Resolvent, heat kernel and torsion under degeneration to fibered cusps}, Mem. Amer. Math. Soc. to appear. 

\bibitem[ARS2]{ARS2}
P. Albin, F. Rochon, and D. Sher: \emph{Analytic torsion and R-torsion of Witt representations on manifolds with cusps}, Duke Math. J. 167 (2018), no. 10, 1883-1950.

\bibitem[Al]{Al} O. Alvarez, \emph{Theory of strings with boundaries: Fluctuations, topology and quantum geometry}, Nuclear Physics B 216 (1983), 125--184.

\bibitem[ABB]{ambjorn}
J. Ambjorn, J. Barkley, T. Budd, \emph{Roaming moduli space using dynamical triangulations}, 
Nucl. Phys. B 858 (2012), 267--292.
 
\bibitem[ADJ]{Amb}
J. Ambjorn, B. Durhuus, T. Jonsson, Quantum Geometry: a statistical field theory approach, Cambridge Monographs on Mathematical Physics, 2005.

\bibitem[BiFe]{bilal}
A. Bilal, F. Ferrari F.: \emph{Multi-Loop Zeta function regularization and spectral cutoff in curved spacetime},   \href{https://arxiv.org/abs/1307.1689}{arXiv:1307.1689}.

\bibitem[BeKn]{BeKn} A.A. Belavin, V.G. Knizhnik, \emph{Algebraic geometry and the geometry of quantum strings.} Phys. Lett. B168 (1986), no 201.

\bibitem[Be]{berestycki} 
N. Berestycki, \emph{An elementary approach of Gaussian multiplicative chaos},  Electron. Commun. Probab.  vol 22 (2017), no. 27, 12 pp.

 \bibitem[Bu1]{Bu1} M. Burger, \emph{Asymptotics of small eigenvalues of Riemann surfaces}. Bull. AMS. \textbf{18} (1988), no 1, 39-40.
 
\bibitem[Bu2]{Bu2} M. Burger, \emph{Small eigenvalues of Riemann surfaces and graphs}, 
Math. Zeit. \textbf{205} (1990), 395--420.
  
 


\bibitem[BeCa]{bender}
E.A. Bender, E.R. Canfield, \emph{The asymptotic number of rooted maps on a surface}, 
J. Combin. Theory Ser. A \textbf{43} (1986) no. 2, 244--257.


\bibitem[Bo]{Bo} D. Borthwick, Spectral theory of infinite-area hyperbolic surfaces. 
Progress in Mathematics, 256. Birkh\"auser Boston, Inc., Boston, MA, 2007. xii+355 pp.



\bibitem[Cu]{Curien}
N. Curien, \emph{A glimpse of the conformal structure of random planar maps}, 
Commun. Math. Phys. {\bf 333} (2015) no. 3, 1417-1463. 

\bibitem[Da]{cf:Da} F. David, \emph{ Conformal Field Theories Coupled to 2-D Gravity in the Conformal Gauge}, Mod. Phys. Lett. A \textbf{3} (1988) 1651--1656.

\bibitem[DKRV]{DKRV}
F. David, A. Kupiainen, R. Rhodes, V. Vargas, \emph{Liouville Quantum Gravity on the Riemann sphere}, 
\emph{Commun. Math. Phys.}, March 2016, Volume 342, Issue 3, pp 869--907. 


\bibitem[DRV]{DRV}
F. David, R. Rhodes, V. Vargas, \emph{Liouville Quantum Gravity on complex tori},
\emph{Journal of Mathematical physics} {\bf 57}, 022302 (2016). 

\bibitem[DhKu]{DhKu} E. D'Hoker, P.S. Kurzepa, \emph{2-D quantum gravity and Liouville Theory}, Modern Physics Letter A, vol 5 (1990), no 18, 1411--1421.

\bibitem[DhPh]{DhPh}  E. D'Hoker, D.H. Phong, \emph{Multiloop amplitudes for the bosonic Polyakov string.} 
Nuclear Phys. B \textbf{269} (1986), no. 1, 205--234.

\bibitem[DhPh2]{DhPh2}  E. D'Hoker, D.H. Phong, \emph{On determinants of Laplacians on Riemann surfaces},  Commun. Math. Phys. 104 (1986) no 4, 537-545.

\bibitem[DhPh3]{DhPh3}  E. D'Hoker, D.H. Phong, \emph{The geometry of string perturbation theory},  Rev. Mod. Phys. 60 (1988) 917-1065. 

\bibitem[dFMS]{difrancesco}
P. Di Francesco, P. Mathieu, D. Senechal, Conformal Field Theory, Graduate Texts in Contemporary Physics 1997, Springer.

\bibitem[DiKa]{DistKa} J. Distler,  H. Kawai, \emph{Conformal Field Theory and 2-D Quantum Gravity or Who's Afraid of Joseph Liouville?}, Nucl. Phys. \textbf{B321} (1989) 509--517.

\bibitem[Du1]{dubedat}
J. Dub\'edat, \emph{SLE and the Free Field: partition functions and couplings}, Journal of the AMS, 
\textbf{22} (2009) (4), 995--1054. 



\bibitem[DMS]{DMS}
B. Duplantier, J. Miller, S. Sheffield, \emph{Liouville quantum gravity as a mating of trees},  	\href{http://arxiv.org/abs/1409.7055}{arXiv:1409.7055}. 


\bibitem[DRSV1]{Rnew7}
B. Duplantier, R. Rhodes, S. Sheffield, V. Vargas, \emph{Critical Gaussian multiplicative chaos: convergence of the derivative martingale}, Annals of Probability \textbf{42} (2014), no5, 1769--1808. 


\bibitem[DRSV2]{Rnew12}
Duplantier  B., Rhodes R., Sheffield S., Vargas V.: Renormalization of Critical Gaussian Multiplicative Chaos and KPZ formula, \emph{Commun. Math. Phys.} \textbf{330} (2014), no 1,  283--330.


\bibitem[DuSh]{cf:DuSh} Duplantier, B., Sheffield, S.: Liouville Quantum Gravity and KPZ, \emph{Inventiones Mathematicae} \textbf{185} (2) (2011) 333-393, 
%


\bibitem[FKZ]{FKZ} F. Ferrari, S. Klevtsov and S. Zelditch,
\emph{Random K\"ahler metrics}, Nucl. Phys. B \textbf{869} (2012), no. 1, 89--110





\bibitem[Ga]{gaw}
 K. Gawedzki, \emph{Lectures on conformal field theory. In Quantum fields and strings: A course for mathematicians}, Vols. 1, 2 (Princeton, NJ, 1996/1997), pages 727--805. Amer. Math. Soc., Providence, RI, 1999. 

\bibitem[GuZw]{GuZw} L. Guillop\'e, M. Zworski, \emph{Upper bounds on the number of resonances for
non-compact Riemann surfaces}, J. Funct. Anal. \textbf{129} no.2 (1995), 364--389.

\bibitem[Ji]{Ji} L. Ji, \emph{The asymptotic behavior of Green's functions for degenerating hyperbolic surfaces}, Math. Z. \textbf{212} (1993), 375--394.
 
\bibitem[HRV]{HRV}
Y. Huang, R. Rhodes, V. Vargas, \emph{Liouville quantum field theory in the unit disk},  Ann. Inst. H. Poincar\'e Probab. Statist. \textbf{54} (2018), no 3, 1694--1730.

\bibitem[Ka]{cf:Kah} J-P. Kahane, \emph{Sur le chaos multiplicatif},
Ann. Sci. Math. Qu{\'e}bec, \textbf{9} no.2 (1985), 105--150.
%
\bibitem[Kleb]{Kleb}
I. Klebanov, \emph{String theory in two dimensions}, \href{https://arxiv.org/pdf/hep-th/9108019v2.pdf}{arXiv:hep-th/9108019}. 

\bibitem[KlZe]{KlZe} S. Klevtsov, S. Zelditch, \emph{Heat kernel measures on random surfaces},  Adv. Theor. Math. Phys. \textbf{20} (2016), no. 1, 135--164.

\bibitem[KPZ]{cf:KPZ} V.G. Knizhnik, A.M. Polyakov, A.B. Zamolodchikov, \emph{Fractal structure of 2D-quantum gravity}, Modern Phys. Lett A \textbf{3}(8) (1988), 819--826.

\bibitem[MaMi]{MaMi} N.E. Mavromatos, J.L. Miramontes, \emph{Regularizing the functional  integral in 2D-quantum gravity}, Mod. Phys. Lett. A \textbf{04} (1989), no 19, 1847--1853.

\bibitem[MeZh]{MeZh} R.B. Melrose, X. Zhu, \emph{Boundary behaviour of Weil-Petersson and fiber metrics for Riemann moduli spaces}, International Mathematics Research Notices,  rnx264 (2017). https://doi.org/10.1093/imrn/rnx264

\bibitem[Mi]{miermont}
G. Miermont, \emph{Tessellations of random maps of arbitrary genus}, Annales de l'ENS 
\textbf{42} (2009), fascicule 5, 725-781. 


\bibitem[OPS]{OPS}  B. Osgood, R. Phillips, P. Sarnak, \emph{Extremals of determinants of Laplacians.}
J. Funct. Anal. \textbf{80} (1988), no. 1, 148--211.

\bibitem[Pa]{Pa} S.J. Patterson, \emph{The Selberg zeta function of a Kleinian group}, in Number Theory, Trace Formulas and Discrete Groups, (Academic Press, Boston 1989), 409--441.

\bibitem[Pol]{polchinski}
J. Polchinski: {What is string theory?}, Lectures presented at the 1994 Les Houches Summer School ``Fluctuating Geometries in Statistical Mechanics and Field Theory.''  \href{http://xxx.lanl.gov/abs/hep-th/9411028}{arXiv:hep-th/9411028}.

\bibitem[Po]{Pol}
A.M. Polyakov A.M., \emph{Quantum geometry of bosonic strings}, Phys. Lett. 
\textbf{103B} 207 (1981).
 
\bibitem[RaSi]{RaSi}  D.B. Ray, I.M. Singer, \emph{R-torsion and the Laplacian on Riemannian manifolds}, Advances in Math. \textbf{7} (1971) no 2, 145--210.

\bibitem[RhVa1]{review} R. Rhodes, V. Vargas, \emph{Gaussian multiplicative chaos and applications: a review}, Probab. Surveys \textbf{11} (2014), 315-392.
%
\bibitem[RhVa2]{Houches}
R. Rhodes, V. Vargas, \emph{Lecture notes on Gaussian multiplicative chaos and Liouville Quantum Gravity}, 
to appear in Les Houches summer school proceedings, \href{https://arxiv.org/abs/1602.07323}{arXiv:1602.07323}.  

\bibitem[RoVa]{RoVa} 
R. Robert, V. Vargas, \emph{Gaussian multiplicative chaos revisited}, Ann. Probab. {\bf 38} 
605-631 (2010). 

\bibitem[Sa]{Sa} P. Sarnak, \emph{Determinants of Laplacians.} Comm. Math. Phys. \textbf{110} (1987), no. 1, 113--120.

\bibitem[Sc]{Sc} M. Schulze, \emph{On the resolvent of the Laplacian on functions for degenerating surfaces of finite geometry}, Journ. Funct. Anal. \textbf{236} (2006), no 1, 120--160.


%
\bibitem[Se]{seiberg}
Seiberg N.: Notes on Quantum Liouville Theory and Quantum  Gravity, Progress of Theoretical Physics, suppl. 102, 1990.

\bibitem[SWY]{SWY} R. Schoen, S. Wolpert, S.T. Yau, \emph{Geometric bounds on low eigenvalues of a compact Riemann surface}, Geometry of the Laplace operator, Proc. Sympos. Pure Math, Vol 36, AMS, Providence R.I. 1980, 279--285.

\bibitem[Sha]{shamov} A. Shamov, \emph{On Gaussian multiplicative chaos},  Jour. Funct. Anal. \textbf{270} no 9, 3224--3261.

\bibitem[She]{She07}
S. Sheffield, \emph{Gaussian free fields for mathematicians}, \emph{Probab. Th. Rel. Fields}, \textbf{139} (2007)  521--541.

\bibitem[TaTe]{TT} L. Takhtajan, L-P. Teo, \emph{Quantum Liouville Theory in the Background Field Formalism I. Compact Riemann Surfaces}, Comm. Math. Phys. {\bf 268} (2006) no. 1, 135--197.  

\bibitem[Tr]{Tr} 
A.~J.~Tromba, 
\emph{Teichm\"uller theory in Riemannian Geometry}, 
Lectures in Math., ETH Z\"urich, Birkh\"auser Verlag, Basel (1992).



\bibitem[Wo1]{Wo1} S. Wolpert, \emph{On the Weil-Petersson geometry of the moduli space of curves}. 
Amer. J. Math. \textbf{107} (1985), no. 4, 969--997.

\bibitem[Wo2]{Wo2} S. Wolpert, \emph{Asymptotics of the spectrum and the Selberg zeta function on the space of Riemann surfaces.} Comm. Math. Phys. \textbf{112} (1987), no. 2, 283-315.

\bibitem[Wo3]{Wo3} S. Wolpert, \emph{The hyperbolic metric and the geometry of the universal curve}, Journ. Diff. Geom. \textbf{31} (1990), 417--472.
\end{thebibliography}
 \end{document}